\documentclass{cshonours}
\usepackage[utf8]{inputenc}
\usepackage{amsmath}
\usepackage{amsthm}
\usepackage{amsfonts}
\usepackage{amssymb}
\usepackage{graphicx}
\usepackage{geometry}
\usepackage{tikz}
\usepackage{algorithm}
\usepackage{algpseudocode}
\usepackage{mathtools}
\usepackage{dsfont}
\usepackage{subfiles}
\usepackage{longtable}
\usepackage{booktabs}
\usepackage{physics}
\usepackage[style=numeric,natbib=true, backend=bibtex]{biblatex}
\usepackage[utf8]{inputenc}
\usetikzlibrary{arrows}
\usepackage{pgf,tikz,pgfplots}
\pgfplotsset{compat=1.13}
\numberwithin{equation}{section}
\theoremstyle{definition}
\newtheorem{definition}{Definition}[section]
\newtheorem{theorem}{Theorem}
\algnewcommand\algorithmicforeach{\textbf{for each}}
\algdef{S}[FOR]{ForEach}[1]{\algorithmicforeach\ #1\ \algorithmicdo}
\title{Dissertation: Quantum Approximate Optimisation Applied to Graph Similarity}

\author{Nicholas Pritchard 21726929\\\textit{Supervisors:} Prof. Jingbo Wang and Prof. Amitava Datta}
\begin{document}
\definecolor{rvwvcq}{rgb}{0.08235294117647059,0.396078431372549,0.7529411764705882}
\maketitle
\section*{Abstract}
	Quantum computing promises solutions to classically difficult and new-found problems through controlling the subtleties of quantum computing. The Quantum Approximate Optimisation Algorithm (QAOA) is a recently proposed quantum algorithm designed to tackle difficult combinatorial optimisation problems utilising both quantum and classical computation\cite{farhi_quantum_2014-1}. 
	The hybrid nature, generality and typically low gate-depth make it a strong candidate for near-term implementation in quantum computing. Finding the practical limits of the algorithm is currently an open problem. Until now, no tools to facilitate the design and validation of probabilistic quantum optimisation algorithms such as the QAOA on a non-trivial scale exist.\newline
	Graph similarity is a long standing classically difficult problem withstanding decades of research from academia and industry. Determining the maximal edge overlap between all possible node label permutations is an NP-Complete task which has faced little research from classical computer science and provides an apt measure of graph similarity. We introduce a novel quantum optimisation simulation package facilitating investigation of all constituent components of the QAOA from desktop to cluster scale using graph similarity as an example.\newline
	Our simulation provides class-leading flexibility and performance. We investigate eight classical optimisation methods each at six levels of decomposition; the most exhaustive study to date. Moreover a novel encoding for permutation based problems such as graph similarity through edge overlap to the QAOA allows for significant quantum memory savings at the cost of additional operations. This compromise extends into the classical portion of the algorithm as the inclusion of infeasible solutions creates a particularly difficult cost-function landscape.\newline
	We present performance analysis of our simulation and of the quantum algorithm itself setting a precedent for investigating and validating numerous other difficult problems to the QAOA as we move towards realising practical quantum computation.
\newpage
\section*{Acknowledgements}
Firstly, I would like to thank my supervisors, Prof. Jingbo Wang and Prof. Amitava Datta. Not only was their support invaluable but the faith they placed in me to learn in an area of study completely new to me was inspiring. Thank you for giving me the chance to learn some Physics.\newline
I would also like to thank Matthew, Mitchell, Kooper, Gareth, Sam, Amit, Leigh, Edric, Ja-Jet and Lyle from the Quantum Dynamics Research Group who made the office arguably the most vibrant destination on campus and for their feedback, assistance and encouragement of my work.\newline
Finally, I thank my parents for their unending encouragement and support throughout the year. 
\tableofcontents
%\listoftables  %optional
%\listoffigures  %optional

\chapter{Introduction}
Seeking use from digital computers is an invariable goal in computer science. The computers we currently know stem from mathematics recruiting the laws of physics to realise our definition of computing. Graphs are widely used to represent real-world phenomenon in manner suitable for computation. Determining whether one graph is identical or similar to another when node correspondence is unknown is a long standing difficulty. Approximation of such measures prove useful in a myriad of real-world contexts. Quantum computing defines computation as a physical process linking computation to the underpinning physics stronger than previously encountered. Historically intractable problems can be explored in new ways as increasingly sophisticated quantum algorithms are formulated. Combinatorial optimisation is amongst the most general and practically applicable computational paradigms. Based on the formulation of quantum annealing the recently formulated Quantum Approximate Optimisation Algorithm (QAOA) \autocite{farhi_quantum_2014-1} approaches NP-Complete combinatorial optimisation problems on discrete gate-based quantum computers. We explore the long-standing difficult problem of graph similarity via the QAOA using a problem encoding scheme novel to this algorithm. Furthermore, empirical validation is the only powerful method known to evaluate heuristic based or very probabilistic algorithms such as the QAOA. We present Qolab (Quantum Optimisation Laboratory), a software package designed to provide scalable efficient simulation of the QAOA from desktop to cluster scale for a generalised problem encodings and variations on the QAOA itself. A familiarity with the fundamentals of quantum computing is essential to understanding quantum algorithms. This familiarity is difficult to acquire without a background in quantum-physics. To aid with this we additional present a brief introduction to the field providing the minimal required knowledge to understand the QAOA. As physical implementations of gate-based quantum computers grow in scale, focus is shifting towards practical quantum computation with tangible real-world applications.
\section{Computer Science Thus Far}
Arguably, computing begins with the abacus. Over the past millennium humanity has discovered increasingly sophisticated methods to implement the model of computing. Such developments are governed by physics. The most radical discoveries in the field invariably drive innovation in the era following; Galileo formulated simple machines establishing a relationship between mathematics and theoretical and experimental physics. Newton's laws of motion gave birth to the machine age, similarly the discoveries of electromagnetism drove the development of the information age. We are yet to realise a fitting use for quantum mechanics. Quantum computing is the strongest candidate thus far and has garnered ferocious support from private industry and public institutions across the globe. Quantum computing represents a natural development that radically diverges from classical computer science.
\section{Combinatorial Optimisation}
A combinatorial optimisation problem consists of finding some optimal object from a finite set of possible choices. Specifically we define combinatorial optimisation problems with respect to the following components.
\begin{definition}
	A combinatorial optimisation is specified precisely by the following components
	\begin{itemize}
		\item A specific problem type $I$. We must be able to efficiently determine if an arbitrary problem belongs to the set we are concerned with.
		\item For each valid problem instance $p \in I$ a solution validation function $S : p \rightarrow \mathcal{P}(\mathcal{U})$ which determines if a given input $x$ is a feasible solution to $p$. The computational resources required to store both $x$ and $p$ must be bounded by some polynomial. This is required to efficiently verify an arbitrary solution $y$ as a valid solution to $p$.
		\item An objective function $C(x) : I \times \mathcal{U} \rightarrow \mathbb{Z}$ which maps a feasible solution to a non-negative integer value indicating the quality of the solution. The highest value of $C(x)$ in the set of feasible solutions indicates the optimal solution for a given problem instance
	\end{itemize}
\end{definition}
Given a problem instance $p$ of type $I$ we aim to find an $x$ such that
\begin{equation}
C(x) = max(S(p))
\end{equation}
This leads to a clean definition of \textit{NP-optimisation} problems.
\begin{definition}
	An NP optimisation problem $NP$ is a combinatorial optimisation problem where the set of feasible solutions $S(x,NP)$ cannot be exhausted in polynomial time. 
\end{definition}
We further define a \textit{bounded optimisation problem NPO-PB}
\begin{definition}
	A bounded NP optimisation problem $NPO-PB$ is an NP optimisation problem with the additional constraint that the resources required to represent $C(x) \forall x \in S(NPO-PB)$ must be bounded by some polynomial.
\end{definition}
Given the intractability of $NP$ and $NPO-PB$ optimisation problems approximation algorithms are tolerated relaxing the problem to find an optimal approximate solution. This general problem description captures a number of useful problems in computer science such as planning, scheduling, protein folding, vertex colouring and the travelling salesman problem \autocite{garey_computers_1979}. Finding approximate solutions of superior accuracy for such problems is an open area of research in classical computer science.\newline 
\section{Graph Similarity}
Graphs are well generalised mathematical structure. A graph encodes relations between entities and as such graphs can represent a vast number of real-world problems. The features of faces, topology of the Internet, road-networks, decision-flow, computer programs, cosmological bodies and any other number of natural and unnatural phenomenon can be expressed through this marvellous data-structure. 
For completeness we define the graph.
\begin{definition}
	A graph $G(v,e)$ is a collections of vertices $v$ and a collection of $(v,v')$ pairs termed edges $e$ which may be directed or un-directed \autocite{diestel_graph_2000}. 
\end{definition}
The rich history of graph theory provides many metrics on graphs such as shortest path-length. While many of these measures are trivial or at least tractable to compute they typically concern themselves with the internal structure of a particular graph. Finding information about the overall structure of a graph is a much more laborious task. Graph isomorphism is a quintessential structural problem when comparing graphs; are two given graphs alike? This problem has no P algorithm and showing NP-completeness is a long-standing open problem \autocite{garey_computers_1979} and so it occupies a unique complexity class often termed graph isomorphism complete \autocite{weisstein_graph_2018}. Difficulty arises from the factorial number of mappings between the features of two unlabelled graphs with unknown node correspondence. Relaxing the isomorphism problem provides numerous measures of similarity between graphs which are known to be NP-Complete \autocite{garey_computers_1979} and have been given a large amount of academic and industrial attention.\newline
We specifically investigate a measure of \textit{whole graph similarity} which we defined as
\begin{definition}{Whole Graph Similarity:}
	Given two graphs $G_1(v_1, e_1)$ and $G_2(v_2, e_2)$ with possibly different numbers of vertices and edges, find an algorithm which returns a measure of similarity $S | S \in [0,1]$. Furthermore:
	\begin{enumerate}
		\item $S(G_1, G_1) = 1$
		\item $S(G_1, G_2) = S(G_2, G_1)$
	\end{enumerate}
	\label{ref:similarity} 
\end{definition}
This method is simple to understand and allows for concise error determinations when comparing a brute-force optimal and computed approximate solution. Further, we define and present the measure of edge overlap (EO) as a whole graph similarity measure adhering to Definition \ref{ref:similarity}
Consider two directed, un-weighted graphs $G_1$ and $G_2$ with no known vertex labelling. For each permutation $\sigma$ of potential node labels between $G_1$ and $G_2$ we provide a penalty of one for every edge (or non-edge) which differs between the two graphs. This penalty score is then normalised by the maximum number of edges possible. The maximum possible score is bound by the number of edges possible which is $v^2$ for a directed graph, $v(v-1)$ for an un-directed graph with self-edges and $\frac{v(v-1)}{2}$ for an un-directed graph without self-edges. Normalisation yields an edge-difference value $e_d$ which we take as our similarity value. Importantly, graphs of differing numbers of vertex can be compared by adding degenerate, unconnected nodes to the smaller graph.\newline The computational difficulty of finding this value arises from the factorial number of possible node-labelling permutations.\newline
For example, consider the two graphs depicted in figure \ref{fig:graphSim}. The maximum number of identical edges is $14$ counting the identical edges present and identical edges which are missing (they only differ in two edges). Thus the graph similarity for these two is $1 - \frac{2}{4^2} = 1 - \frac{1}{8} = 0.875$.\newline
%TODO: Better diagram
\begin{figure}[h!tbp]
	\begin{center}
		\resizebox{!}{1in}{
			\begin{tikzpicture}[line cap=round,line join=round,>=triangle 45,x=1cm,y=1cm]
			\begin{axis}[
			x=1cm,y=1cm,
			axis lines=none,
			xmin=-6,
			xmax=8,
			ymin=-1.6,
			ymax=2.25,
			xtick=\empty,
			ytick=\empty,]
			\clip(-4.5,-4.59183983444371) rectangle (4.5,5);
			\draw [->,line width=2pt] (-3,2)-- (0,2);
			\draw [->,line width=2pt] (-3,2)-- (-3,-1);
			\draw [->,line width=2pt] (-3,-1)-- (0,2);
			\draw [->,line width=2pt] (0,2)-- (0,-1);
			\draw [->,line width=2pt] (1,2)-- (4,2);
			\draw [->,line width=2pt] (4,2)-- (4,-1);
			\draw [->,line width=2pt] (4,-1)-- (1,-1);
			\draw [->,line width=2pt] (1,2)-- (1,-1);
			\begin{scriptsize}
			\draw [fill=rvwvcq] (-3,2) circle (2.5pt);
			\draw [fill=rvwvcq] (0,2) circle (2.5pt);
			\draw [fill=rvwvcq] (-3,-1) circle (2.5pt);
			\draw [fill=rvwvcq] (0,-1) circle (2.5pt);
			\draw [fill=rvwvcq] (1,2) circle (2.5pt);
			\draw [fill=rvwvcq] (4,2) circle (2.5pt);
			\draw [fill=rvwvcq] (1,-1) circle (2.5pt);
			\draw [fill=rvwvcq] (4,-1) circle (2.5pt);
			\end{scriptsize}
			\end{axis}
			\end{tikzpicture}
		}
		\caption{Two example unlabelled graphs ($G_1$ and $G_2$)}
		\label{fig:graphSim}
	\end{center}
\end{figure}
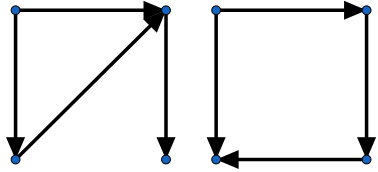
\subsection{Graph Similarity as Combinatorial Optimisation}
We map graph-similarity to a bounded NP-optimisation problem as follows.
\begin{itemize}
	\item The set of problem instances $I$ is the set of all un-weighted graphs.
	\item Given a pair of graphs $G_1, G_2$ the set of feasible solutions is the set of mappings between the vertices of $G_2$ to $G_1$. The number of candidate mappings grows $\mathcal{O}(V^2)$ hence growing non-polynomially. A solution can be validated by checking that each vertex is mapped and that the mapping is bijective.
	\item The objective function is the edge-overlap between graphs $G_1, G_2$ under a candidate mapping which is bound by $V^2$. 
\end{itemize}
\begin{algorithm}\label{alg:gSim}\
	\caption{Classical brute-force algorithm to find Maximal Edge Overlap}
	\begin{algorithmic}[1]
		\State $G_1 \gets <E_1>$
		\State $G_2 \gets <E_2>$
		\State Best$ \gets 0$
		\ForEach{Permutation $X$ of $1,2...V$}
			\State $C(x) \gets V^2$
			\For{$i \gets 0$ to $V$}
				\For{$j \gets 0$ to $V$}
					\If{$G_1[i,j] != G_2[X[i],X[j]]$}
						$C(x) \gets C(x) - 1$
					\EndIf	
				\EndFor
			\EndFor
			\If{$C(x) > $Best}
				\State Best $\gets C(x)$
			\EndIf
		\EndFor
		\State \Return Best
	\end{algorithmic}
\end{algorithm}
Many classically difficult problems can be simple to define but remain difficult due to the explosive growth of the solution space. In the case of graph-similarity the number of candidate solutions grow $\mathcal{O}(v!)$.\newline
This exact formulation of graph-similarity demands a brute-force algorithm to find an exact solution presented in Algorithm \ref{alg:gSim} as it is feasible for only a single candidate to be optimal (consider two isomorphic graphs for example).
\section{The Quantum Approximate Optimisation Algorithm}\label{sec:QAOA}
The Quantum Approximate Optimisation Algorithm (QAOA) is a hybrid quantum-classical algorithm designed to approximately solve NP-Complete combinatorial optimisation problems \autocite{farhi_quantum_2014}. The QAOA is unique among other notable gate-based quantum computing algorithms for two. The gate depth required for useful computation is low and the approximate nature of the algorithm make it suitable for near-term, noisy quantum computers. Furthermore, the general formulation of the algorithm inspires other similar algorithms which allow for more sophisticated problem encoding. The QAOA and derivative algorithms are likely candidates for near-term practical use. Finding the practical limits of this algorithm remains an open problem.
\subsection{The Quantum Adiabatic Algorithm (QAA)}
Farhi et al. \autocite{farhi_quantum_2001} note the possibility of exploiting the exponential number of items a qubit register can represent and formulate the quantum adiabatic algorithm (QAA). The QAA is founded in by the \textit{quantum adiabatic theorem} where a quantum systems evolves according to the Schr\"{o}idgner equation
\begin{equation}
i\frac{d}{dt}\ket{\psi(t)} = \hat{H}(t)\ket{\psi(t)}
\end{equation}
The QAA utilises a quantum register of $n$ qubits. Two system Hamiltonians are prepared $\hat{B}$ and $\hat{C}$. $\hat{B}$ often termed the driver Hamiltonian is an easily prepared maximal energy state. $\hat{C}$ often termed the problem Hamiltonian encodes the values of all possible solutions.  Importantly, $\hat{C}$ encodes the problem solution as the highest energy-state of our quantum system. The Hamiltonian governs the evolution path of our system according to
\begin{equation}
\hat{H}(t) = (1-\frac{t}{T})\hat{B} + \frac{t}{T}(\hat{C}))
\end{equation}
such that $\hat{H}(0) = \hat{B}$ and $\hat{H}(T) = \hat{C}$. Adiabatic evolution ensures that if $\hat{B}$ begins in a maximum-energy state, as $t \rightarrow T$ the system will remain in a maximal energy state. Measurement at time $T$ should yield a near-optimal solution with high-probability \autocite{farhi_quantum_2001}.
\subsection{Moving to the Quantum Approximate Optimisation Algorithm}
The Quantum Approximate Optimisation Algorithm (QAOA) \autocite{farhi_quantum_2014-1} builds on the foundation of the QAA by noting that adiabatic evolution is impossible to implement exactly. A Suzuki-Trotter decomposition of the evolution into discrete increments is much simpler however and can be implemented in a gate-based quantum computer \autocite{wecker_training_2016}.\newline
Farhi et al. define combinatorial optimisation problems with regards to \textit{maximum satisfiability}. The corresponding cost function is defined as
\begin{equation}\label{eq:C}
c(x) = \sum_{i=1}^{m}c_i(x)
\end{equation}
$c_i(x)$ check if the $i$-th clause is satisfied by the input bit-string. Since this problem is NP-Complete, any other NP-complete problem can be mapped to this formulation in polynomial time.
The solution can be encoded as a diagonal Hamiltonian $\hat{C}$ where the $i$-th eigenvalue contains the cost-function value of $i$ as a $n$-length bit-string. This diagonal operator is defined by the action of $\hat{C}$ on the computational basis states, just as any other quantum-computing gate. 
\begin{equation}\label{eq:UC}
\hat{U}_C(\gamma) = e^{-i\gamma\hat{C}}
\end{equation} 
where $\gamma$ is a real-valued parameter which is restricted to the interval $[0, 2\pi)$. Importantly, the QAOA operates on \emph{all} bit-strings of length $n$ without regard for feasibility. This is not an issue for problems like MAX-SAT where all bit-strings are valid but cause issue for problem with more nuanced encodings.\newline
An operator $\hat{B}$ defines how candidate bit-strings are considered. The canonical formulation is given by
\begin{equation}\label{eq:B}
\hat{B} = \sum_{i=1}^{n}\sigma_i^x
\end{equation} 
Where $\sigma_x^i$ is the Pauli-X operator, the quantum equivalent of the NOT gate. This Hamiltonian allows all candidate bit-strings to be considered without restriction. Similar to $\hat{U}_C$ we define $\hat{U}_B$ as 
\begin{equation}\label{eq:UB}
\hat{U}_B(\beta) = e^{-i\beta\hat{B}}
\end{equation}
The resulting quantum state for a given set of parameters is by
\begin{equation}\label{eq:state}
\ket{\vec{\gamma},\vec{\beta}} = \hat{U}_B(\beta_p)\hat{U}_C(\gamma_p)...\hat{U}_B(\beta_1)\hat{U}_C(\gamma_1)\ket{\psi}
\end{equation}
Where $\ket{s}$ is a trivial equal-super-position of $n$ qubits.
In practice repeated sampling of a single QAOA-iteration functions as the output of our system. This output is defined analytically as the expectation value of $\ket{\vec{\gamma},\vec{\beta}}$, $F_p(\vec{\gamma},\vec{\beta})$ and is defined by the equation
\begin{equation}
F_p(\vec{\gamma},\vec{\beta}) = \bra{\vec{\gamma},\vec{\beta}}\hat{C}\ket{\vec{\gamma},\vec{\beta}}
\end{equation} 
For reference the expectation value of a quantum state $\ket{\psi}$ is computed as a dot-product of $\ket{\psi}^\dagger$, $\hat{C}$ and $\ket{\psi}$ and is effectively a weighted sum of the probability to measure a given bit-string multiplied by its corresponding cost-function value. Typically, the amount of decomposition is small and is expressed as the number $p$ ($\approx2$). Finding a solution to the original combinatorial optimisation problem is now accomplished by a parameter search on the $2p$ transformation parameters, which can be performed classically. \newline
The QAOA is therefore a generalised frame-work application to a wide range of combinatorial optimisation problems and a strong candidate for near-term practical use. We present a high-level circuit representation of the QAOA in Figure 1.2.
\begin{center}
	\begin{figure}\label{fig:QAOAcirc}
		\includegraphics[width=\textwidth]{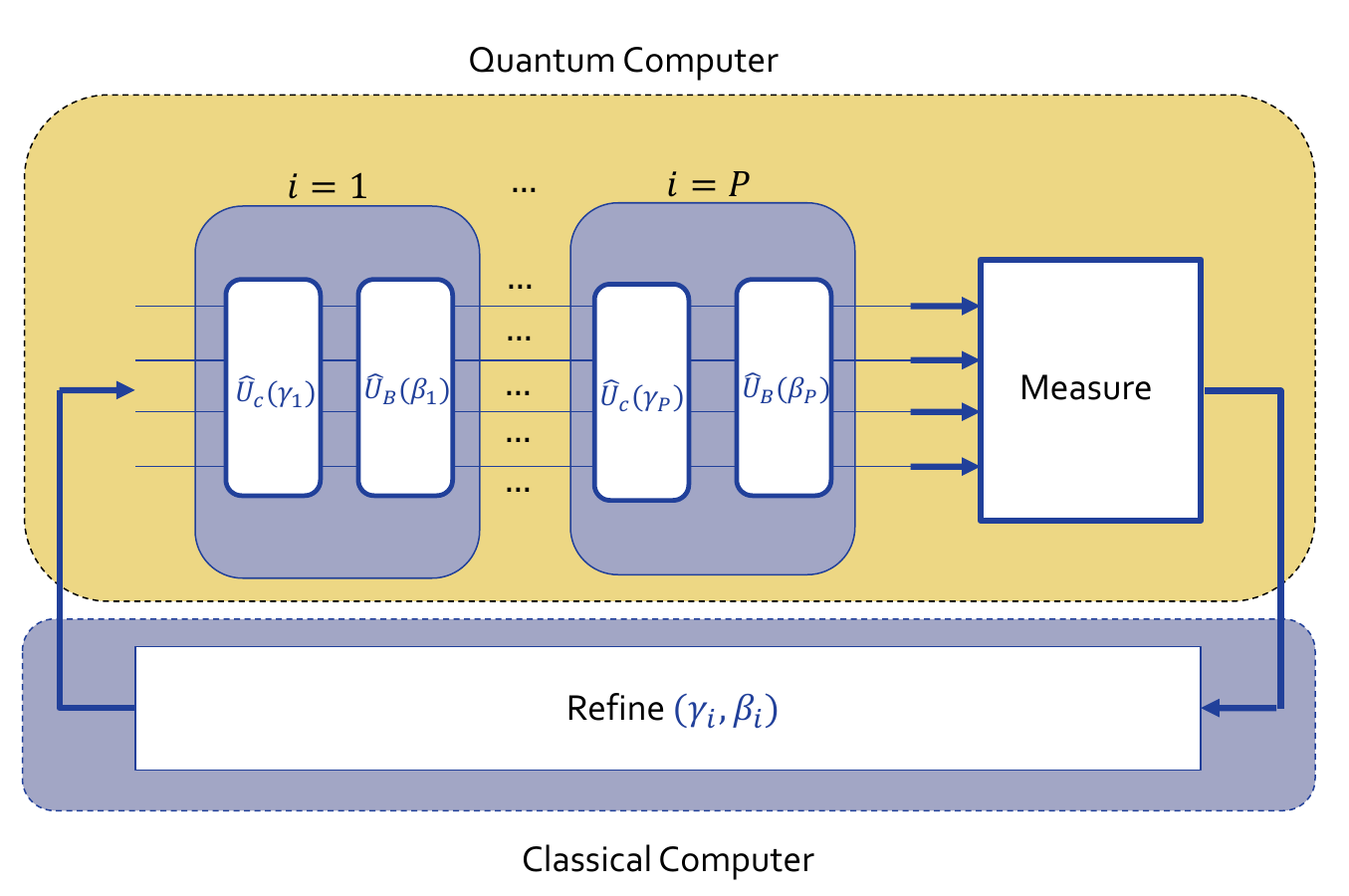}
		\caption{Circuit diagram of the QAOA framework}
	\end{figure}
\end{center}
\subsection{Simulating the QAOA}
Classical simulation of an arbitrary quantum process is a well-established difficulty \autocite{feynman_simulating_1982}. Additionally, the exponential and anti-intuitive behaviour of qubits make designing quantum algorithms a generally difficult task. The current demand for exacting bit-level design for quantum algorithms makes experimentation with general quantum algorithms such as the QAOA exceedingly difficult. Current quantum computing simulations either lose performance by offering generality or lose generality by offering extreme performance for a single task. This motivates the development of a simulation package targeted at the general QAOA framework intending to provide a compromise between the two extremes. By providing an exacting simulation of the QAOA we make the best use of \textit{high-performance computing}(HPC) tools allowing for powerful empirical validation. Experimental analysis is essential to understand the QAOA since closed-form correctness analysis is infeasible for all but the most trivial, well-conditioned cases \cite{hadfield_quantum_2018}.
\section{Unifying Hypothesis}
Given the long-standing combinatorial difficulty in comparing whole-graph structure, the advent of near-term quantum computing hardware and a ferocious appetite for exploring the use of such machines it seems appropriate to investigate quantum combinatorial optimisation with regards to graph similarity.\newline
We propose that for a fixed number of samples the QAOA will generate more correct solutions to na\"{i}ve random-sampling for determining the maximal edge-overlap between two directed, un-weighted graphs. We expect improved performance for undirected graphs and that the number of quantum evaluations required will scale favourably establishing the QAOA as a valid approach to graph similarity.

\chapter{Literature Review}
We present a discussion of literature surrounding graph similarity measures, a rapid introduction to gate-based quantum computing, a survey of current investigation into the QAOA and finally a discussion on computing the matrix exponential; a cornerstone of simulating the QAOA.
\section{Graph Similarity}
\subsection{Applications of Graph Similarity Measures}
The generality of graph-theory generates many important real-world contexts where a measure of similarity between graphs is valuable. Mapping an individual's social network has remained of interest since the original publication of the 'small-world' phenomenon which in itself provides a good model for understanding the relationship between information sources \autocite{james_six_2006}. Computing the structural neighbourhoods of vertices and edges provides methods for indexing the world wide web allowing a level of access to information unparalleled in all of human history \autocite{brin_anatomy_1998}. 
Strug \autocite{strug_automatic_2013} formulates a machine learning method to evaluate the quality of a design through graph similarity between components. In the field of data-mining classifying graphs based on similar features is applicable to many problem settings such as social network mining, drug design and anomaly detection in program-execution \autocite{han_topological_2015}. Facial recognition and object tracking is an obvious application for graph similarity if features can be mapped onto graph structures \autocite{wiskott_face_1997}. Heymans and Singh \autocite{heymans_deriving_2003} frame the determination of evolutionary pathways in terms of graph similarity computed between the metabolic pathways of varying organisms. The problem of chemical compound matching has been a common benchmark problem prompting strong commercial and academic backing for a number of years \autocite{han_topological_2015}, \autocite{raymond_rascal:_2002}, \autocite{hattori_development_2003}. Just as graphs are able to express a vast number of phenomenon the measures constructed to compare them are just as varied, finding structural similarities between graphs is a long-standing intensely valuable endeavour. 
\subsection{Classical Graph Similarity}
The precise methods used to approximate graph similarity are as varied as the measures defined. However there are two general types of similarity measures commonly employed; whole-graph similarity and vertex-wise similarity \autocite{kollias_fast_2014}.
\subsubsection{Vertex-Wise Similarity}\label{subsec:vertex_sim}
Measures of similarity formulated per vertex maintain a multitude of specific formulations. However, they all follow the same basic paradigm; vertices are considered similar if their neighbourhoods are similar. Such measures are far removed from the definition of graph isomorphism and seldom penalise differences between graphs. Lades et al. \autocite{lades_distortion_1993} presents an object recognition scheme extending classical artificial neural-networks to a more dynamic architecture testing their method by formulating facial-recognition as an elastic graph matching problem. Fortunately in the field of machine vision a degree of error is tolerated as the main design criterion for such systems is efficient calculation of average-case approximations \autocite{lades_distortion_1993}. Nevertheless, the search for more accurate and robust systems yields the extremely fruitful field of machine learning and machine vision we see today.\newline
Analysis of the network structure of the Internet provides a large example of graph similarity measures. Kleinberg \autocite{kleinberg_authoritative_1999} proposes the extraction of improving Internet search queries based on the structure of hyper-links between pages. The neighbourhood of vertices are examined to determine high-quality pages to return as results to user queries. The ranking of pages is defined with regards to 'hub' and 'authority' pages. A 'hub' is a page which references many high-quality authorities and an 'authority' pages is one which is referenced by many high-quality hubs \autocite{kleinberg_authoritative_1999}. Clearly, there exists a large overlap in links between any pair of hub and authority pages sharing a similar neighbourhood; in this sense ranking such pages is a graph similarity problem. Kleinberg \autocite{kleinberg_authoritative_1999} further formulates this computation as an eigenvalue decomposition on the quality values of candidate pages. This approach was later extended by Brin and Page \autocite{brin_anatomy_1998} in the design of the Google search engine. Restriction of scale to a local neighbourhood at any given vertex provides tractability. The success of the vertex-wise description of graph similarity is testament to the power of structural information present in graphs driving Google to forefront fields in computer science such as machine learning \autocite{silver_mastering_2016} and quantum computing \autocite{boixo_characterizing_2016}, \autocite{boixo_simulation_2017}, \autocite{kelly_preview_2018}. 
Zager and Verghese \autocite{zager_graph_2008} use vertex-wise similarity to compute whole graph matching. While results are promising and significantly better than random assignment the Hungarian algorithm employed to perform matching produces imperfect results. This preliminary work highlights the difficultly of the more general problem of graph matching.
Later, Kolias et al. \autocite{kollias_fast_2014} extend vertex-similarity calculations to a graph-global problem presenting scale-able parallel algorithms applicable to graphs two orders of magnitude larger than previously feasible on the order of millions of vertices.

\subsubsection{Whole Graph Similarity}
Whole graph similarity aims to compute how similar to given graphs are in their entirety as a relaxation of the graph-isomorphism problem. The maximum common sub-graph problem \autocite{garey_computers_1979} is a very general metric which captures a natural sense of global similarity. The goal is to find the largest collection of matching vertices and edges present in two graphs. Determining the maximum common sub-graph is intuitively applicable to the analysis of chemical compounds. Hattori et al \autocite{hattori_development_2003} show this method to be effective to not only determine compound similarity in of itself but to discover and classify systemic aspects of biology. Atoms make a natural analogy to vertices as bonds between them do as edges. The maximum common sub-graph is a particularly applicable measure to biological and chemical sciences as many larger compounds are built from well known constituent cliques. As such the techniques employed by Hattori et al. \autocite{hattori_development_2003} make heavy use of heuristics to overcome the computational complexity involved with determining common structure. The computation of the maximum common sub-graph is in general NP-hard \autocite{garey_computers_1979} but in the context of chemical compounds this is not entirely accurate \autocite{hattori_development_2003}. The natural limits on atomic bonding limit complexity somewhat; carbon atoms are only permitted to maintain a maximum of four bonds for instance \autocite{hattori_development_2003}. The use of heuristics significantly increases computation efficiency but degrades solution quality appropriately.\newline
Strug \autocite{strug_automatic_2013} constructs a similar measure of maximum sub-graph comparison but in a generalised context of computer aided design. The problem is framed with regards to hierarchical graphs creating a natural set of useful sub-graphs to be compared thus reducing the overall number of graph comparisons required. Strug \autocite{strug_automatic_2013} employs machine learning inspired methods of kernel matching to approximate solutions degrading solution quality in a similar trade-off for computational performance as other methods.\newline
Papadimitriou et al. \autocite{papadimitriou_web_2010} apply maximum sub-graph matching to identify graph-dissimilarity. By tracking snapshots of the Internet described as a graph one can identify anomalies through the differences between them. Three types of anomalies are searched for: missing connected sub-graphs; missing random vertices, and connectivity changes. Five methods to compute whole graph similarity are considered.\newline
Vertex ranking computes the correlation between two sets of vertices given a pre-computed quality score for each and is the least successful method evaluated \autocite{papadimitriou_web_2010}.\newline
Vertex similarity, is related to methods discussed previously in section \ref{subsec:vertex_sim}. This method performs admirably but is not superior to many other methods.\newline
Vertex and edge-overlap computes similarity between two general graphs. Simply stated, this method considers two graphs similar if they share a large number of edges and vertices. This method utilises edit distance as a measure of similarity capturing a natural intuition for similarity \autocite{papadimitriou_web_2010}. A similar edit-distance measure is employed by Zheng et al. \autocite{zheng_efficient_2015} for searching graph databases. Various filters are employed to reduce the computational complexity of chemical database searches by returning graphs with an edit distance below a carefully chosen bound. This method outperforms contemporary solutions by a significant margin \autocite{zheng_efficient_2015}. With regards to anomaly detection however, vertex overlap fails to detect large missing connected sub-graphs. \autocite{papadimitriou_web_2010}.\newline
Sequence similarity considers graphs to be similar if they share a large number of smaller sub-graphs. Papadimitriou et al. \autocite{papadimitriou_web_2010} utilise their own method of shingling which converts a graph to a linear sequence of tokens. This method is poorly suited to detecting anomalies between graphs but is better suited to structures which are inherently linear. The final method considered by Papadimitriou et al. \autocite{papadimitriou_web_2010} computes a signature of each graph and uses the hamming distance between the two as a measure of similarity. Signature based similarity shows the best performance in detecting anomalies.\newline
Yongkoo et al. \autocite{han_topological_2015} classify graphs by sub-graph matching based on selective applications of exact graph isomorphism tests. Computational efficiency is derived from the intelligent choice of features to be tested. Features are chosen by building a topology where frequent sub-graphs are given an identification tag which can then be queried quickly. Results show significant improvements in accuracy and speed in classifying anticancer behaviour over leading implementations \autocite{han_topological_2015}. However, solutions are still no more than $86\%$ accurate.\newline
Computing measures of graph-wide similarity remains to be a difficult problem approximated by many algorithmic methods. The widespread use of graph similarity measures in a broad-spectrum of problem contexts promote consistent attention from academic and industrial institutions. Despite decades of research no tractable exact algorithm has been found and all known solutions make significant sacrifice to solution quality or scope. This long accepted difficulty makes graph similarity a prime candidate for investigation by quantum computing. 

\section{Quantum Computing}
Quantum mechanics stands among the most notable scientific achievements of the 20\textsuperscript{th} century describing phenomena unlike anything previously encountered. Reminiscent of how electromagnetism spawned the realisation of the digital computer, practical applications of quantum physics promise exotic new technologies. Quantum computing is a major frontier in this regard. We present an introduction to fundamental concepts and historic developments in the field suitable for readers without a strong background in quantum physics.
\subsection{Bits vs. Qubits}
The fundamental difference between classical and quantum computing lies in the physical paradigm used to represent atomic data. In classical computing all data is represented in a register of bits each maintaining the value of zero or one. The exact construction of such a register varies, the fundamental principle however, remains constant. Feynman \autocite{feynman_simulating_1982} makes the elegant observation that the simulation of quantum phenomenon and quantum systems is extremely difficult to perform classically. The natural extension to this observation is whether we can use these inherently complicated quantum phenomenon to perform useful computation. Deutsch \autocite{deutsch_quantum_nodate} later defines the quantum Turing machine by refining the fundamental Church-Turing hypothesis from an abstract construction to the physically related Church-Turing principle. Re-framing the definition of computation with explicit regard to the laws of nature results in a whole new paradigm of computing machines which base their principles on quantum phenomenon. This discovery has driven a quest extending over three decades to construct such machines and to realise the upper echelons of computation permitted by the physical world.\newline
The fundamental building block of quantum computing is the logical abstraction of any two-state quantum system known as a qubit \autocite{nielsen_quantum_2000}. Mathematically we use \textit{dirac} \autocite{dirac_new_1939} notation to represent the state of a qubit; as an example analogous to classical bits a qubit can exist in the $\ket{0}$ or $\ket{1}$ state. One drastic difference is that a qubit can exist in any linear combination of these states called \textit{superposition}
\begin{equation}\label{eq:qubitSuperposition}
\ket{\psi} = \alpha\ket{0} + \beta\ket{1},
\end{equation}
where $\alpha$ and $\beta$ are any complex numbers such that $|\alpha|^2+|\beta|^2=1$. The states $\ket{0}$ and $\ket{1}$ specify a \textit{computational basis} but we could just as easily specify any other pair of basis states as long as they are orthonormal.\newline
To quickly grasp the potential for computational power qubits afford first consider the unit sphere. In classical computing, each bit is permitted to only occupy one of the poles, each corresponding to either a zero or one. In contrast a qubit can occupy any combination of the two states $\ket{0}$ and $\ket{1}$ and can thus occupy any point on surface of the sphere. This representation is known as the \textit{bloch sphere} and is a very useful tool to represent the state of a qubit \autocite{nielsen_quantum_2000}. One might believe that this allows a single qubit to store infinite information as there are infinitely many unique points on a sphere however this is not exactly the case. This superposition of basis states is only present during a quantum computation, when a measurement is made a qubit will collapse onto one of the two basis states. More specifically, a measurement will collapse a qubit into either basis state with probabilities $|\alpha|^2$ or $|\beta|^2$ respectively. This fact captures the core difficulty in the design of quantum over classical algorithms; we can only measure a qubit once. Classical computing is built upon setting and examining the state of bits explicitly and freely. The techniques demanded in quantum computing are more nuanced.\newline
While a measured qubit will only yield a single bit of information, nature is excellent at keeping track of all the quantum information stored by a qubit in superposition. The goal of quantum computing then is to extract as much of this information as possible. 
\subsection{Logic Gates and Quantum Circuits}
A qubit register is simply a collection of multiple qubits just as a classical register is simply a collection of bits. If we have $n$ classical bits then together there are $2^n$ possible values that register can represent but only one is represented at any given time. This is already powerful but the quantum mechanical nature of qubits allow for vastly more power; a register of $n$ qubits in superposition can represent $2^n$ states \emph{simultaneously}.\newline 
Analogous to classical circuits, a quantum circuit is a series of quantum logic gates operating on some initial quantum register state. Typically, gates are defined with regards to their actions on qubits which can be described in matrix form. Incredibly, the only restriction required for a quantum logic gate is that it is a unitary operator. 
\begin{definition}
	A unitary operator ($\hat{U}$) when multiplied by its conjugate transpose ($\hat{U}^\dagger$) results in the identity matrix ($\mathds{1}$) \autocite{nielsen_quantum_2000}.
\end{definition}
For example, the single-qubit NOT gate swaps the amplitudes of a qubit. One can describe the effects of this gate using the following matrix
\begin{equation}
X = \begin{bmatrix}0&1\\1&0\end{bmatrix}.
\end{equation}
We see the effect of this operation on the qubit $\ket{\psi} = \alpha\ket{0} + \beta\ket{1}$.
\begin{equation}
X\begin{bmatrix}\alpha\\\beta\end{bmatrix} = \begin{bmatrix}\beta\\\alpha\end{bmatrix}.
\end{equation}
There are a number of possible sets of quantum logic gates which define universal computation, we present one such set. 
\subsubsection{Phase Shift}
The phase shift gate realises an arbitrary rotation in the computational basis. This gate has no real classical analogy as it operates exclusively on superposition states. Importantly, this gate leaves the probability amplitudes of a qubit untouched. It is described by
\begin{equation}
R_\phi = \begin{bmatrix}1 & 0\\0&e^{i\phi}\end{bmatrix}.
\end{equation}
\subsubsection{Hadamard Gate}
The Hadamard gate acts on a single qubit to produce an equal-superposition of two basis states. One may intuitively think of the Hadamard gate as transforming a qubit 'halfway' between two basis states \autocite{nielsen_quantum_2000}. Successive applications of the Hadamard gate on each in a register of $n$-qubits $\ket{0..0}$ will result in an equal superposition between the entire computational basis represented by $n$-qubits. It is described by the matrix
\begin{equation}
H = \frac{1}{\sqrt{2}}\begin{bmatrix}1&1\\1&-1\end{bmatrix}.
\end{equation}
\subsubsection{Control Not (CNOT)}
The control-NOT (CNOT) gate acts on two qubits known as the control $\ket{x}$ and target $\ket{y}$ qubits. The CNOT gate performs a logical NOT on the target qubit if and only if the control qubit is in the state $\ket{1}$ but leaves the control qubit unchanged. One can see that this is a quantum analogue to the classical XOR-gate. One can express the effect of this gate on two states as $\ket{\psi,\phi} \rightarrow \ket{\phi, \phi \oplus \psi}$ where $\oplus$ is addition-modulo $2$ (the definition of the classical XOR gate). The CNOT gate is described by the following matrix
\begin{equation}
\text{CNOT} = \begin{bmatrix}1&0&0&0\\0&1&0&0\\0&0&0&1\\0&0&1&0\end{bmatrix}
\end{equation}
\subsubsection{Measurement}
Measurement is not strictly a quantum-logic-gate as it permanently alters the state of a qubit (it is not a unitary operation). However, it is included in quantum circuit diagrams as it is necessary for useful computation. One may choose to think of qubit-measurement as a gate with a single qubit input and single classical bit output. Furthermore, if qubits are entangled with each other, the measurement of a single qubit will reveal the state of other qubits simultaneously \autocite{nielsen_quantum_2000}.
\subsection{Notable Algorithms}
In the three decades following the original definition of a gate-based quantum computer a vast number of surprising discoveries have been made and the physical implementation of useful universal quantum computers impend on the near future. While we do not present a summary of the entire field here, we summarise a selection of the most well-known quantum algorithms.
\subsubsection{The Deutsch-Jozsa Algorithm}
Deutsch and Jozsa \autocite{deutsch_rapid_1992} define the first quantum algorithm to find a solution more efficiently than any classical computer. While not directly practical this algorithm provides inspiration for more sophisticated algorithms, serving an excellent introduction to quantum algorithms. Suppose we have a function $f$ which accepts a single $n$-bit number from the range $[0,2^n-1]$ as input and produces either zero or one as an output. Further, this function is guaranteed to either be \textit{constant} meaning it returns the same value for all inputs or \textit{balanced} where exactly half of all inputs produce zero and the other half produce one. Using classical computers a deterministic algorithm requires $\frac{2^n}{2}+1$ queries to reach an answer. The original Deutsch-Jozsa algorithm \autocite{deutsch_rapid_1992} requires only two function evaluations to compute an answer and is deterministic. Later, Cleve et al. \autocite{cleve_quantum_1998} improve the Deutsche-Jozsa algorithm to only require a single query yet remain deterministic. This is clearly a vast improvement over a deterministic classical computer and still a sizeable improvement over stochastic classical algorithms \autocite{nielsen_quantum_2000}. Nielsen and Chuang \autocite{nielsen_quantum_2000} provide an excellent summary of this algorithm and its physical implementation. %The general concept is summarised as follows. We start with an $n$-qubit register (initialised to $\ket{0...0}$ and a second register of a single qubit. We start by performing a Hadamard gate to all qubits in the first register and the single qubit answer register followed by a single function evaluation performed by the gate $U^f$. This unitary operator $U^f$ com At this stage, the query register 'holds' resultant outputs of $f$ for all possible inputs. Finally,
\subsubsection{Shor's Algorithm for Integer Factorisation}
We begin by defining the integer factorisation problem which is equivalent to the discrete-logarithm and order-finding problems.
\begin{definition}
	Given a positive composite integer $N$, what prime numbers when multiplied together produce $N$? \autocite{nielsen_quantum_2000}
\end{definition}
Integer factorisation underpins public-key encryption systems widely used today. The exponential growth in classical complexity involved with factorising large primes have provided security for decades \autocite{chailloux_efficient_2017}. Peter Shor \autocite{shor_algorithms_1994} proposes a quantum algorithm to solve this problem with exponential speed-up over classical computers. The asymptotic run-time of Shor's algorithm grows polynomially with the length of the integer to be factored. Computing the quantum Fourier transform, a quantum analogue to the well established discrete Fourier transform, is core to this algorithm and is inspired by the Detsch-Jozsa algorithm \autocite{nielsen_quantum_2000}.
%The core component of this algorithm is the calculation of the quantum Fourier transform, a quantum analogue to the well established discrete Fourier transform inspired by the Deutsch-Jozsa algorithm \autocite{nielsen_quantum_2000}. 
Shor's algorithm is not deterministic and may require a polynomial number of repeat computations to produce a correct solution with high-probability \autocite{shor_polynomial-time_1997}. Shor's algorithm has thus produced the burgeoning field of quantum cryptanalysis. Cleve et al. \autocite{cleve_quantum_1998} present a method to break the well-known RSA cryptography scheme explicitly.\newline
Quantum cryptography is becoming a more prominent field spawning a large volume of research and commercial development. In the near-term, quantum key distribution schemes face rapid progress and immanent implementation \autocite{chailloux_efficient_2017}. In addition to security concerns, quantum computers pose a threat to relatively novel concepts such as crypto-currencies prompting careful threat analysis \autocite{aggarwal_quantum_2017}. Shor's algorithm to this day remains arguably the most infamous quantum algorithm. 
\subsubsection{Grover's Search}
Unstructured search is a very general problem which is defined simply. Given a finite set of possibilities, find options which satisfy a particular condition. In most practical contexts there exists structure in the search space which is exploited to design efficient algorithms; consider a binary search on a sorted list for example. Grover \autocite{grover_fast_1996} presents a quantum algorithm for unstructured search with time complexity $\mathcal{O}(\sqrt{N})$. Bennett et al. \autocite{bennett_strengths_1997} later show the lower limit of time complexity for this task using a quantum computer is $\Omega(\sqrt{N})$.\newline
Grover's search algorithm starts by preparing a superposition of $n$-qubits representing $2^n = N$ possible items. In this state, some number $M$ of the $N$ possible items will correspond to satisfactory elements and the rest will not. The vector sum over all desirable elements will produce a basis vector $\ket{\alpha}$ and a sum over all other unsatisfactory elements produce an orthogonal basis vector $\ket{\beta}$. In this new basis, a state with a high amplitude in the $\ket{\alpha}$ axis corresponds to a high probability of measuring a satisfactory item. A subroutine termed a 'Grover iteration' is performs a rotation in the $\ket{\alpha}, \ket{\beta}$ basis in the $\ket{\alpha}$ direction and is applied $\mathcal{O}(\sqrt{N})$ times. Specifically the number of iterations is approximately $\pi\sqrt{N}/4$ when searching for a single item \autocite{grover_how_1998}. After this a measurement is made revealing with high-probability an item satisfying our query \autocite{nielsen_quantum_2000}. The exact number of iterations is dependent on each problem instance growing polynomially and optimally \autocite{grover_how_1998}.\newline
The optimal nature of Grover's algorithm is a surprising result motivating further research into quantum computing. Exploiting the exponential scale of quantum information remains a central goal as quantum algorithms are sought to solve previously intractable classical problems.
\subsection{Quantum Supremacy}
Quantum supremacy describes the potential ability for quantum computers to outperform classical computers for some problems. This problem fundamental to the field of quantum computing is difficult to demonstrate for a number of reasons. The performance of a quantum computer must be proven to be superior to any classical computer requiring rigorous proofs of lower-bound complexity for both quantum and classical formulations of a problem. Boixo et al. \autocite{boixo_characterizing_2016} suggest the construction of particular problems simulating quantum phenomenon to aid in this endeavour. Relaxing the definition of quantum supremacy permit the use of benchmarking and practical performance as measures of supremacy inspired by the evaluation of heuristic-based algorithms in classical computing. Formal supremacy is an important milestone in the field, however useful quantum computation is the true goal of the field.\newline
\subsection{Classical Quantum Simulation}
The obvious approach to establish supremacy is to simulate quantum circuits of increasing size. The point at which simulation becomes intractable reveals a lower limit on formal quantum supremacy. Such a bound is difficult to show analytically resulting in the some of the largest single-task computations in history.\newline
The recently developed high-performance distributed quantum simulator \textit{qHiPSTER} \autocite{smelyanskiy_qhipster:_2016} allows for simulation of quantum circuits up to 40-qubits in scale. Smelyanskiy et al. \autocite{smelyanskiy_qhipster:_2016} find memory to be the limiting factor of simulation postulating that simulations greater than 49 qubits are in-feasible until 2024.\newline
Boixo et al. \autocite{boixo_characterizing_2016} build on the work of Smelyanskiy \autocite{smelyanskiy_qhipster:_2016} by simulating 42-qubit circuits using 70 terabytes of memory. Further Boixo et al. \autocite{boixo_characterizing_2016} formalise the task to demonstrate supremacy based on building very dense, partially randomised circuits acting upon a grid of qubits. This introduces both size and depth as simulation bounds. Importantly, such a measure is shown to be efficiently measured on a hypothetical physical quantum processor. The scheme is based on multiple fast evaluations of a circuit revealing a single sampled output. Over a vast number of samples an accurate distribution is achieved.\newline
H\"{a}ner and Steiger \autocite{haner_0.5_2017} simulate a 45-qubit system. Deriving improvements in memory and communication overheads by kernel optimisation and a scheduling algorithm ordering the processing of sub-circuits. This simulation, the largest of its time, required 8,192 nodes and 0.5 petabytes of memory.\newline
Pednault et al. \autocite{pednault_breaking_2017} provide methods to simulate beyond the previously conceived 49-qubit limit using only three terabytes of memory simulating 56-qubit circuits. The scheme employed by Pednault et al. \autocite{pednault_breaking_2017} reformulates circuit simulation as tensor operations cutting down on memory requirements significantly and allows the use of more generalised tensor mathematics. Similar to the scheduling concept used by H\"{a}ner and Steiger \autocite{haner_0.5_2017}, the computation of gates which entangle qubits (introducing an exponential factor) is deferred.\newline
Boixo et al. \autocite{boixo_simulation_2017} further improve their original scheme \autocite{boixo_characterizing_2016} formulating circuit execution as an un-directed graph. A variable elimination scheme is developed reducing average memory requirements. This scheme is especially powerful for smaller circuits allowing workstation simulation of a larger scale than previously possible.\newline
Chen et al. \autocite{chen_64-qubit_2018} apply more aggressive gate partitioning producing exponentially more independent circuits to simulate, allowing better use of distributed resources. Further, Chen et al. \autocite{chen_64-qubit_2018} estimate the computational cost of simulating a 72-qubit circuit, deeming it feasible for a computer identical to that used by Pednault et al \autocite{pednault_breaking_2017}.\newline
Li et al. \autocite{li_quantum_2018} compute both sampling and full simulation tasks for circuits of 49-qubits at 39 and 55-depths respectively. A gate partitioning scheme in addition to dynamic programming methods are used to construct an efficient ordering of sub-tasks reducing memory overheads. 131,072 nodes and nearly one petabyte of memory are used.\newline
Chen et al. \autocite{chen_64-qubit_2018} extend the variable elimination work of Boixo et al. \autocite{boixo_simulation_2017} and apply it in a distributed manner. Circuits of varying sizes and depths are analysed factoring estimated real-world noise with the intent to derive a lower-bound on hardware accuracy. Again, 131,072 nodes and around one petabyte of memory are used.\newline
Markov et al. \autocite{markov_quantum_2018} refine the benchmarks defined by Boixo et al. \autocite{boixo_simulation_2017} further increasing classical simulation complexity. The use of public cloud resources allow Markov et al. \autocite{markov_quantum_2018} to associate a monetary cost with such simulations generating further motivation for implementing such a scheme in quantum hardware. 
\newline\newline
Quantum computing brings focus to two frontiers; the fundamentals of computing as we gain an understanding of qubits, and cutting-edge classical simulation. The few algorithms already discovered bring large implications to the world of computer science encouraging research into finding quantum advantage in increasingly challenging and diverse classical problems. We introduce one such algorithm in the following section designed to make use of both quantum and classical machines in an effort to realise practical quantum computation sooner. 

\section{The Quantum Approximate Optimisation Algorithm (QAOA)}
In section \ref{sec:QAOA} we introduce the algorithm itself, here we discuss prior investigation and experimentation.
\subsection{The QAOA and Quantum Supremacy}
The requirement that a quantum supreme algorithm must exhibit performance superior to any classical algorithm is difficult to formulate. The ultimate goal to implement a quantum supreme algorithm on physical hardware remains an open yet vital problem. Farhi et al. \autocite{farhi_quantum_2014} present an analysis of the QAOA applied to the E3LIN2 problem, a linear equation optimisation intending to demonstrate provable supremacy. Taking $p = 1$ Farhi et al. \autocite{farhi_quantum_2014} provide an analytic formulation to show their result. The QAOA produces answers satisfying $\frac{1}{2} + \Omega(D^{\frac{-3}{4}})$ of the required clauses.\newline
Spurred by this claim, Barak et al. \autocite{barak_beating_2015} present a superior classical algorithm for the same problem satisfying $\frac{1}{2} + \Omega(\frac{1}{\sqrt{D}})$ fraction of the required clauses.\newline
However, the formulation of the QAOA examined by Farhi et al. \autocite{farhi_quantum_2014} is a coarse approximation of the QAA using only a single trotterisation ($p$ = 1). Fahri et al. \autocite{farhi_quantum_2014} suggest a number of possible improvements requiring further analysis such as increasing $p$ and introducing variables for each clause to be optimised. Such improvements are difficult to formulate analytically and hence experimental motivation for such analysis is required to justify such work.\newline
For these reasons Farhi and Harrow \autocite{farhi_quantum_2016} propose the QAOA may still demonstrate quantum supremacy. The QAOA may demonstrate quantum supremacy in two ways. \newline
Farhi and Harrow \autocite{farhi_quantum_2016} argue the inherent quantum nature of the QAOA itself cannot be replicated classically. More specifically if there did exist such an algorithm, Farhi and Harrow \autocite{farhi_quantum_2016} propose the complexity hierarchy would collapse. Secondly, physical quantum computers will allow the QAOA to be run on problem instances prohibitively large for classical computation and hence may generate superior solutions in these instances. This provides evidence for the QAOA to be among the first algorithms implemented in quantum hardware despite classical competition to find superior algorithms.
\subsection{Optimisation of the Classical Component}
The hybrid nature of the QAOA naturally leads to two frontiers of development and research, the quantum and classical components. A large proportion of efforts understandably focus on the quantum component an understanding of the classical component is essential. Guerreschi and Smelyanskiy \autocite{guerreschi_practical_2017} investigate three classical optimisation methods for hybrid quantum algorithms with experimental analysis on the QAOA. Gradient-free and quasi-Newton methods are investigated in an experimental manner. The Nelder-Mead algorithm for gradient-free optimisation is an appropriate method for the general case of a small value of $p$ while the quasi-Newton method using finite derivative methods provides superior results with a matching increase in implementation complexity \autocite{guerreschi_practical_2017}. 
The function space the QAOA generates is typically very difficult to form a gradient in, hence the increase in computational complexity. The work of Guerreschi and Smelyanskiy \autocite{guerreschi_practical_2017} lays a solid foundation for experimental simulation of the QAOA with regards to complexity and performance.
\subsection{Extensions of the QAOA}
Augmentation and extension of the QAOA is possible in addition to direct optimisation.
Instead of solving optimisation problems directly Wecker et al. \autocite{wecker_training_2016} modify the QAOA to find a quantum state which seeks a maximal overlap between the objective function and the ground energy state of the given instance. This approach is applied to the MAX 2-SAT problem. This change in formulation should lead to more accurate results but introduces more complexity into the classical optimisation. The general approach Wecker et al. \autocite{wecker_training_2016} employ uses classical machine learning techniques to train the algorithm with known difficult instances of the MAX 2-SAT problem. The training yields an optimisation schedule describing how the algorithm explores the state space of possible parameter-values and is subsequently tested on a another set of problem instances. This machine learning approach solves instances of MAX 2-SAT and MAX 3-SAT faster than the well known annealing scheme CFLLS \autocite{crosson_different_2014}, demonstrating the largest improvement in the hardest instances. This experimental approach provides concrete data suggesting this modified QAOA is also suitable for near-term implementation. However, a generalisation of this approach and the specific learning methods employed are not discussed in detail. Additionally, performance comparison to the original formulation of the QAOA is not presented.\newline\newline
Rather than adjusting the main objective of the QAOA, Hadfield et al. \autocite{hadfield_quantum_2017} extend the QAOA to a more general framework termed the Quantum Alternating Operator Ansatz (Referred to as QAOA in the original paper out of respect but here as the QAOAn to avoid confusion). The main modification considers general parameterised families of unitaries over the specific family of fixed local Hamiltonians. Loosely speaking this allows the QAOAn to operate on registers describing a wider range of quantum systems which in turn allows the encoding of more varied problems. This is accomplished by varying the formulation of the mixing operator $\hat{U}_B$ to restrict considered bit-strings to valid solutions only. This decomposes the mixing operator into a number of operations which makes the algorithm more powerful but more difficult to implement. Hadfield \autocite{hadfield_quantum_2018} describes a large number of problem specific formulations of the QAOAn including very well known hard classical problems such as the travelling salesman and job scheduling to demonstrate the potential impact of this reformulation of the QAOA. Unlike Farhi and Harrow's  \autocite{farhi_quantum_2016} original approach to demonstrating quantum supremacy analytically, Hadfield et al. \autocite{hadfield_quantum_2017} propose an empirical approach similar to how heuristic algorithms are often analysed. Their reasoning cites the difficulty often found when proving the supremacy of heuristic algorithms directly versus the easier task of bench-marking an algorithm over a suitable well-known set of problem instances. This approach is similar to that of Wecker et al. \autocite{wecker_training_2016}. Disappointingly no experimental data is presented, however this work provides a solid foundation for future investigation and publications of the QAOAn due to the large number of problem specifications provided.\newline
Mash and Wang \autocite{marsh_quantum_2018} propose a more tightly describe alteration to the QAOA for NPO PB problems based on the observation that the mixing operator describes a continuous time quantum walk on the quantum register. By imposing restrictions on this operator bit-strings are partitioned into feasible and in-feasible sets allowing for greater performance. 

\section{Computing the Matrix Exponential}
Simulating the QAOA requires solving the time-dependent Schr\"{o}dinger equation
\begin{equation}
\ket{\Psi(t)} = e^{-iHt}\ket{\Psi(0)},
\label{eq:schrodinger}
\end{equation}
as a central component of the algorithm. Efficient and accurate numerical approximate of Equation \ref{eq:schrodinger} requires efficient and accurate computation of the matrix exponential for large, sparse and complex-valued matrices. We define the matrix exponential for completeness. 
\begin{definition}
	$exp(A) \equiv e^A = \sum_{n=0}^{\infty}\frac{A^n}{n!} = I + A + \frac{A^2}{2!} + \frac{A^3}{3!} + ...,$
	\label{def:matrixExpm}
\end{definition}
Matrix exponential computation is a highly investigated problem with over 35 years of investigation. Despite these efforts there is no single superior method, rather an array of methods each with their own intricacies, benefits and short-comings. Moler and Van Loan \autocite{moler_nineteen_1978} present a now canonical review of $19$ candidate methods. Their work remains so influential Moler and Van Loan \autocite{moler_nineteen_2003} present an updated version 35 years later. Careful algorithm selection and implementation is key to guarantee both performance and numerical accuracy. Such choices are highlighted by implementation for HPC users.\newline
We note that for the diagonal case, computing the matrix exponential involves exponentiation each element of the matrix \autocite{moler_nineteen_1978}. Through eigenvalue decomposition one can diagonalise most matrices reducing the exponential to this simpler case in addition to two matrix-matrix multiplications. However, such a method requires the use of sophisticated eigenvalue solvers, a major computational effort in of itself which may be slower than many other methods in the general case.\newline
Computing the Taylor series directly results in a slow-convergence and low-accuracy in the general case. Using a Pad\'{e} approximation provides better accuracy with less terms, however again, na\"{i}ve application of series expansion results in poor general-case performance. As such more sophisticated methods provide superior performance and are necessary for practical use.
\subsubsection{Scaling and Squaring}
The most popular method available for dense matrices is scaling and squaring. This method relies on a property unique to the matrix exponential
\begin{equation}
e^A = (e^{A/m})^m.
\label{eq:scalSquare}
\end{equation}
Selecting $m$ carefully as the smallest power of two such that $\norm{A}/m \leq 1$ allows for accurate and efficient use of Taylor or Pad\'{e} approximants. Scaling and Squaring is among the most widely used methods due to strong accuracy and elegance. Al-Mohy and Higham \autocite{al-mohy_new_2010} champion this error presenting highly in-depth error analysis and precise algorithms for computing optimal $m$ for IEEE precision arithmetic. Further Higham and Tisseur \autocite{higham_block_2000} present an algorithm for estimating the $1$-norm of arbitrary matrices, a key component of the aforementioned matrix exponential algorithm \autocite{al-mohy_new_2010}. Scaling and squaring is widely implemented in many commercial packages such as MATLAB, Scipy, Mathematica and Expokit \autocite{sidje_expokit:_1998}. Scaling and squaring is best suited for dense matrix exponentiation and is thus poorly suited for distributed implementation due to the use of matrix-matrix products.
\subsection{Computing $e^A \cdot \mathbf {v}$}
In many cases including our own, the computation of the action is the matrix exponential on a vector is our task. This slightly different problem allows for alternative methods to be used with potentially less computational overhead. 
\subsubsection{Scaling and Squaring}
Higham and Al-Mohy \autocite{al-mohy_computing_2011} present a method based on their algorithm \autocite{al-mohy_new_2010} for computing the action of the matrix exponential. Now their algorithm determines through one-norm estimation the optimal scaling value to minimise the number matrix-vector multiplications required. Aside from the estimation of the one-norm this algorithm is currently untested in a distributed memory implementation.
\subsubsection{Krylov Subspace}
The most popular method for large, sparse matrix exponentiation is the Krylov subspace method \autocite{moler_nineteen_2003}. This method approximates the matrix exponential onto a smaller Kylov subspace which then allow for dense matrix methods to be applies efficiently. Re-using the constructed subspace allows successive value of $t$ to be computed at low-cost and as such is considered the canonical sparse-matrix method. Mathamatica's \verb|MatrixExp[A,v]|, Expokit \autocite{sidje_expokit:_1998} and SlepC/PetSc \autocite{balay_petsc_2018}, \autocite{hernandez_slepc:_2005} implement the Krylov subspace method. Furthermore, SlepC/PetSc offer the only commercially available distributed memory implementation of the matrix exponential.
\subsubsection{Chebyshev Approximation}
The Chebyshev approximation method spawns from quantum chemistry where a series approximation of the matrix exponential is computed by the Chebyshev polynomial forming each step \autocite{fang_one_1996, wang_time-dependent_1998, wang_quantum_1999, ndong_chebychev_2010}. Post-multiplying the Chebyshev series with our vector $v$ allows for direct computation of $e^{tA}\cdot\mathbf{v}$ without ever computing a full exponential matrix. Bessel J zero functions form the coefficients of the expansion which allow for fast and accurate convergence. Chebyshev approximation requires either eigenvalue scaling or use of the dense scaling and squaring method. However estimates of the eigenvalues do not effect accuracy greatly, but effects the number of iterations required for convergence \autocite{izaac_pyctqw:_2015}. The Chebyshev method is appealing for HPC applications has only matrix-vector or vector-vector operations are required when eigenvalue scaling is used. This allows for trivial memory parallelisation with minimal communication making this method a strong candidate growing support in HPC applications. Furthermore Auckenthaler et al. demonstrate that the Chebyshev method is superior to the scaling and squaring method \autocite{auckenthaler_matrix_2010} while Bergmaschi et al. \autocite{bergamaschi_efficient_2000} suggest the Chebyshev method is superior to the Krylov subspace method. Despite practical performance, very few packages implement this method: Expokit \autocite{sidje_expokit:_1998} implements Chebyshev approximation for dense matrices, and pyCTQW \autocite{izaac_pyctqw:_2015} a Python package built upon PetSc/SlepC to simulate continuous time quantum walks.\newline
Computing the action of the matrix exponential is critical to efficient and accurate simulation of the QAOA at desktop to cluster-scale. 
\section{Quantum Computing Applied to Graph Similarity Problems}
We are not the first to consider applying quantum computing to classically difficult graph theoretic problems. Lucas \autocite{lucas_ising_2014} provides a vast number of mappings for classical NP-Complete problems to the Ising model of computing. The Ising model of computation can in turn be mapped onto a quantum annealer through the quantum adiabatic algorithm (QAA)\autocite{farhi_quantum_2001}. Hen and Young \autocite{hen_solving_2012} map the graph isomorphism problem to a quantum annealer with experimental results. Hen and Young analyse experimental implementation supporting the conjecture that quantum annealers can discriminate between non-isomorphic graphs \autocite{hen_solving_2012}. Furthermore they suggest that hardware and simulation improvements will better validate their claims. 
\newline\newline
Graph similarity is an openly difficult problem to compute classically despite the vast practical use it demonstrates. As the field of quantum computing matures historically intractable problems are explored with often surprising results expanding the scope of what feasible computing permits. We also provide a brief introduction to quantum computing assuming no prior knowledge of quantum physics in addition to a few historic quantum algorithms. We present the Quantum Approximate Optimisation Algorithm as a general method to approach hard combinatorial optimisation problems; optimisation and extensions to the algorithm are an open area of research. Seeking the limits of computation will always be integral to the field of computer science, an endeavour extending into the realm of quantum computing. 
\chapter{Methods}
\section{The Quantum Approximate Optimisation Algorithm (QAOA)}
The Quantum Approximate Optimisation Algorithm (QAOA) stands unique among many other quantum algorithms as it is in essence, a Monte-Carlo Algorithm. The solution quality various for a fixed amount of execution but has in practice shown excellent performance inspiring academic and industrial investigation \cite{otterbach_unsupervised_2017}. To encode a problem into the QAOA we require the following.
\begin{itemize}
	\item A problem Hamiltonian $\hat{C}$ which implements the cost function of our candidate problem
	\item A mixing Hamiltonian $\hat{B}$ which defines which bit-strings are permitted for evaluation by the algorithm.
	\item A suitable amount of decomposition ($p$ value)
	\item An initial state generation scheme
\end{itemize}
A QAOA iteration is defined as $p$ applications of successive $\hat{U}_C$ and $\hat{U}_B$ unitary operators. At each application of $\hat{U}_C, \hat{U}_B$ a corresponding corrective parameter $\gamma_i, \beta_i$ is applied to approximate an annealing scheme while dropping the adiabatic requirement of the QAA. Repeated sampling generates an approximate expectation value which is fed into a parameter optimisation scheme to select new $(\vec{\gamma}, \vec{\beta})$. This process repeats until a termination criteria is met or the quantum compute time is exhausted.\newline
Hadfield \cite{hadfield_quantum_2018} and Marsh and Wang \cite{marsh_quantum_2018} introduce restrictions upon the mixing operator restricting the algorithm to feasible solutions at a cost of higher gate depth. Marsh and Wang \autocite{marsh_quantum_2018} reverse the order of applying $\hat{U}_C$ and $\hat{U}_B$ in order to account for the non-trivial maximal energy state introduced by applying mixing restrictions. We provide a novel mapping of permutation based problems such as graph similarity via edge-overlap to the QAOA.
\subsection{General Problem Encoding}
Problem operators are encoded into annealing schemes such as the QAA through the Ising spin-glass model of computation \autocite{santoro_theory_2002} which allows the definition of problem Hamiltonians of the following form.
\begin{equation}\label{eq:spin}
\hat{H} = \mu\sum_{i}h_i\sigma_i + \sum_{i,j}\sigma_i\sigma_j.
\end{equation}
Problems are encoded using pseudo-boolean functions \autocite{boros_pseudo-boolean_2002} and permit a wide variety of difficult problem definition \autocite{lucas_ising_2014}. The first term defines the logic for setting a particular spin $x : \{-1, 1\}$. In the gate-based QAOA each 'spin' corresponds to a qubit encoding a binary variable $x : \{0, 1\}$.
\section{Graph Similarity via QAOA}
\subsection{A Canonical Mapping}
The standard method to map permutation based problem use unary encodings of $n^2$ qubits. To illustrate, consider for an arbitrary pair of graphs each consisting of $V$ vertices a $V^2$ string of qubits
\begin{equation}\label{eq:unaryPerm}
[x_{1,1}x_{1,2}\dots x_{1,v}] [x_{2,1}x_{2,2}\dots x_{2,v}]\dots [x_{v,1}x_{v,2}\dots x_{v,v}],
\end{equation}  
where a binary variable $x_{i,j}$ represents vertex $j$ in $G_2$ mapping to vertex $i$ in $G_1$. Under such a scheme the vast majority of the $2^{V^2}$ bit-strings are not feasible since $n! << 2^{n^2}$. Such an encoding quickly becomes intractable for both simulation and physical quantum hardware where the current state of the art is around $72$ qubits \autocite{markov_quantum_2018} permitting graphs of less than eight vertices. Hadfield \autocite{hadfield_quantum_2018} provides a QAOA mapping for the travelling salesman and other permutation problems but requires mixing constraints to enforce legal candidate solutions. A unary encoding of graph similarity through edge overlap suffers from the same issue.
\subsection{Defining $\hat{C}$}
We propse a novel compact encoding requiring $\mathcal{O}(\lceil log_2(V!)\rceil)$ qubits prepared in $\mathcal{O}(V^3)$ operations. We define our problem Hamiltonian as
\begin{equation}\label{eq:EOHam}
\hat{C} = A\sum_{\sigma \in V!}\sum_{i=1}^{V}\sum_{j=1}^{V}\sum_{u=1}^{V}\sum_{v=1}^{V}(d_1(i,j)d_2(u,v))(x_{i\sigma(u)}x_{j\sigma(v)}).
\end{equation}
This provides us with the edge overlap for all $V!$ permuted labels of $u,v \in G_2$. Our problem operator $\hat{C}\ket{\psi} = C(x)\ket{\psi}$ that is the application of our cost function to all bit-strings $x$. Our diagonal problem operator $\hat{U}_C$ becomes
\begin{equation}\label{eq:GsimUC}
\hat{U}_C(\gamma)\ket{\psi} = e^{-i\gamma c(x)}\ket{\psi}.
\end{equation}
Welch et al. \autocite{welch_efficient_2014} builds on the work of Childs \autocite{childs_quantum_2004} providing a method to implement $e^{-i\gamma\hat{C}}$ efficiently without the use of additional ancillary qubits where each element of the diagonal $\hat{C}$ is itself efficiently computable. We provide a graphical representation of $\hat{C}$ in Figure 3.1
\begin{center}
	\begin{figure}\label{fig:UC}
		\centering
		\includegraphics[width=0.5\textwidth]{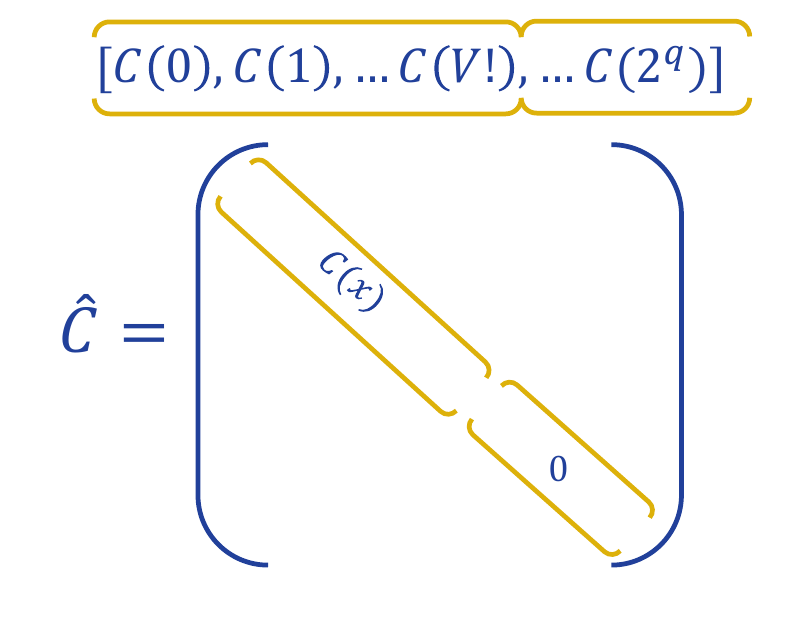}
		\caption{$\hat{C}$}
	\end{figure}
\end{center}
\subsection{Defining $\hat{B}$}
 We define a canonical mixing operator \autocite{farhi_quantum_2014-1}
\begin{equation}\label{eq:Mixing}
\hat{B} = \sum_{i=1}^{n}\sigma_i^x,
\end{equation}
where $\sigma^x$ is the Pauli-x matrix, the quantum equivalent of the NOT gate. Our mixing operator $\hat{U}_B$ becomes
\begin{equation}\label{eq:GsimUB}
\hat{U}_B(\beta)\ket{\psi} = e^{-i\beta\sum_{i=1}^{n}\sigma_i^x}\ket{\psi} = e^{-i\beta\hat{B}}\ket{\psi}.
\end{equation}
We provide a graphical representation of $\hat{B}$ in Figure 3.2.
\begin{center}
	\begin{figure}\label{fig:UB}
		\includegraphics[width=\textwidth]{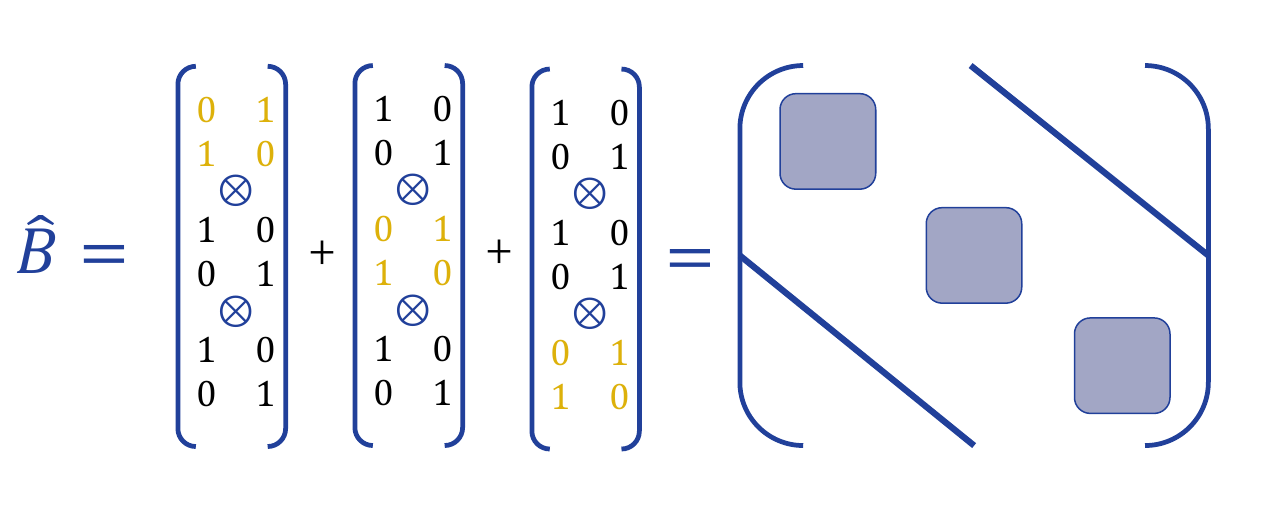}
		\caption{$\hat{B}$}
	\end{figure}
\end{center}
\subsection{Permutation Mapping}
Abrams and Lloyd \autocite{abrams_simulation_1997} provide a method to encode a superposition of $n$ elements using $\mathcal{O}(n^3)$ operations. Chiew et al. \autocite{chiew_graph_2018} provide a full circuit encoding edge overlap values using $\mathcal{O}(n^2log^2(n))$ operations and $\mathcal{O}(nlog(n))$ qubits allowing efficient bit-string mapping over the range $[0,n!]$ to edge-overlap values. To index all such values in $\hat{U}_C$ we require $\lceil log_2(V!)\rceil$ qubits (termed $q$) since for all $n > 2, n! > 2^n$. We trivially map the remaining 'tail' values not included in the range $2^q$ to zero.\newline
Typically the initial state of a QAOA iteration is the poorest possible solution. In our case this difficult to determine thus we use a superposition of all $q$ qubits.
This compact representation suffers from these 'tail' values. We plot the first $2000$ terms in Figure \ref{fig:Tail}.
\begin{figure}[htbp]
	\centering	
	\includegraphics[width=0.5\textwidth]{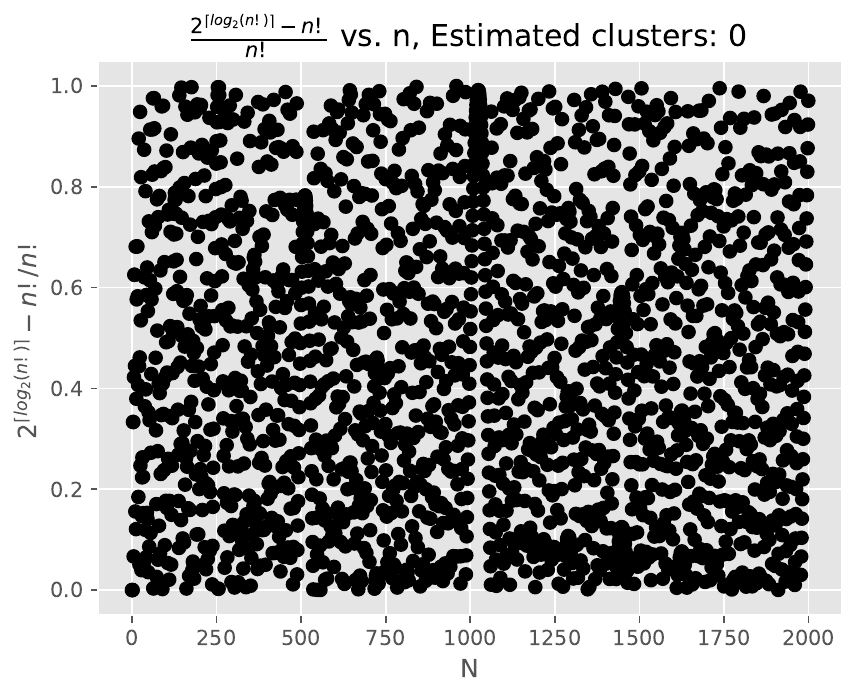}
	\caption{The proportion of infeasible to feasible bit-strings}
	\label{fig:Tail}
\end{figure}
These degenerate values are frustrating from an optimisation standpoint since the QAOA will always start from an equal probability of measuring any bit-string. We cannot apply per-qubit mixing restrictions. The proportion of the tail reveals a uniform distribution (Kolmogorov-Smirnov test \cite{jones_scipy:_2001} for uniform distribution, p-value $1.46e-08$). In Figure \ref{fig:Tail}, a DB-clustering \cite{jones_scipy:_2001} of points confirm this observation further, labelling all points as noise. For interest we present a table of initial data points and present a series expansion of this phenomena in Appendix \ref{ap:Tail}. 
\subsection{Test-Case Generation}\label{sec:TestCase}
Graphs are generated using the standard Erd\"{o}s-R\'{e}nyi method; we randomly assign each edge with $50\%$ probability. Instead of testing two random graphs, we deform an initial graph in order to create difficult test-cases where graphs are quite similar. This allows us to test for overall correctness for each individual trial and over multiple trials. We consider the following deformations
\begin{itemize}
	\item Isomorphism: The original graph is compared to itself. The first element of the problem Hamiltonian is guaranteed to be zero.
	\item Vertical Flipping: The original graph is mirrored in the vertical axis. This generates very dissimilar graphs with minimal effort shown in Figure \ref{fig:flip}
	\item Edge Addition: $v$ non-existent edges are added at random
	\item Edge Removal: $v$ existent edges are removed at random
	\item Edge Addition and Removal: $v$ edges are added and removed randomly
\end{itemize}
\begin{figure}[htbp]
	\centering	
	\includegraphics[width=0.5\textwidth]{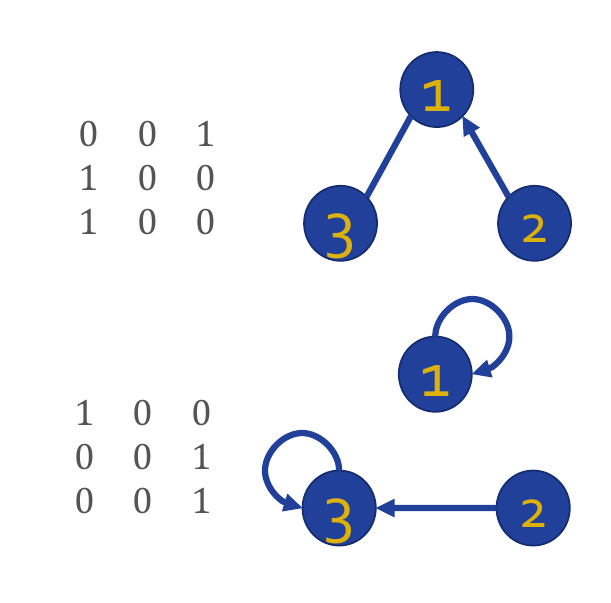}
	\caption{Adjacency Flip}
	\label{fig:flip}
\end{figure}
\section{Simulation Design}
Qolab (Quantum Optimisation Laboratory) is a flexible simulation package for the QAOA and related algorithms from desktop to cluster scale implemented in C. Our package allows for a single interface to both desktop and cluster-based code at the highest level of abstraction possible. The user is required to implement only the cost function used to define $\hat{U}_C$ and any walk masks applied to $\hat{U}_B$ if using the modified QAOA described by Marsh and Wang \autocite{marsh_quantum_2018}. We present an extensive description in Appendix \ref{ap:Qolab}. Exacting implementation targets maximal single node and desktop performance. Multiple processors are utilised by distributing the state-space. \newline
Matrix operations are handled using Intel's Math Kernel Library \autocite{noauthor_intelr_nodate} and we store matrices in compressed column form. \autocite{duff_sparse_1989}.
%Prior code attempts?
\subsection{Simulation Components}
Direct state-vector simluation avoids a large amount of overhead present in contemporary gate-based simulations such as IBM's Qiskit Alpha \cite{ibm_qiskit_nodate}. We capture the high-level structure of our QAOA simulation in Algorithm \ref{alg:Qolab} and a similar description of a walk-restrictive QAOA in Appendix \ref{ap:Pseudo}. The following sections describe the implementation and computational complexity of each component.
\begin{algorithm}
	\caption{Qolab Overview}\label{alg:Qolab}
	\begin{algorithmic}[1]
		\Procedure{QAOA\_Core}{numQubits, P, optimisationMethod, C(x)}
		\State $\hat{U}_C \gets $\textbf{genUC}(C(x))
		\State $\hat{U}_B \gets $\textbf{genUB}(numQubits)
		\State $\vec{\gamma}, \vec{\beta} \gets $\textbf{initialParameters()}
		\While{terminateTest()}
			\State $\ket{\psi} \gets $\textbf{initialState}
			\For{$i = 0$ to $p$}
				\State $\ket{\psi} \gets \hat{U}_C(\gamma_i)\ket{\psi}$
				\State $\ket{\psi} \gets \hat{U}_B(\beta_i)\ket{\psi}$
			\EndFor
			\State $F_p \gets \textbf{Measure}(S)$
			\State $\vec{\gamma}, \vec{\beta} \gets $ \textbf{updateParameters($F_{P}$)}
		\EndWhile
		\EndProcedure
		\State \textbf{Report}()
	\end{algorithmic}
\end{algorithm}
For clarity we make the following definitions:
\begin{itemize}
	\item $V$ - The number of vertices in candidate graphs
	\item $q$ - A number of qubits. For graph similarity this is $2^{\lceil log_2(V!)\rceil}$
	\item $n$ - A number of classical items
	\item $p$ - The amount of trotterisation applied to the QAOA
	\item $P$ - The number of processors present in our cluster
	\item $n_P$ - The size of the state-vector maintained by each processor
	\item $\ket{\psi}$ - The state-vector used by the QAOA
\end{itemize}
\section{Generating $\hat{U}_C$}
\subsection{Permutation Generation}
The power of the QAOA derives performance from the ability to process an exponential number of inputs to our cost function simultaneously. This requires the generation of all $q!$ bit-strings and application to the provided cost function $C(x)$. The time complexity is $\mathcal{O}(q!)$. We are able to generate the $k$-th permutation through a Lehmer code based algorithm. The time complexity to generate all permutations is $\mathcal{O}(log_2(n!))$. Pseudocode is presented in Appendix \ref{ap:Pseudo}. Heap's algorithm has ldong been considered the fastest method to generate all permutations of $n$ elements\autocite{sedgewick_permutation_1977}. Our sub-optimal k-th based permutation scheme allows for independent generation of permutations across distributed processes achieving optimal $\mathcal{O}(n!/P)$ work per process. Further this scheme matches the order generated by our theoretical QAOA encoding.\newline
Nevertheless, since we require a $2^{\lceil log_2(v!)\rceil}$ state-space generations of the $q!$ permutations are reduced in the asymptotic complexity.
\subsection{Generating $\hat{C}$}
The cost function operator $\hat{U}_C$ is generated by evaluating the provided cost function for all permutations of $q$ bits resulting in a $2^q \times 2^q$ diagonal matrix called $\hat{C}$. The time complexity to build $\hat{C}$ is $\mathcal{O}(q! \times Poly(q))$ where $Poly(q)$ is the cost-function itself. For edge-overlap this is $\mathcal{O}(V^2)$. There is little need to parallelise the cost function as it is computationally dominated by the number of bit-string considered.
\section{Generating $\hat{U}_B$}
$\hat{U}_B$ is less trivial to simulate. The definition of $\hat{B}$ given by Equation \ref{eq:B} can be re-phrased with respect to matrix elements directly. The operator defines valid transitions from bit-string to bit-string. 

Without mixing restrictions the original definition connects all bit-strings which differ by a single element. A na\"{i}ve implementation computes the original $\hat{B}$ matrix directly by performing a sequence of Kronecker products. However under this new observation the generation method becomes trivially distributable. We present pseudocode in Algorithm \ref{alg:UB} which includes the application of walk masks described by Marsh and Wang \autocite{marsh_quantum_2018}. When no masks are applied the generated $2^q \times 2^q$ matrix matches that defined in Equation \ref{eq:B} which represents a maximally connected hypercube between candidate bit-strings.  
\begin{algorithm}
	\caption{$\hat{B}$ Generation}\label{alg:UB}
	\begin{algorithmic}[1]
		\Procedure{GenerateB}{numQubits, mask}
			\State UB $\gets \emptyset$
			\State nnz = 0
			\For{$i \gets 0$ to $2^q$}
				\State colB[$i$] $\gets$ nnz
				\For{$j \gets 0$ to $2^q$}
					\State row $\gets i \lor (1 << j)$
					\If{mask(row)}
						\State values[nnz] $\gets$ $1$
						\State colE[$i$] $\gets$ nnz
						\State rowInd[nnz] $\gets$ col
						\State nnz $\gets$ nnz + $1$
					\EndIf
				\EndFor
			\EndFor
		\EndProcedure
		\State \textbf{Report}()
	\end{algorithmic}
\end{algorithm}
\section{Distribution Scheme}
The distribution of our simulation is based on the state-vector itself. We initially consider a scheme where each process is responsible for a unique sub-set of the $2^q$ elements defined as
\begin{equation}\label{eq:decomp}
n_P = \lfloor\frac{2^q}{P}\rfloor.
\end{equation}
This decomposes seamlessly where $log_2(P) \in \mathcal{Z}$ since the state-space grows exponentially. When this is not true we append the remaining items to the final process. The number of appended items is defined as
\begin{equation}\label{eq:decomp_remain}
n_{Pend} = 2^q mod P,
\end{equation}
which by definition grows $\mathcal{O}(P)$. Since $P \ll 2^q$ we deem this acceptable. This decomposition has the benefit of optimal load balancing between all processes in the general case. This simple method works well for distributing the state-vector and realising $\hat{U}_C$ however $\hat{U}_B$ requires more nuance.
\section{Function Evaluation}
Function evaluation method realises a single QAOA iteration defined in Equation \ref{eq:state}. The generation of an initial state is in our case an equal superposition of $q$ qubits. This is represented as a $2^q$ vector of $\frac{1}{\sqrt{2^q}}$ indicating an equal chance of measuring any candidate bit-string. The state-vector $\ket{\psi}$ is distributed equally across all processes according to Equation \ref{eq:decomp}. We describe each part of the function evaluation separately.
\subsection{Setup}
The root node holds the set of parameters ($\vec{\gamma}, \vec{\beta}$) broadcast to all processes. The communication overhead grows $\mathcal{O}(p)$.
\subsection{Applying $\hat{U}_C$}
We need to realise the action 
\begin{equation}\label{eq:UCv}
e^{-i\gamma\hat{C}}\ket{\psi}.
\end{equation}
Since $\hat{C}$ is a diagonal matrix, we can compute $e^{-i\gamma\hat{C}}$ by exponentiating each element of $\hat{C}$ and apply a point-wise multiplication with $\ket{\psi}$ or dot-product with $\ket{\psi}^T$. This requires $\mathcal{O}(2^q)$ operations but is embarrassingly parallel. We make use of vectorisation and parallelisation afforded by multi-core processors. 
\subsection{Applying $\hat{U}_B$}
We need to realise the action
\begin{equation}\label{eq:UBv}
e^{-i\beta\hat{B}}\ket{\psi}.
\end{equation}
The non-trivial structure of the $\hat{B}$ matrix requires a sophisticated method to compute the action of the matrix exponential; a difficult problem withstanding four decades of continual investigation. We implement the Chebyshev expansion method similar to pyCTQW \autocite{izaac_pyctqw:_2015} and depicted in Equation \ref{eq:Cheby}. This method requires only matrix-vector operations to realise Equation \ref{eq:UBv} without storing the final exponentiated matrix at any point. In general,
\begin{equation}\label{eq:Cheby}
e^{tA} = e^{(\lambda_{max}+\lambda_{min})t/2}[J_0(\alpha)\phi_0(\tilde{A}) + 2\sum_{n=1}^{\inf}i^nJ_n(\alpha)\phi_n(\tilde{A})],
\end{equation}
where $\lambda_{max}, \lambda_{min} \in \mathbb{C}$ are eigenvalues of $A$ with maximal and minimal real parts. $\alpha = i(\lambda_{min} - \lambda_{max})t/2$ and $\phi(\tilde{A})$ are the Chebyshev polynomials which computed recursively as
\begin{eqnarray}
\phi_0(\tilde{A}) = I,\\
\phi_1(\tilde{A}) = \tilde{A},\\
\phi_n(\tilde{A}) = 2\tilde{A}\phi_{n-1}(\tilde{A}) - \phi_{n-2}(\tilde{A}).
\end{eqnarray}
Similar to other approximation methods normalisation of $\lambda \in [-1, 1]$ encourages minimal convergence time and therefore we scale our matrix
\begin{eqnarray}
\tilde{A} = \frac{2A - (\lambda_{max} + \lambda_{min})I}{\lambda_{max} - \lambda_{min}}.
\end{eqnarray}
The use of Bessel function zeros as coefficients means that $J_n(\alpha) \approx 0$ when $n > |\alpha|$ and therefore convergence occurs after $|\alpha| \propto t$ terms. Similar to Izaac and Wang \autocite{izaac_pyctqw:_2015} we terminate after the condition
\begin{equation}
|2J_n(\alpha) \leq \epsilon|,
\end{equation}
where $\epsilon$ is chosen to be $10^{-18}$.\newline
However, without knowledge of the maximal and minimal eigenvalues we would be required to solve for these values, a time-consuming and laborious calculation for large matrices. Since we are computing the matrix exponential for only a particular type of matrix (non-negative, Hermitian and symmetric) we find that
\begin{equation}\label{eq:MinMaxEigen}
\lambda_{min, max} = \pm q.
\end{equation}
A derivation of Equation \ref{eq:MinMaxEigen} is found in Appendix \ref{ap:Proofs}. Furthermore, the only downside to over-estimating the range of these eigenvalues is computation time, not accuracy. In the case of a restricted $\hat{B}$ matrix the symmetric, Hermitian and non-negative properties hold and can only have fewer non-zero elements than a fully connected, canonical $\hat{B}$ hence our method is appropriate for simulation of the QAOA.\newline
When considering the decomposition of $\hat{B}$ among multiple processes we consider three options. In all cases the total work required is $O(n^2/p)$ since all values must be considered in the multiplication.
\subsubsection{Column Distribution}
Each process builds and operates upon a section of columns of our overall matrix governed by Equation \ref{eq:decomp}. This produces optimal load balancing between all processors due to the symmetry of $\hat{B}$. A matrix vector multiplication sees each process operate upon a subset of the state-vector but produces a full $2^q$ length vector. These need to be combined across all processes resulting in $\mathcal{O}(n^2/p + n + p)$ work. The difficulty here is that after each multiplication each process contains a full state-vector with $1/p$ of the total solution which must be reduced and scattered to all processes in order to continue work. If we consider $\alpha$ as the time for a single scalar operation, $\lambda$ as the communication latency and each complex value is $32$ bytes long and $\beta$ as the buffer length of a message the total execution time is given by 
\begin{equation}\label{eq:ColScale}
\Theta(\alpha n \lceil\frac{n}{p}\rceil + (p - 1)(\lambda + \frac{32}{p\beta})).
\end{equation}
Moreover the scalability of this scheme is $n^2p$. Since $n = 2^q$ this scheme scales poorly.
\subsubsection{Row Distribution}
Similarly, we could decompose across rows which again does not alter the distribution scheme of $\hat{C}$. In this scheme, after a matrix-vector multiplication each process will hold a separate section of the resultant vector which again needs redistributing to all processes through an all-gather operation. The total execution time is given by
\begin{equation}\label{eq:RowScale}
\Theta(\alpha n \lceil\frac{n}{p}\rceil + \lambda\lceil log p\rceil + \frac{32n}{\beta}).
\end{equation}
We see a similar scalability function of $n^2p$ which is again very infeasible but slightly better. 
\subsubsection{Checker-board Distribution}
We could alternatively decompose $\hat{B}$ across both rows and columns forming a grid of processors. The input vector is distributed across the first row of processes and each subsection is copied down each column. After a vector multiplication operation the result is then reduced across each row. A process in each row will contain a separate section of the final state vector prepared for a $\hat{U}_C$ operation. This scheme exhibits a total execution time of 
\begin{equation}\label{eq:CheckerScale}
\Theta(\frac{\alpha n^2}{p^2} + \lambda\frac{32nlog p^2}{\sqrt{p^2}\beta}),
\end{equation}
if we consider the squared number of processes required to implement this method. More machines are required however the scalability is much better and is given by $n^2log^2p^2$. 
\section{Measurement}
After our final QAOA state is prepared $\ket{\psi}$ results must be communicated back to the root process. The na\"{i}ve method would involve each processor sending its final discrete component of the state-vector to the root process. This would require $\mathcal{O}(2^q)$ communication at the root process. We formulate three schemes for final state measurement.\newline
The expectation value can be computed in an embarrassingly parallel fashion on each process and reduced at the root process. This scheme requires $\mathcal{O}(P)$ communication of a single value. \newline
Aggregating over the state-space and cost-function allows us to reduce the entire state-vector which grows $\mathcal{O}(2^q)$ to a smaller distribution of unique cost-function values which is problem-dependent in size. By definition, the QAOA demands a polynomial cost function however and for graph simialrity this is $\mathcal{O}(v^2)$ which is significantly more feasible.
\section{Optimisation}
The well-established nlopt non-linear optimisation library \autocite{johnson_nlopt_2011} to implement local and global derivative-free optimisation schemes. This library implements useful features such as run-time constraints. Notably IBM's Qiskit \autocite{ibm_qiskit_nodate} utilises this package and as such Qolab matches the state of the art in functionality. Johnson \autocite{johnson_nlopt_2011} provides an explanation of each available method. Initialising multiple processors to complete independent optimisation on the same problem is a trivial scheme providing no increase in problem size but increases for evaluation.
\chapter{Results}
We simultaneously investigate our mapping of graph similarity to the QAOA providing performance analysis of our Qolab package. Source code is available in Appendix \ref{ap:Qolab}, the full set of original data-files, aggregate data and code to generate all plots is available in Appendix \ref{ap:Data}. The exponential complexity of quantum state-vector simulation places a unique strain on performance analysis since we can by definition do no better than $\mathcal{O}(2^q)$. We provide correctness investigation of the QAOA up to $19$ qubits and performance analysis for $22$ qubits. We investigate eight different classical optimisation schemes.
\section{Definitions}
We consider a number of performance metrics defined below. We cannot present a standard approximation factor for our problem since the optimal value of our cost function is zero. We provide alternate correctness measures which illuminate the same trends.
\begin{itemize}
	\item Number of Evaluations - We consider an evaluation the generation of a single $\ket{\vec{\gamma}, \vec{\beta}}$ state. This measure provides a rough estimate of how many logical (aptly sampled) quantum states are required. Lower is better.
	\item Sample Error - We use the expectation value of our system until the last iteration which samples the resulting state space $V^2$ times. The best value observed is compared to the known optimal as a standard approximation ratio. Higher is better.
	\item Expectation Error - The difference between the optimal solution and our expectation value normalised over the minimal possible value to account for changing graph sizes. Lower is better. A value of zero indicates certain measurement of the optimal solution 
	\item Classical Comparison - We compare the final expectation value of our state with the corresponding expectation value of random solution selection. A negative value indicates better likelihood of a better solution from the QAOA. Higher is better
	\item Expectation Improvement - The net improvement in expectation value from the state generated by initial parameters. Higher is better.
\end{itemize}
\subsection{Methodology}
We present results from three datasets:
\begin{itemize}
	\item An initial set using an old cost function containing $16, 190$ trials. This alternate cost function sets the optimal value to zero, sub-optimal solutions to negative values and pads infeasible solutions with a maximal penalty of $-V^2$.
	\item A set of directed graph results containing $10, 800$ trials.
	\item A set of undirected graph results containing $2500$ trials.
\end{itemize}
In all trials we consider all deformation methods described in Section \ref{sec:TestCase} equally and are aggregated together to provide an overall impression of QAOA performance.\newline
The classical parameter optimisation component of the QAOA is critical to correctness and efficiency. Guerreschi and Smelyanskiy \autocite{guerreschi_practical_2017} provide the most in-depth analysis to date, we investigate eight different optimisation schemes. We provide extensive plotting of our results in Appendix \ref{ap:Data}. We explicitly discuss the Nelder-Mead \cite{nelder_simplex_1965}, Subplex \cite{rowan_functional_1990}, BOBYQA \cite{powell_direct_1998}, Multi-Level Single-Linkage (MLSL) \cite{jones_lipschitzian_1993} and Dividing Rectangles (DIRECT) \cite{rinnooy_kan_stochastic_1987} methods. The Subplex algorithm is a variation on the legendary Nelder-Mead simplex algorithm \cite{nelder_simplex_1965} designed to target noisy function spaces. Simplex-based algorithms maintain a simplex of $n + 1$ points to optimise $n$ parameters. The BOBYQA algorithm estimates trust regions by forming quadratic models of the cost function. MLSL maintains a variety of local optimisations starting from a series of random points using heuristics to avoid repeated searching of local optima. The DIRECT algorithm is a deterministic global search algorithm based on dividing the search space into increasingly smaller hyper-rectangles.
\section{Alternate Cost Function}
We see the effect of our 'tail' values most prominently when considering a slightly different cost function. Mapping optimal solutions to zero and penalising missing edges to negative values and mapping non-solution bit-strings to a minimal value. However, since these values now contribute to our expectation value, the overall performance is diminished depicted in Figure \ref{fig:costFunctions}. Our final cost function performs significantly better where error grows significantly slower for all amounts of decomposition. This trend holds across all optimisation methods tested.
\begin{figure}
	\begin{tabular}{cc}
		\includegraphics[width=0.5\textwidth]{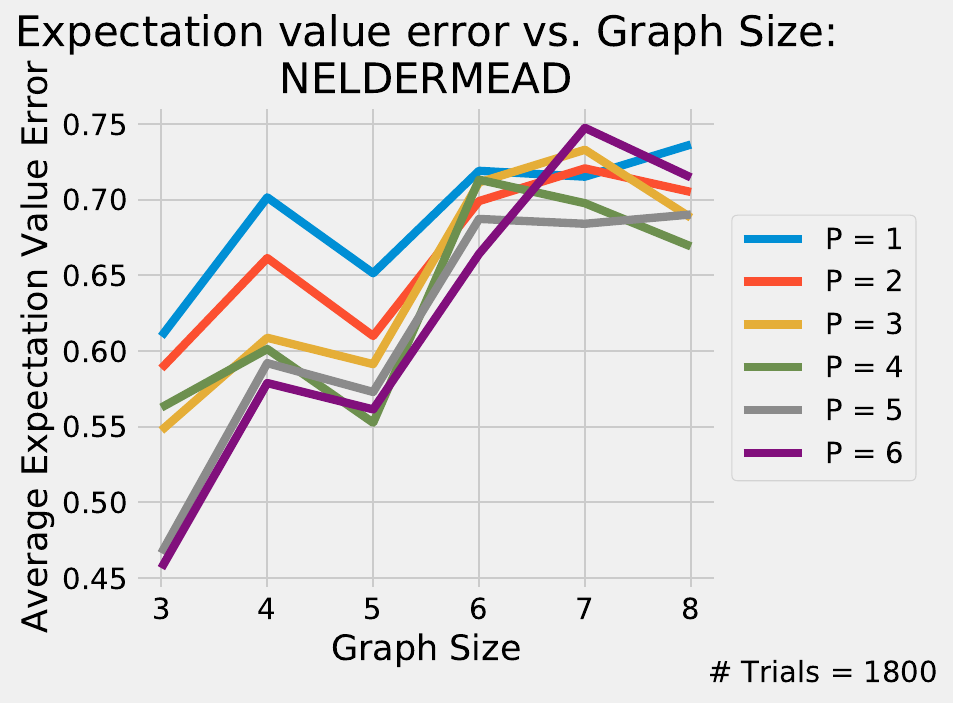} &   \includegraphics[width=0.5\textwidth]{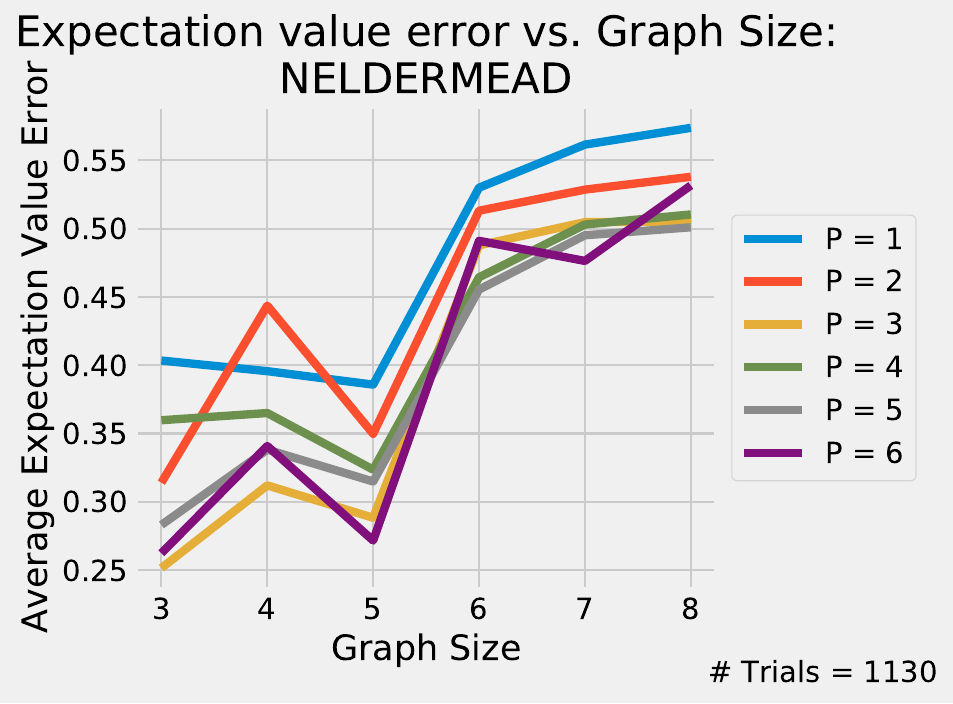} \\
		(a) Alternate Cost Function & (b) Final Cost Function \\[6pt]
	\end{tabular}
	\label{fig:costFunctions}
	\caption{Expectation error for differing cost functions}
\end{figure}
\section{Increased Decomposition}
Figure \ref{fig:finalNELD} depicts performance metrics for the Nelder-Mead algorithm. We see excellent performance for this method matching current literature \cite{guerreschi_practical_2017}. In Figure \ref{fig:finalBOB} we see the effect increased decomposition makes despite inferior end-results. We see that increasing QAOA decomposition results in an almost monotonic improvement in both sampled and expected solution quality at the cost of nearly double the number of function evaluations required for termination. This generalises to most optimisation methods tested. This is intuitive since we can always zero out parameters to provide the performance of a coarser QAOA scheme. The QAOA suffers from dimensionality issues since each increase in $p$ exponentially increases the number of possible values requiring more optimisation iterations to make use of these additional parameters. This is most clearly seen in Figure 4.4(c).
\begin{center}
	\begin{figure}
		\centering
		\begin{tabular}{cc}
			\includegraphics[width=0.5\textwidth]{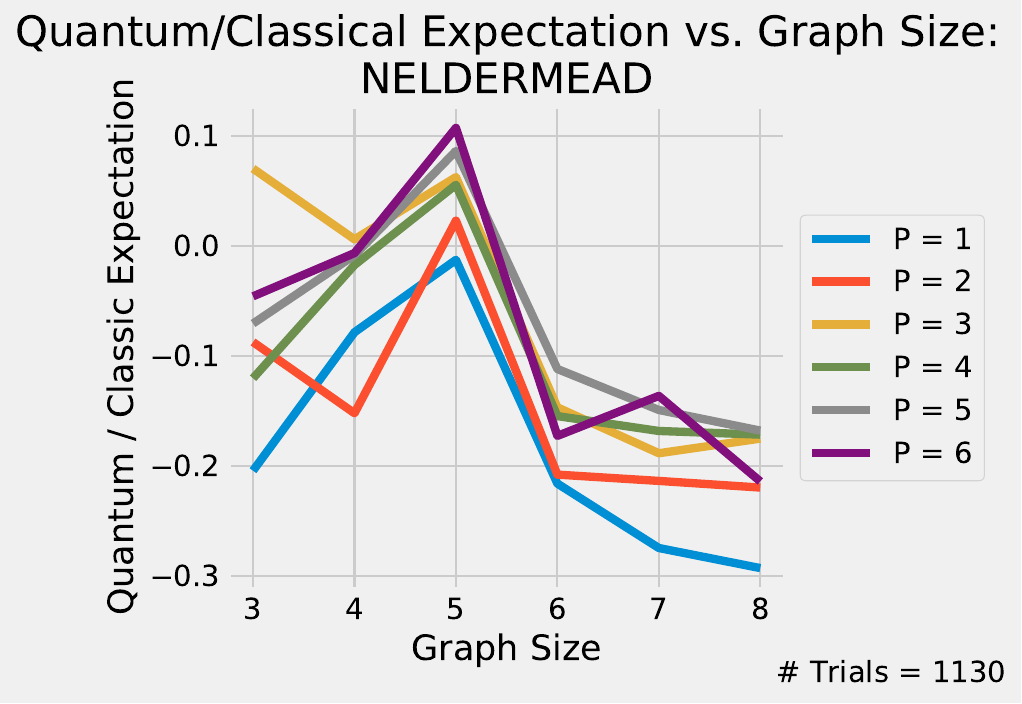} &   \includegraphics[width=0.5\textwidth]{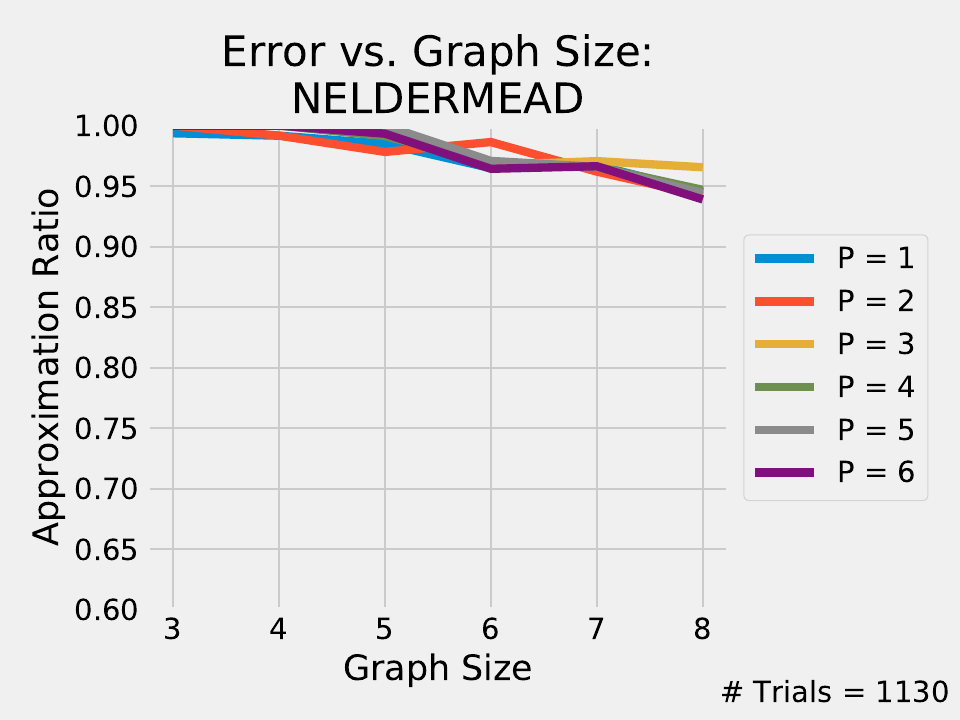} \\
			(a) Expectation value comparison & (b) Solution error \\[6pt]
			\includegraphics[width=0.5\textwidth]{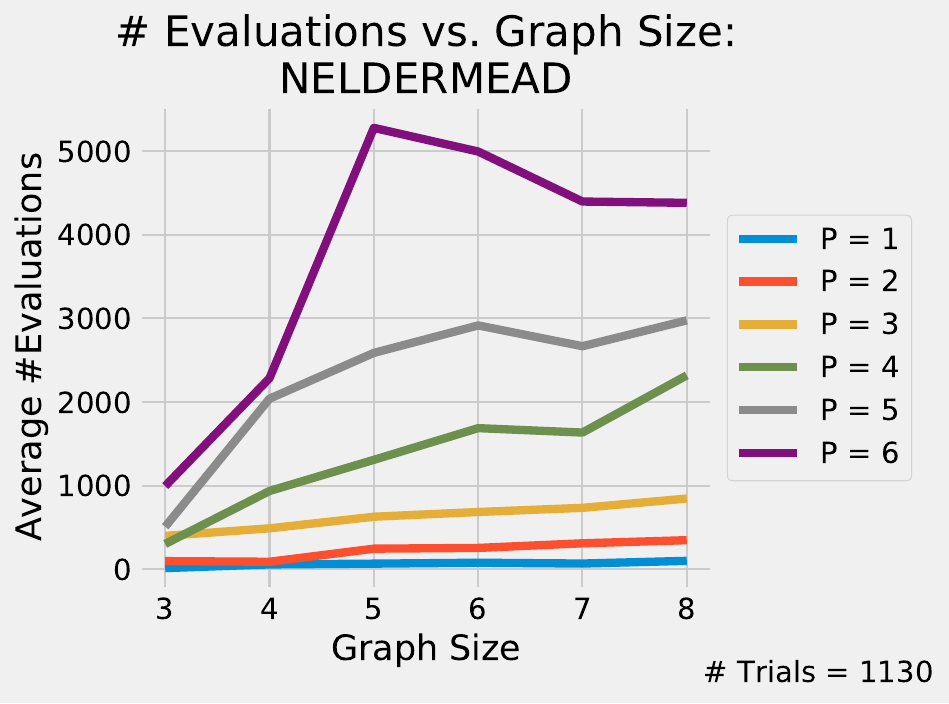} &   \includegraphics[width=0.5\textwidth]{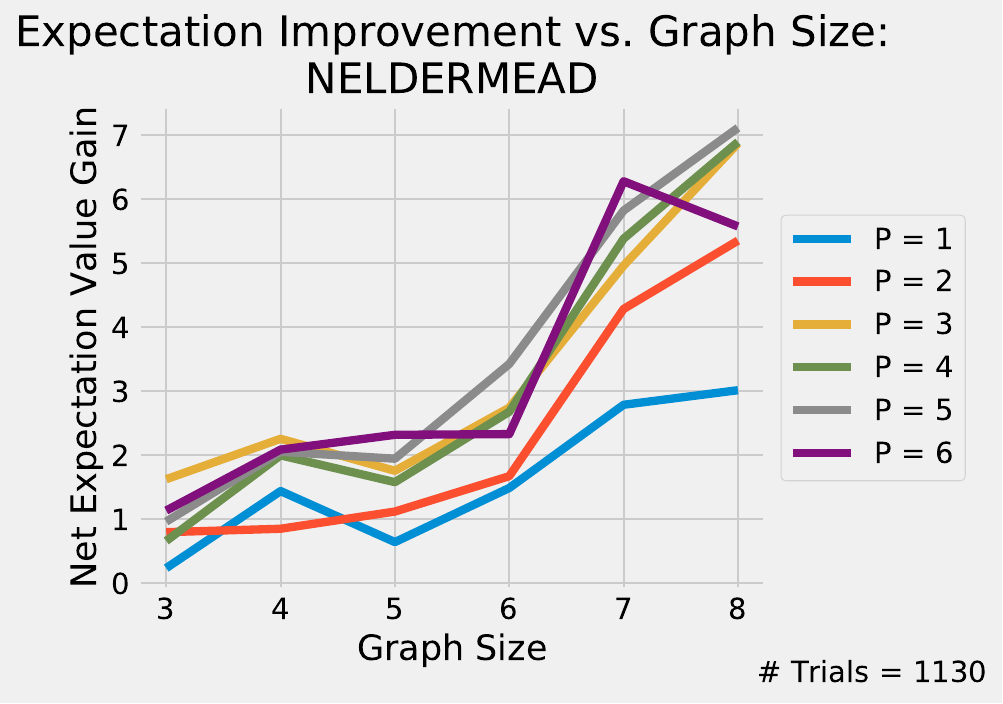} \\
			(c) Function evaluations required & (d) Improvement\\[6pt]
		\end{tabular}
		\label{fig:finalNELD}
		\caption{Final performance for the Nelder-Mead algorithm (directed graphs)}
	\end{figure}
\end{center}
\begin{center}
	\begin{figure}
		\centering
		\begin{tabular}{cc}
			\includegraphics[width=0.5\textwidth]{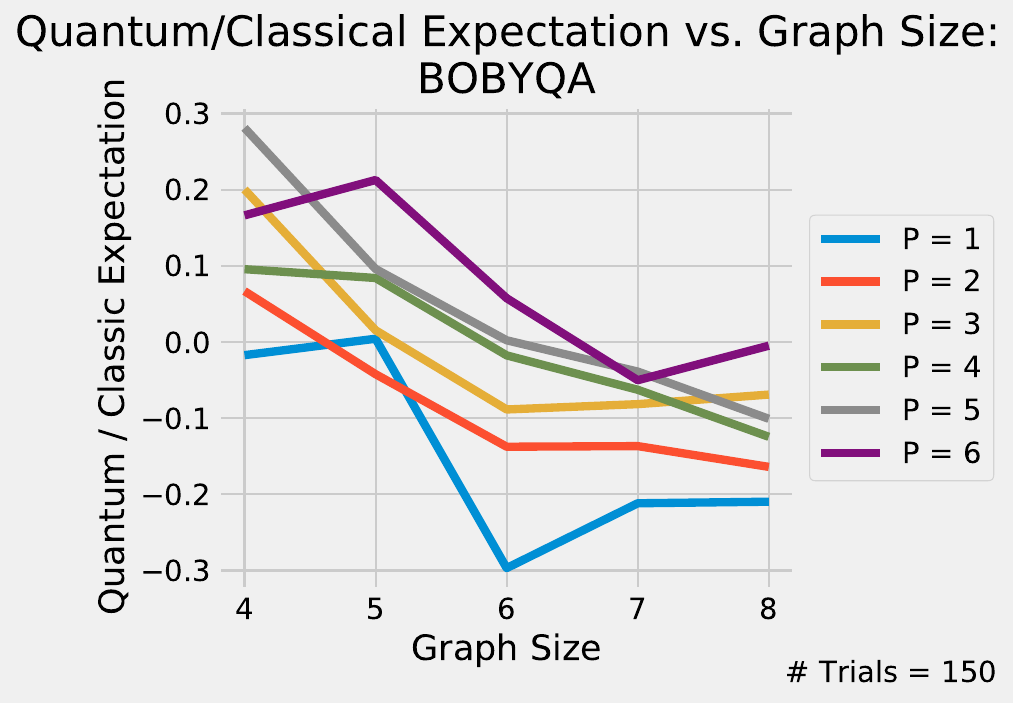} &   \includegraphics[width=0.5\textwidth]{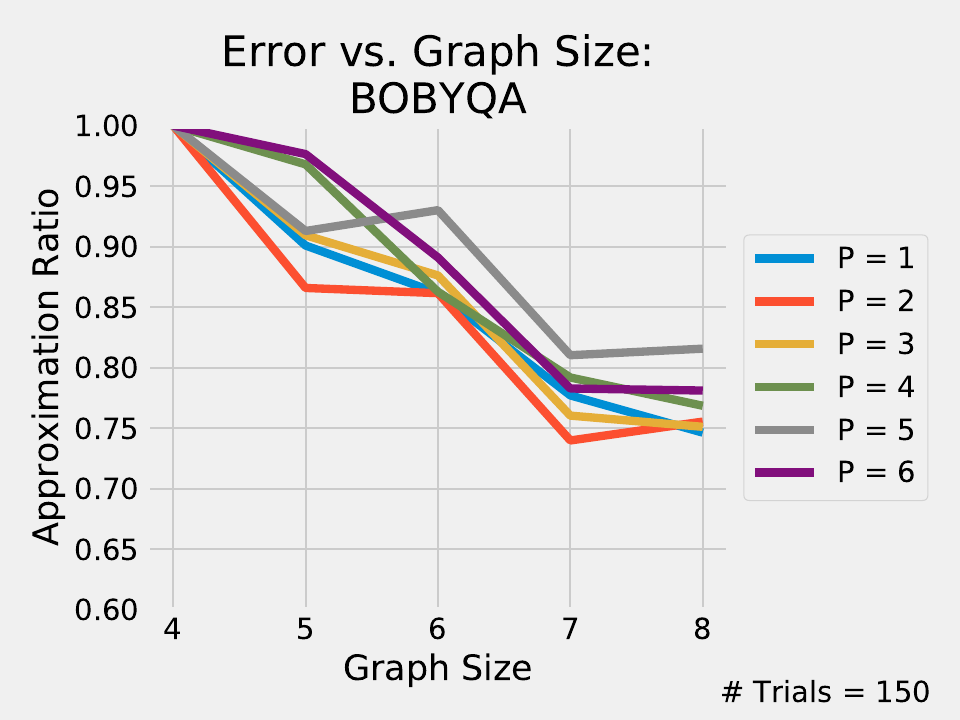} \\
			(a) Expectation value comparison & (b) Solution error \\[6pt]
			\includegraphics[width=0.5\textwidth]{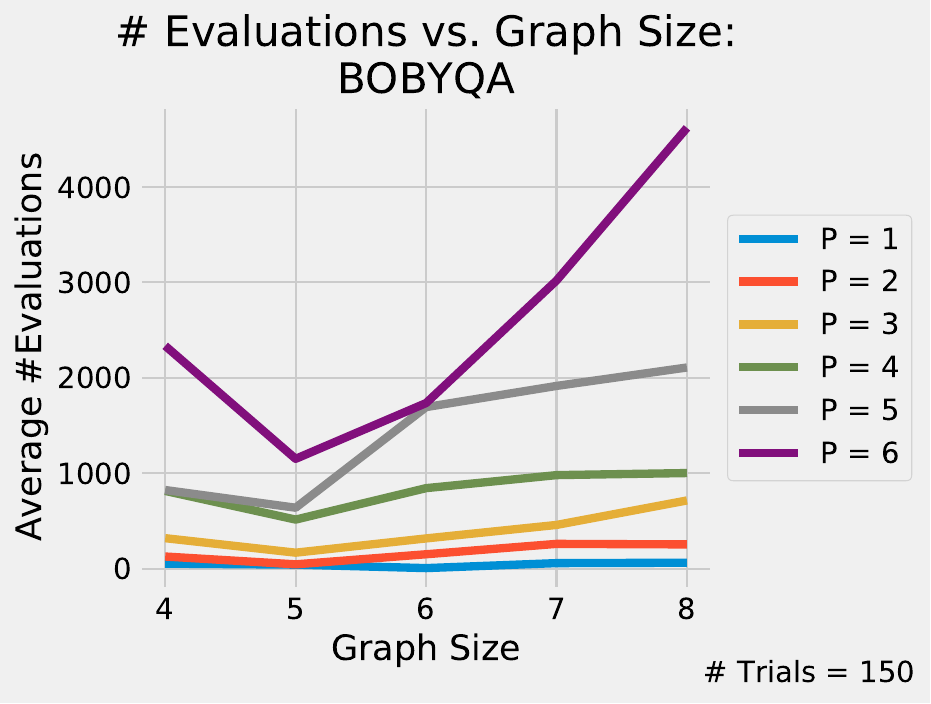} &   \includegraphics[width=0.5\textwidth]{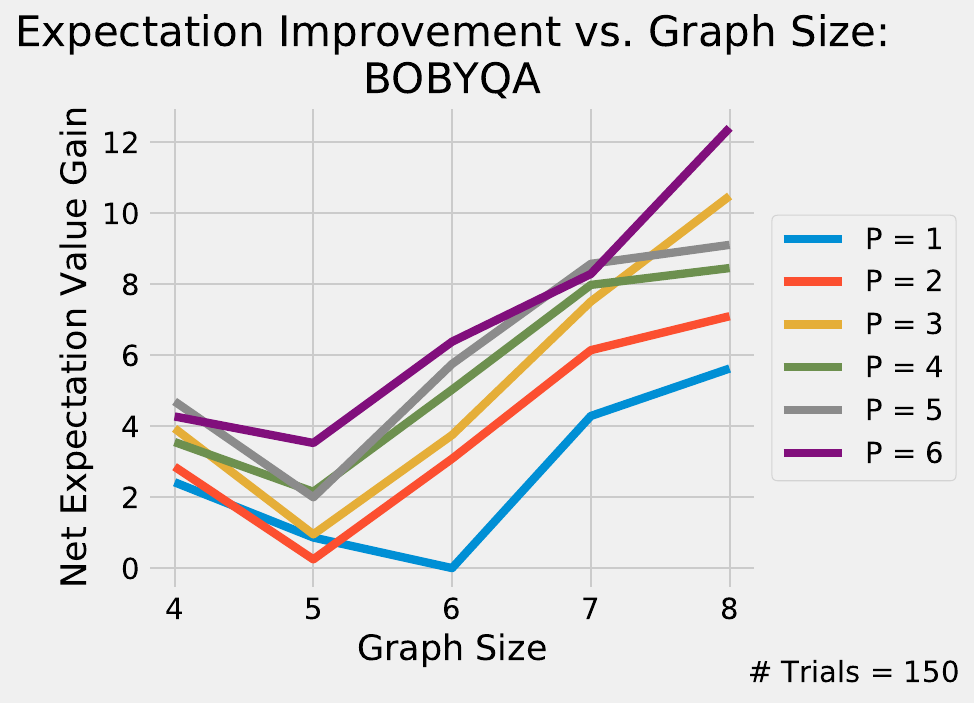} \\
			(c) Function evaluations required & (d) Improvement\\[6pt]
		\end{tabular}
		\label{fig:finalBOB}
		\caption{Final performance for the BOBYQA algorithm (directed graphs)}
	\end{figure}
\end{center}
\newpage
\section{Optimisation Methods}
\subsection{Correctness}
The global nature of both the MLSL and DIRECT algorithms results in generally superior sampled and expected solutions over local methods. The MLSL algorithm achieves this by exhausting the maximal number of function evaluations defined as
\begin{equation}\label{eq:Scaling}
MAX = S \times P \times Graph Size,
\end{equation}
where $S$ is a scaling parameter (we nominally choose $200$, the standard value used by Scipy. \cite{jones_scipy:_2001}).
The DIRECT algorithm terminates significantly earlier than other methods with a significant degradation in performance. For all algorithms tested the effect of infeasible solutions results in a lower expectation value versus classical sampling in the larger test cases indicating either more optimisation time is required or a more nuanced problem encoding. The cost function for graph similarity can span a range of $V^2$ values and produces a cost-function landscape containing many local optima. The derivative and non-linear nature of the optimisation schemes tested results in the scheme terminating at said optima.

\begin{center}
	\begin{figure}
		\centering
		\begin{tabular}{cc}
			\includegraphics[width=0.5\textwidth]{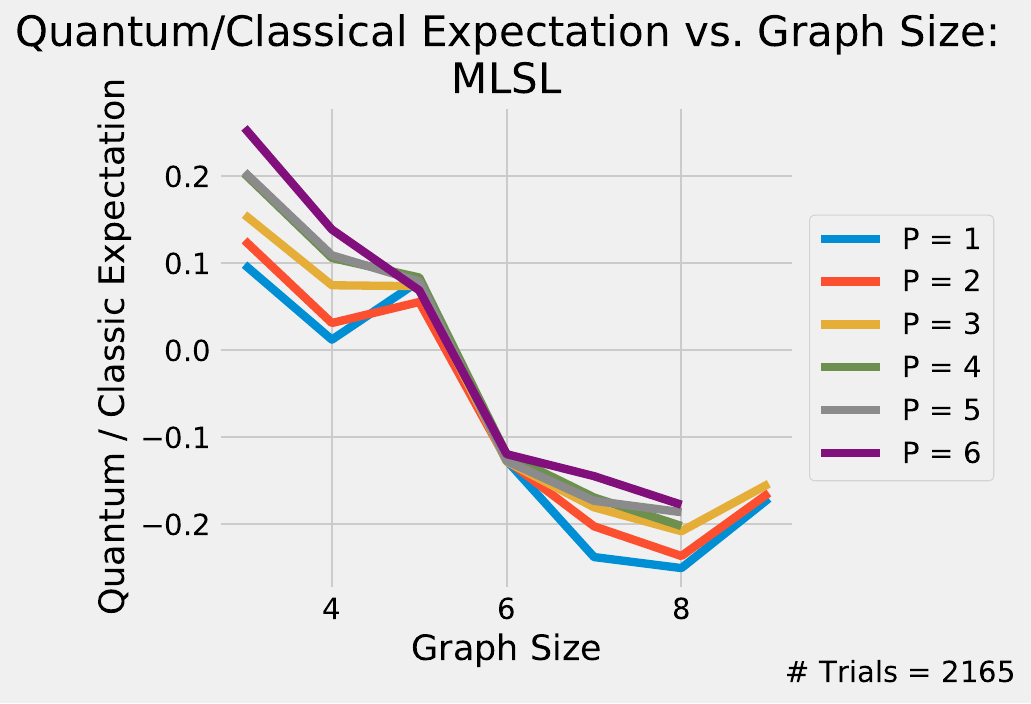} &   \includegraphics[width=0.5\textwidth]{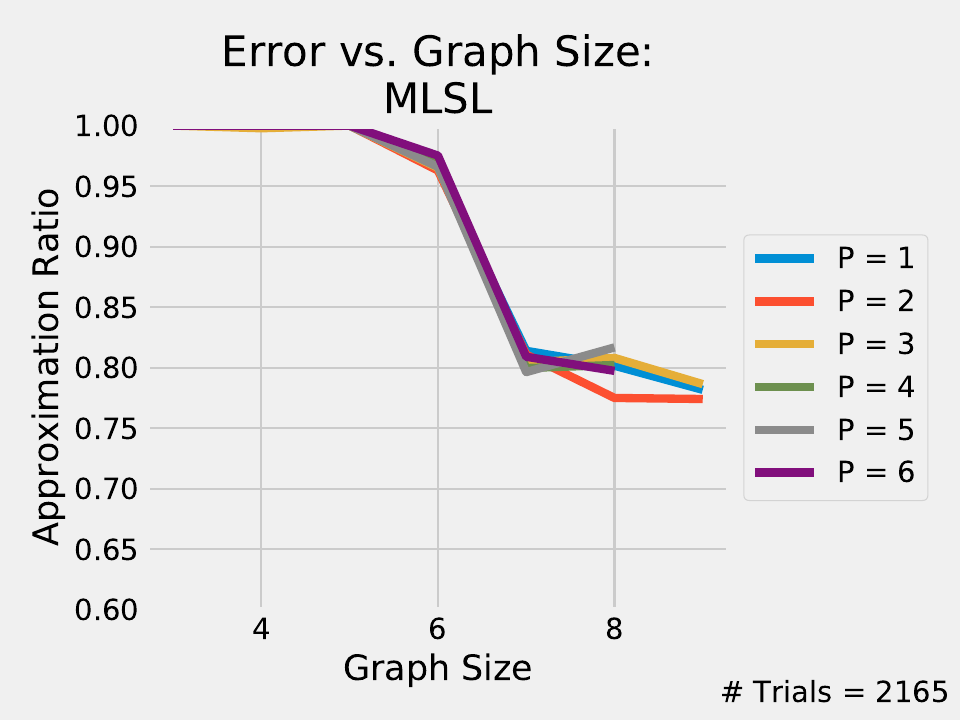} \\
			(a) Expectation value comparison & (b) Solution error \\[6pt]
			\includegraphics[width=0.5\textwidth]{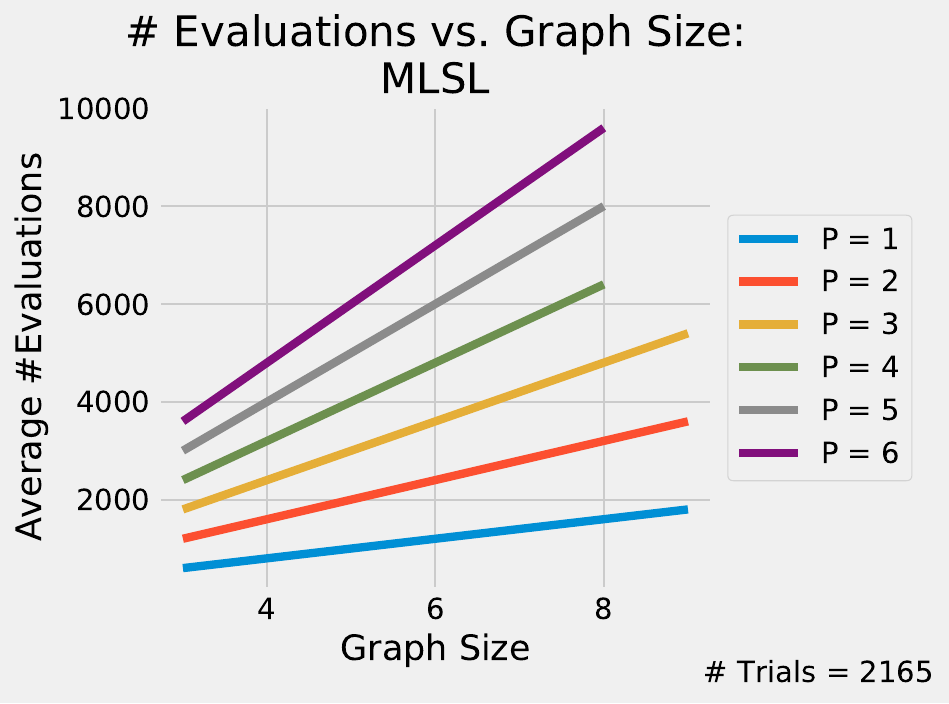} &   \includegraphics[width=0.5\textwidth]{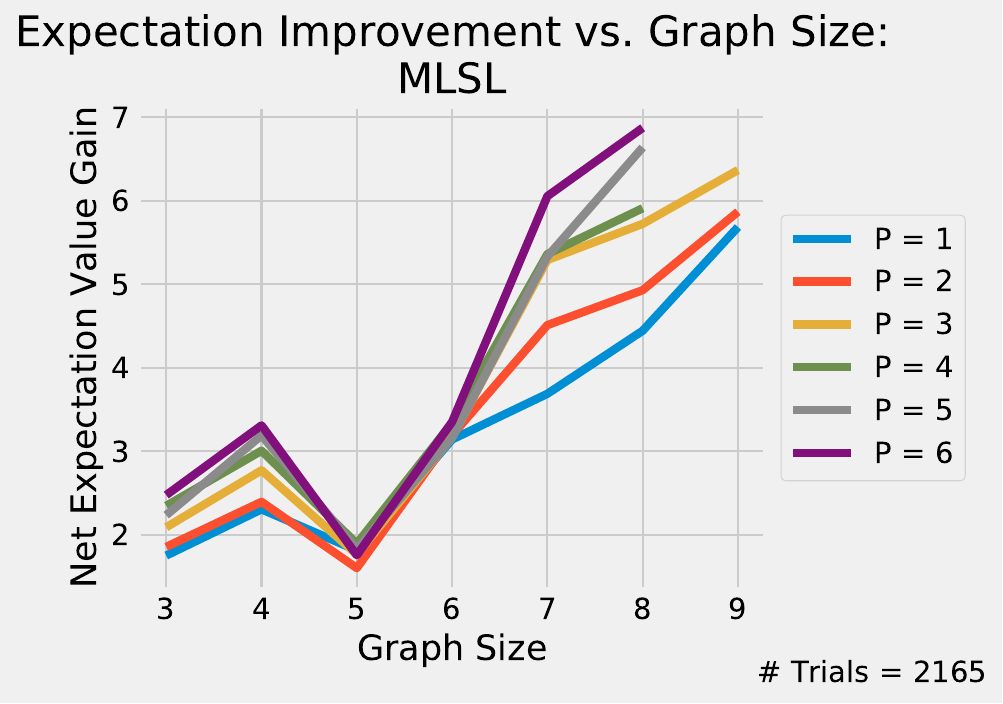} \\
			(c) Function evaluations required & (d) Improvement\\[6pt]
		\end{tabular}
		\label{fig:finalMLSL}
		\caption{Final performance for the MLSL algorithm (directed graphs)}
	\end{figure}
\end{center}
\begin{center}
	\begin{figure}
		\centering
		\begin{tabular}{cc}
			\includegraphics[width=0.5\textwidth]{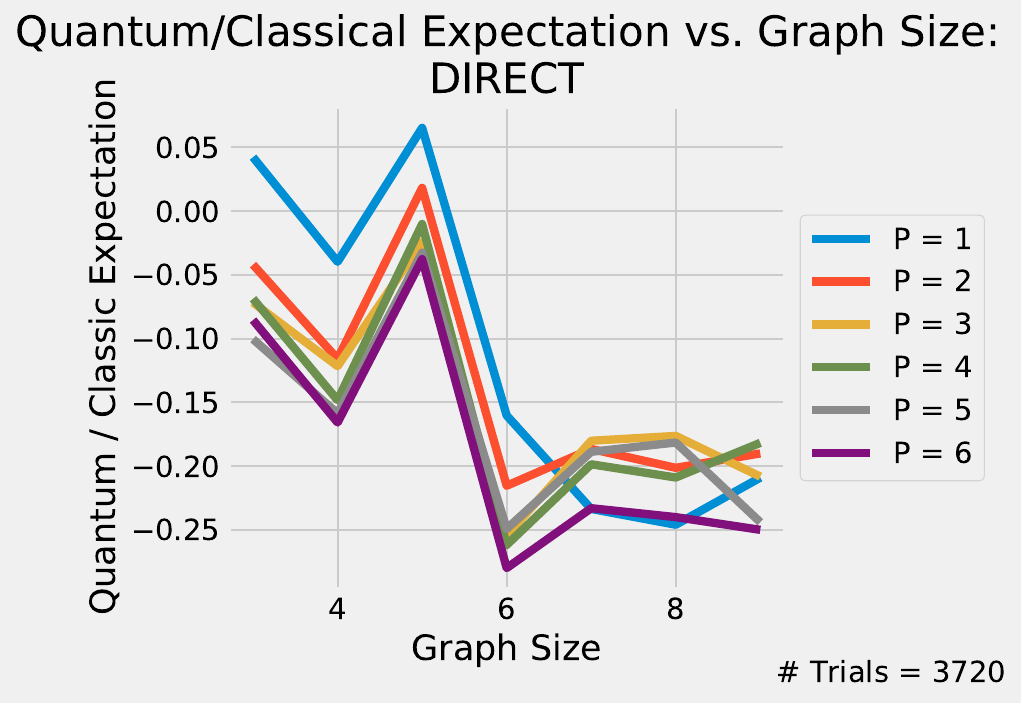} &   \includegraphics[width=0.5\textwidth]{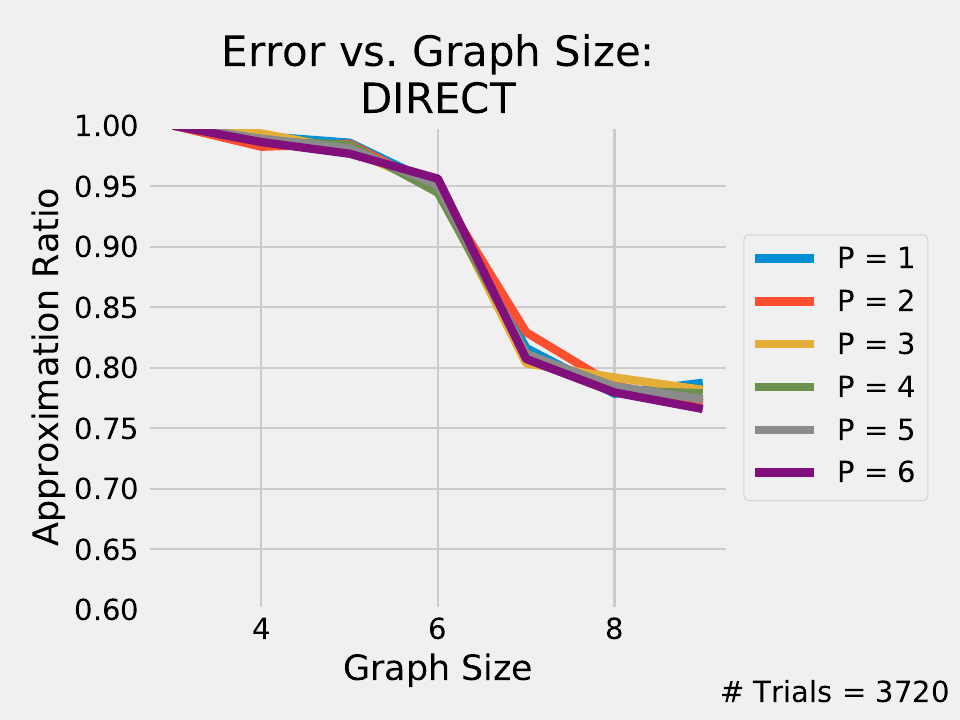} \\
			(a) Expectation value comparison & (b) Solution error \\[6pt]
			\includegraphics[width=0.5\textwidth]{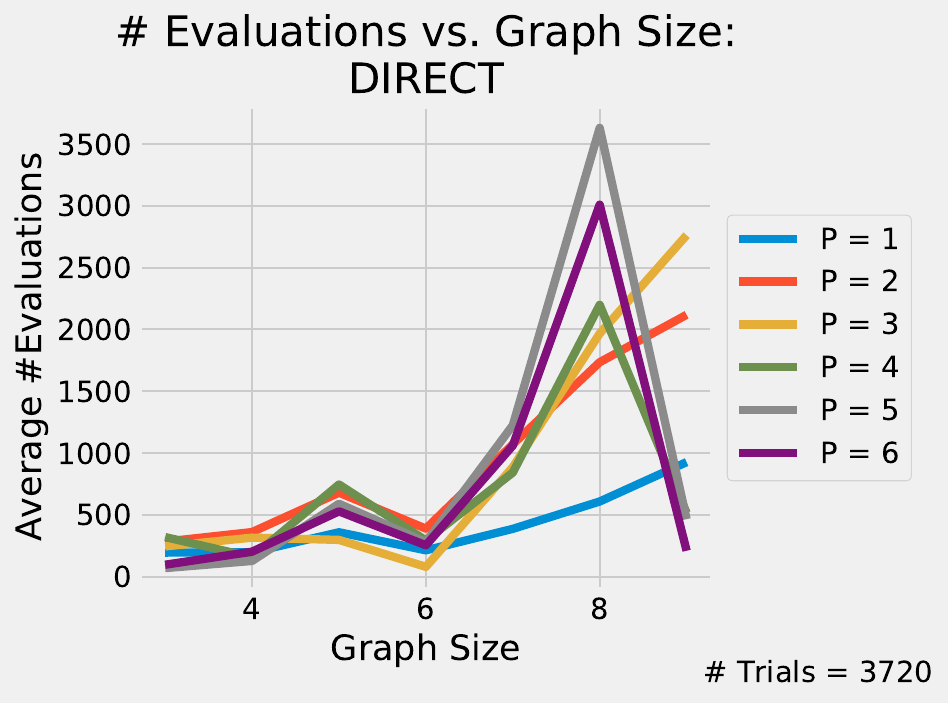} &   \includegraphics[width=0.5\textwidth]{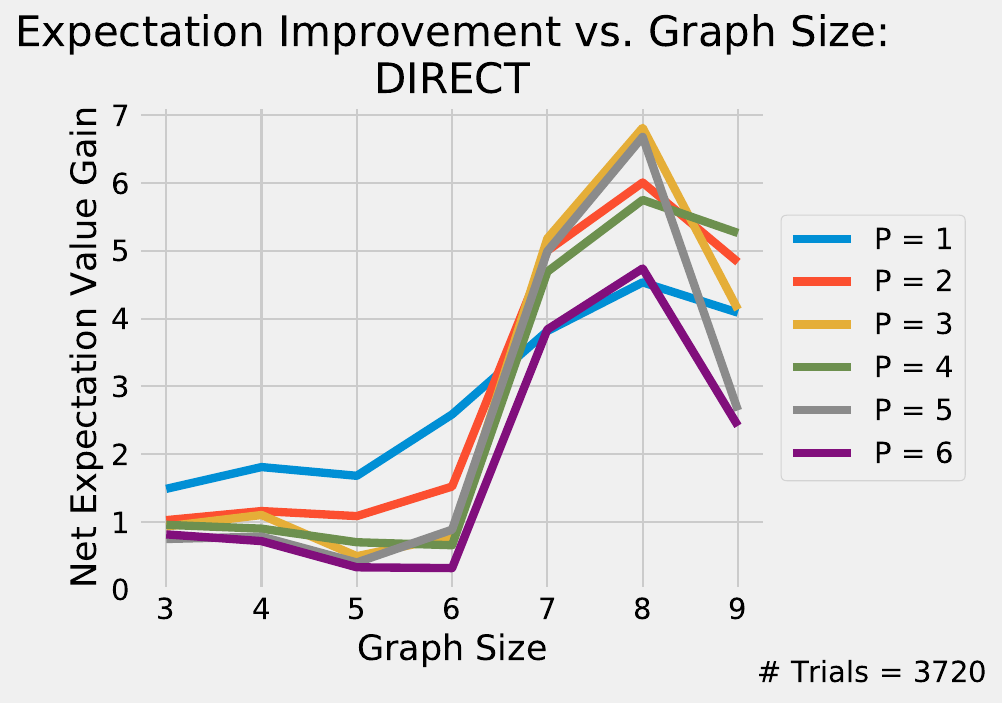} \\
			(c) Function evaluations required & (d) Improvement\\[6pt]
		\end{tabular}
		\label{fig:finalDIRECT}
		\caption{Final performance for the DIRECT algorithm (directed graphs)}
	\end{figure}
\end{center}
\subsection{Efficiency}
There is no clearly superior optimisation method for our mapping of graph similarity to the QAOA. Generally, global methods provide more correct solutions at a cost of vastly more function evaluations whereas local algorithms typically terminate with fewer iterations but produce poorer results as one would expect. 
\section{Directed vs. Undirected Graphs}
\begin{center}
	\begin{figure}
		\centering
		\begin{tabular}{cc}
			\includegraphics[width=0.5\textwidth]{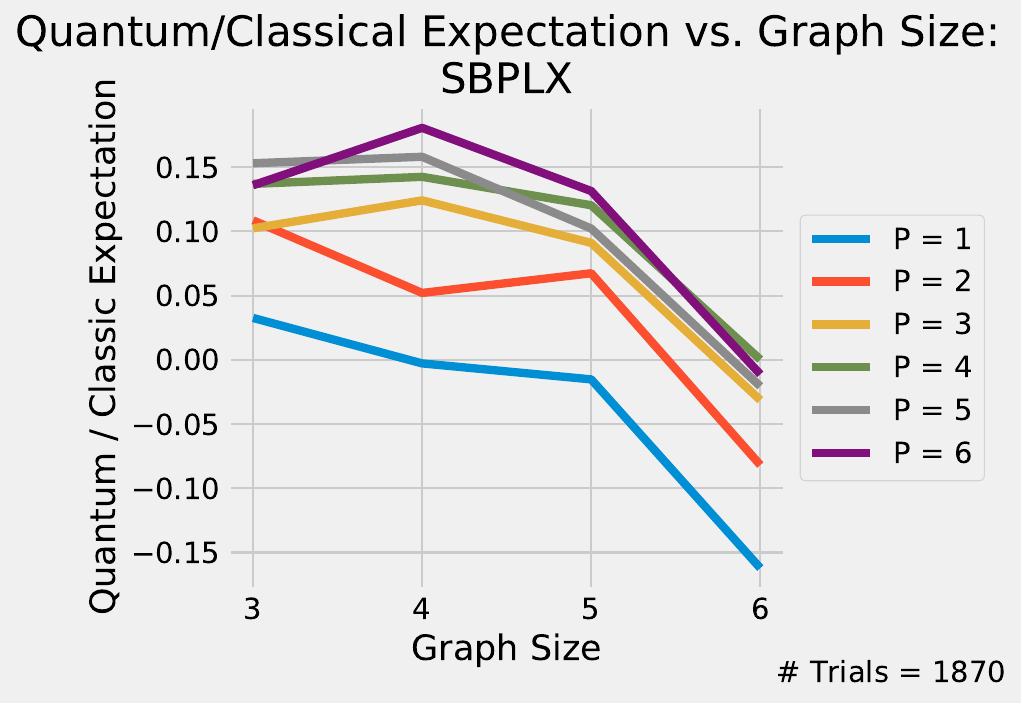} &   \includegraphics[width=0.5\textwidth]{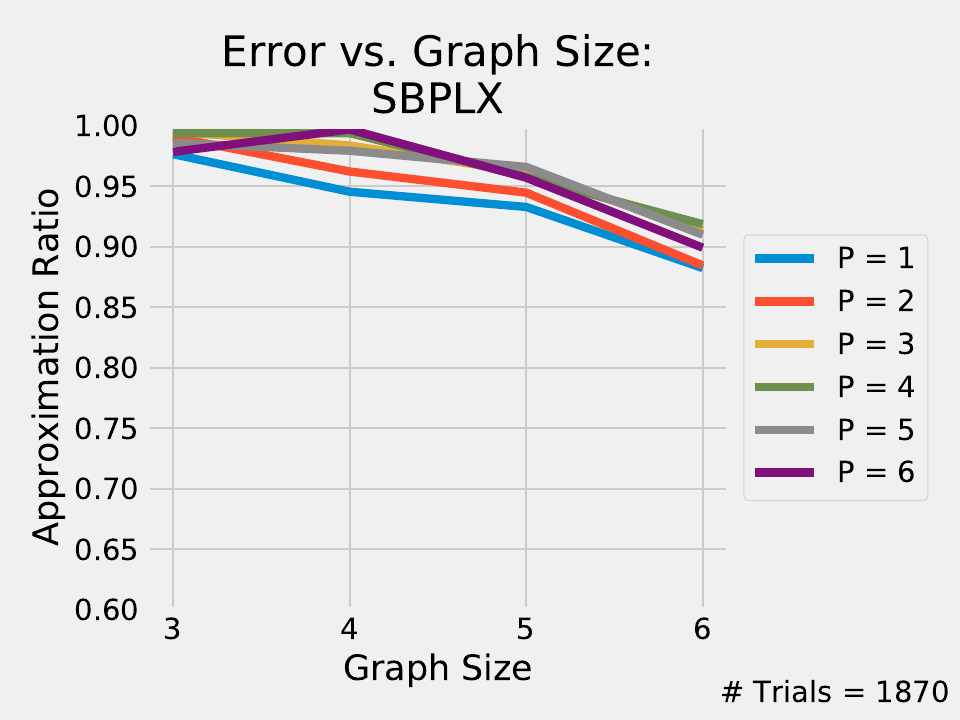} \\
			(a) Expectation value comparison & (b) Solution error \\[6pt]
			\includegraphics[width=0.5\textwidth]{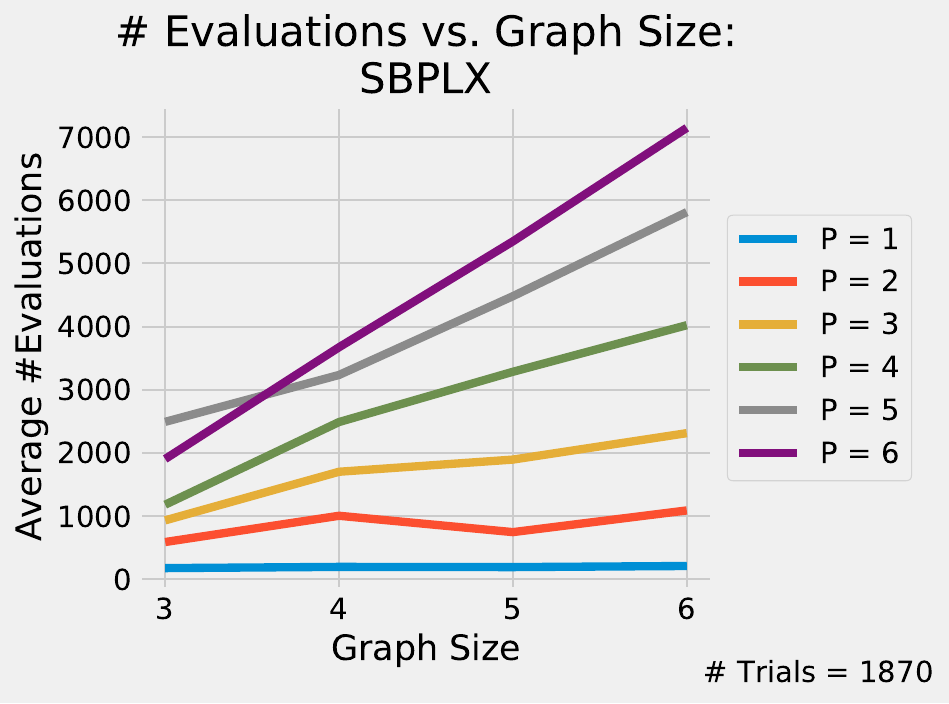} &   \includegraphics[width=0.5\textwidth]{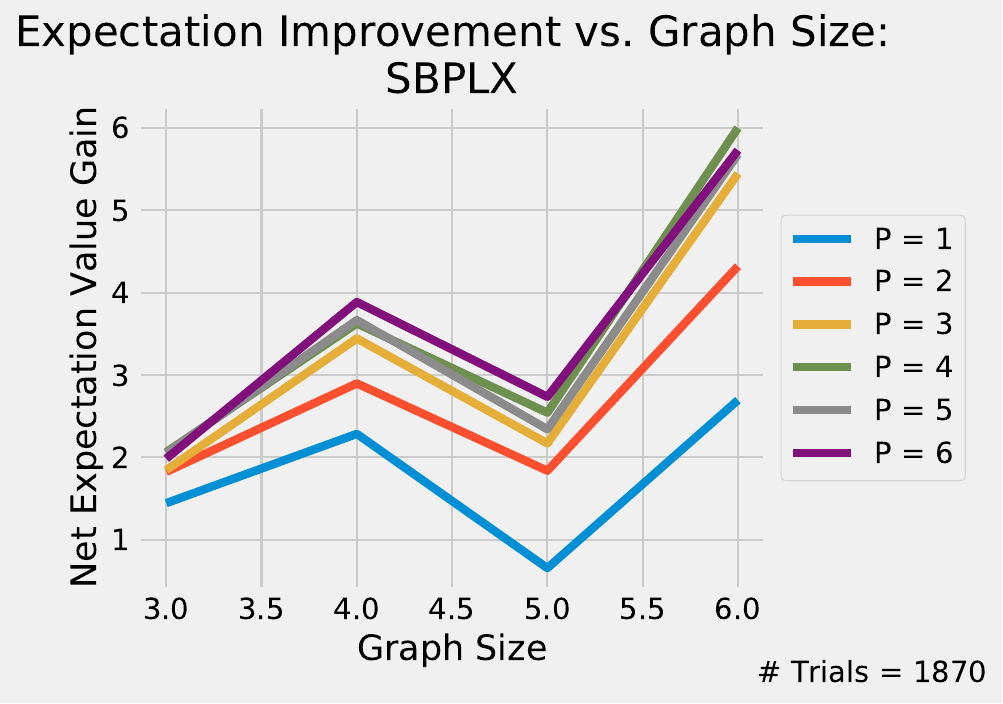} \\
			(c) Function evaluations required & (d) Improvement\\[6pt]
		\end{tabular}
		\label{fig:undirectSBPLX}
		\caption{Final performance for the Subplex algorithm (undirected graphs)}
	\end{figure}
\end{center}
We compare the performance of the Subplex algorithm between directed and undirected graphs in Figure \ref{fig:undirectSBPLX}. We see superior solution quality in the undirected case. This is expected since the mappings of node labellings in the undirected case generates a cost function landscape spanning fewer unique values resulting in a 'smoother' cost-function landscape. 
\section{Simulation Performance}
\begin{figure}\label{fig:Threading}
	\begin{tabular}{cc}
		\includegraphics[width=0.5\textwidth]{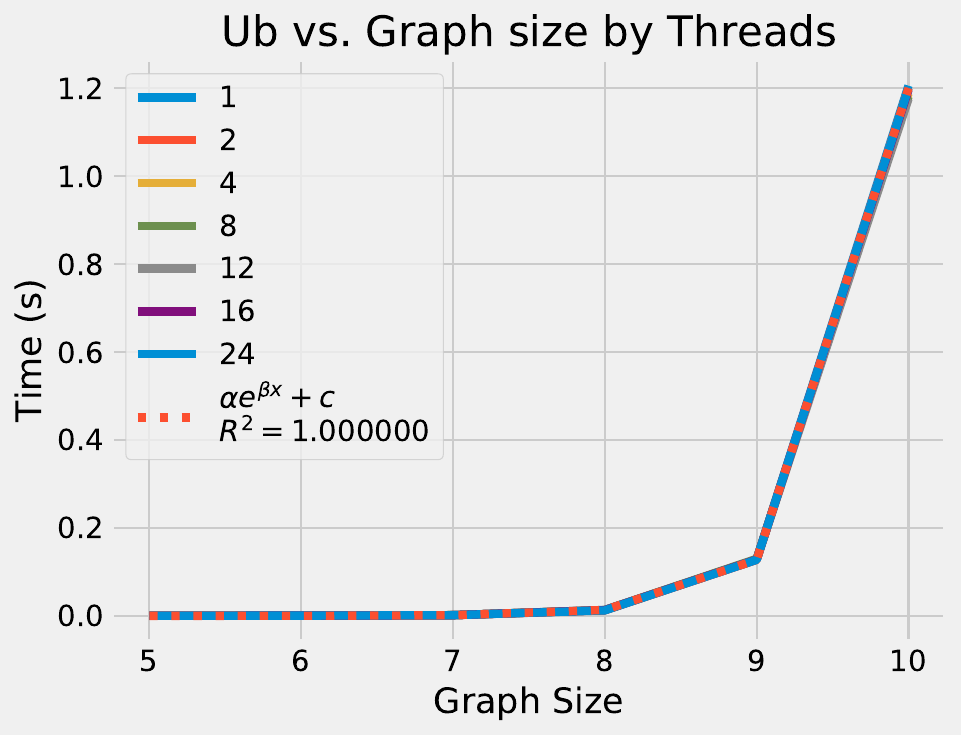} &   \includegraphics[width=0.5\textwidth]{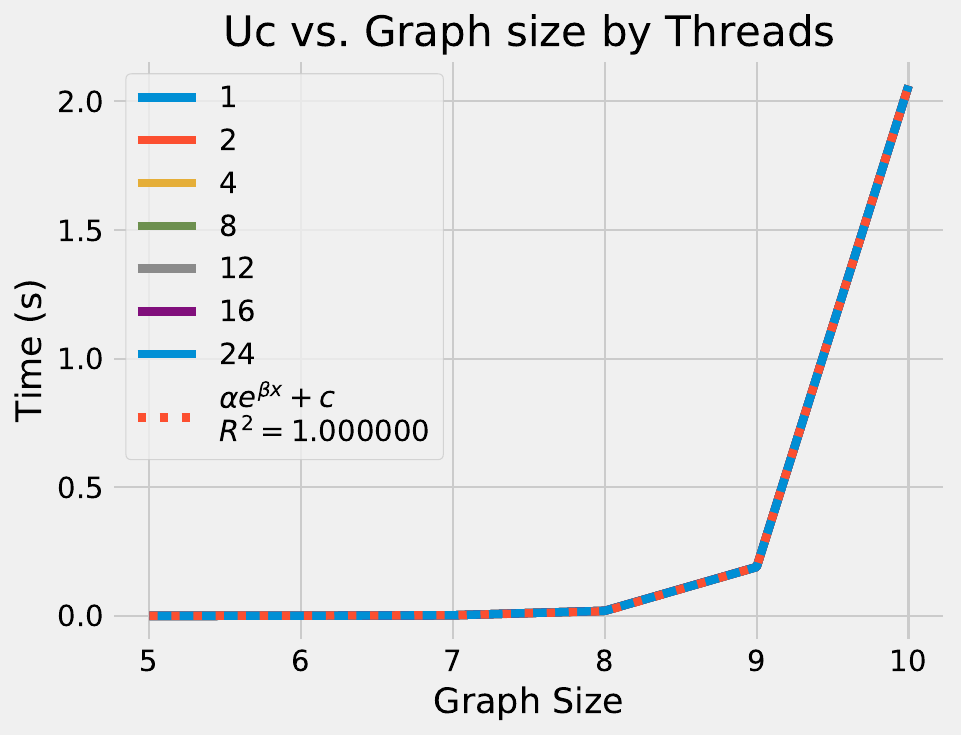} \\
		(a) $\hat{B}$ generation & (b) $\hat{C}$ generation \\[6pt]
		\includegraphics[width=0.5\textwidth]{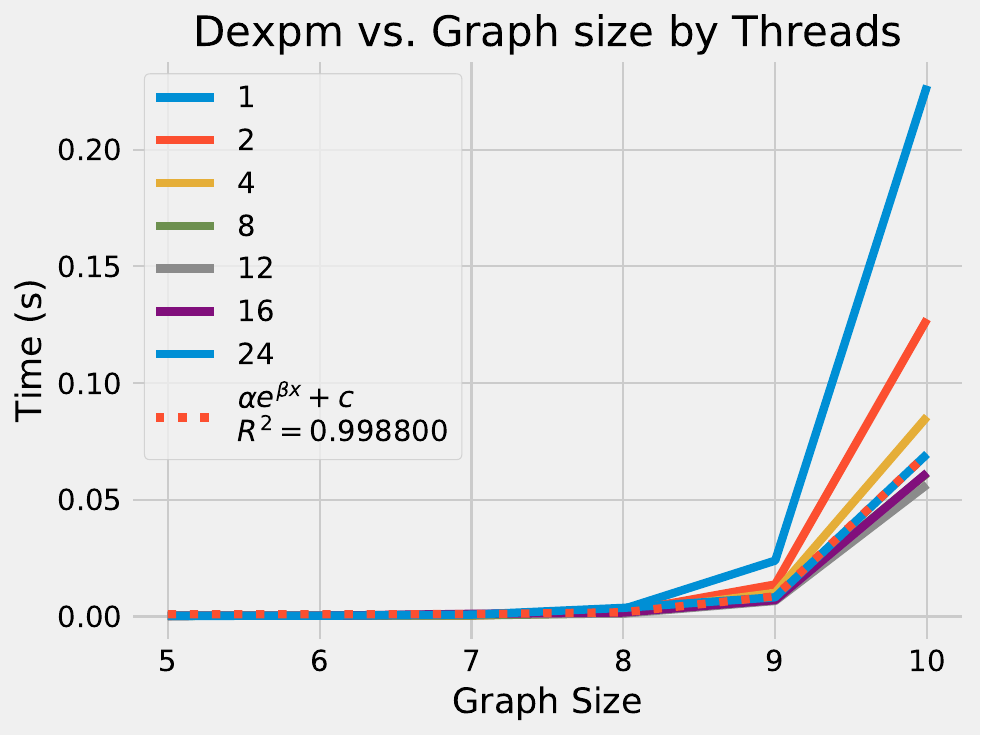} &   \includegraphics[width=0.5\textwidth]{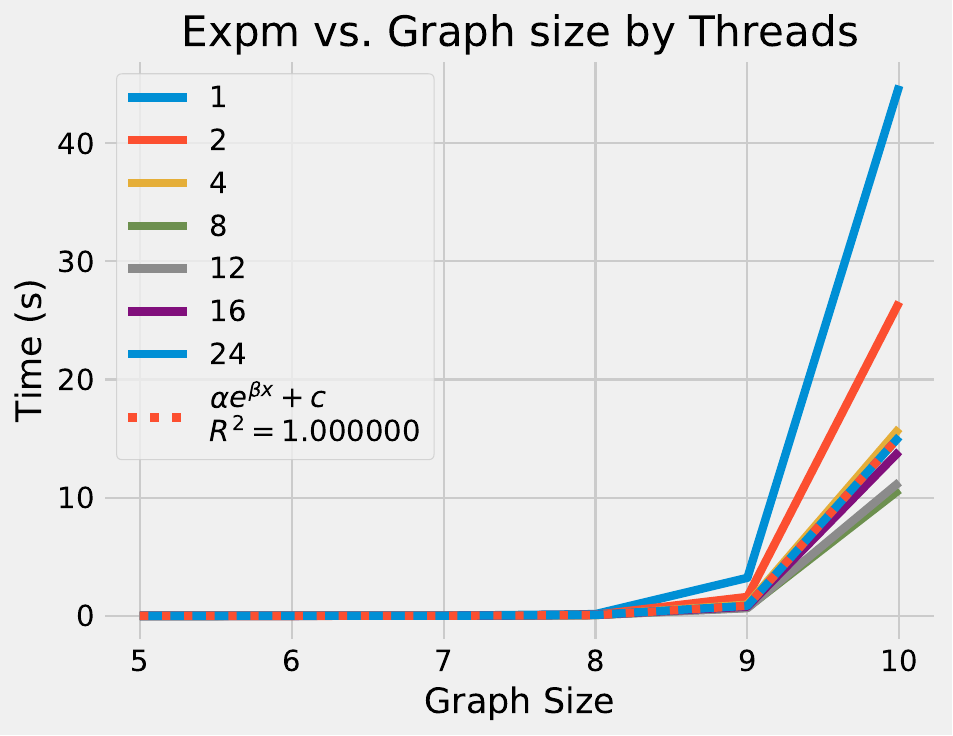} \\
		(c) Diagonal $\hat{C}$ Matrix Exponential & (d) $\hat{B}$ Matrix Exponential \\[6pt]
	\end{tabular}
	\caption{Threading performance}
\end{figure}
\begin{figure}\label{fig:Speedup}
	\centering
	\includegraphics[width=0.5\textwidth]{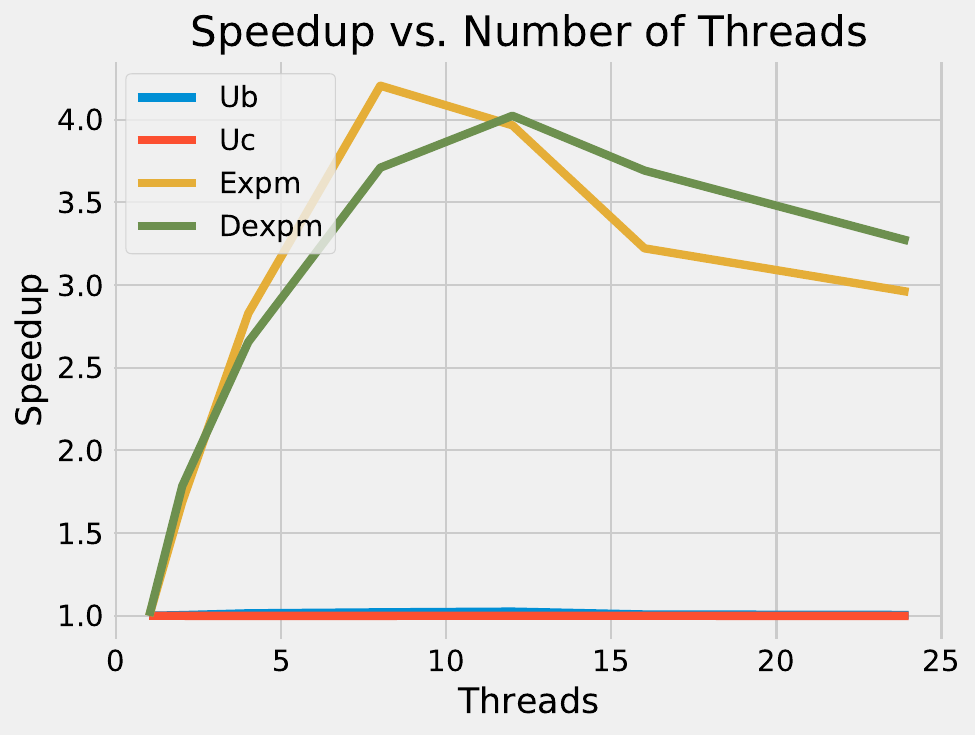}
	\caption{Relative speed-up for critical simulation tasks}
\end{figure}
\begin{figure}\label{fig:timing}
	\centering
	\includegraphics[width=0.5\textwidth]{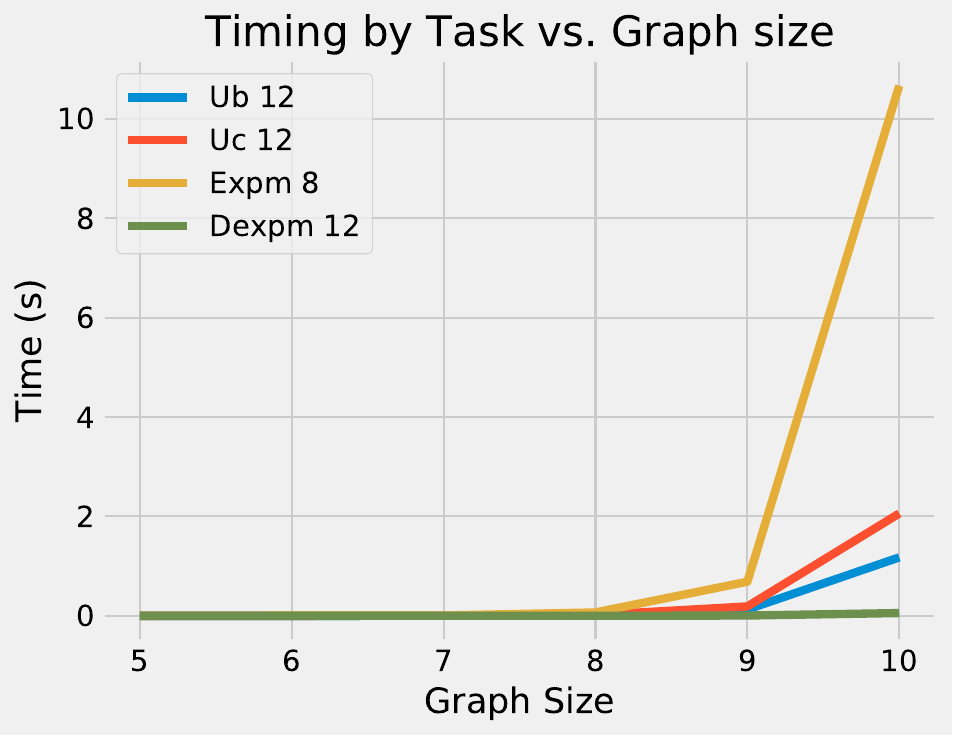}
	\caption{Timing for simulation tasks with optimal threading}
\end{figure}
We plot timing data for four major components of the simulation in Figure \ref{fig:Threading} and relative speed-up for each task in Figure \ref{fig:Speedup}. We compare the simulation time of each task with optimal threading in Figure \ref{fig:timing}, we see that the full matrix exponential dominates the computational workload.
Although the work required by all tasks scales exponentially with graph-size the most intensive task is the computation of $\hat{U}_B$. This justifies our effort to make this operation efficient and lack of threading for the simpler tasks. Generating $\hat{C}$ and $\hat{B}$ occurs once in a trial whereas the $\hat{U}_C$ and $\hat{U}_B$ operations are performed thousands of times.
\chapter{Conclusion}
We explore the Quantum Approximate Optimisation Algorithm (QAOA) \cite{farhi_quantum_2014-1}, an efficient combinatorial optimisation algorithm designed to approach NP-Complete problems. The low gate-depth, hybrid nature and generality makes the QAOA a strong candidate for near term implementation and practical use. We explore the algorithm in a novel manner independent of physical implementation through bespoke high-performance simulation testing up to $22$ qubits. Our Quantum Approximate Optimisation package provides class-leading flexibility allowing simulation of both the QAOA and derivative algorithms \cite{marsh_quantum_2018} while providing efficient scaling from desktop to cluster scale. We investigate the use of eight classical optimisation schemes seeking insight into any differences in performance offered by different optimisation paradigms finding a general trade-off between correctness and efficiency when considering global versus local methods.\newline
We encode the problem of graph similarity to the QAOA by considering the edge-overlap between two unlabelled vertices for both directed and undirected graphs; a similarity measure not yet considered classically. We provide a novel encoding scheme saving $\mathcal{O}(V)$ qubits at the cost of considering infeasible solutions. The run-time cost of this space saving extends through both the quantum and classical portions of the algorithm; the former due to the increased gate-depth and circuit complexity required to encode the problem and the latter due to the non-solutions. The QAOA is able to optimise out of an initial local minima including the infeasible solutions but requires significantly more optimisation to find good solutions. We investigate an initial starting state of superposition between all qubits testing the resilience of the QAOA against a significant number of infeasible and undesirable solutions.\newline
Run-time restrictions placed on high-performance compute resources make evaluating larger problem instances problematic for continuous simulation of the QAOA; implementation of a check-point system or multi-start paradigm would allow for exploration of larger problem instances. Investigation of a wider variety of hard combinatorial optimisation problems may reveal additional avenues to exploit the sub-structure of certain problems and further discriminate between classical optimisation schemes. Moreover the formulation of a parameter optimisation schemes targeted directly for the QAOA may lead to significant performance benefits. Finally, tighter integration into pre-existing quantum computing resources (both simulators and physical implementations) will facilitate a more complete investigation of this promising and exotic algorithm.

\newpage
\printbibliography
\appendix
\chapter{Original Honours Proposal}
\label{ap:Proposal}
\newcommand{\namelistlabel}[1]{\mbox{#1}\hfil}
\newenvironment{namelist}[1]{%1
\begin{list}{}
    {
        \let\makelabel\namelistlabel
        \settowidth{\labelwidth}{#1}
        \setlength{\leftmargin}{1.1\labelwidth}
    }
  }{%1
\end{list}}
\begin{namelist}{xxxxxxxxxxxx}
\item[{\bf Title:}]
	Quantum Graph Similarity and Applications
\item[{\bf Author:}]
	Nicholas Pritchard
\item[{\bf Supervisor:}]
	Professor Jingbo Wang and Professor Amitava Datta
\item[{\bf Degree:}]
	BSc(Hons.)
\end{namelist}
\section*{Background}
The goal of quantum computing is the exploit the complexity of quantum systems for useful computation. Such a motivation arises from the fact that despite decades of research modelling quantum systems in classical computers has alluded the scientific community \autocite{boixo_characterizing_2016}. The development of near-term physical hardware \autocite{kelly_preview_2018}, \autocite{sete_functional_2016} has combined with a surge of interest to apply quantum computing to numerous fields in computation \autocite{chailloux_efficient_2017, farhi_quantum_2014, biamonte_quantum_2018, ciliberto_quantum_2018, cleve_quantum_1998}. Graph similarity and graph-isomorphism are long-standing difficulties in computer science. Many useful formulations of graph similarity exist such as compound matching in chemistry \autocite{han_topological_2015, hattori_development_2003, raymond_rascal:_2002}, machine vision \autocite{lades_distortion_1993} and web-search \autocite{brin_anatomy_1998}.  This problem has no tractable exact formulation for graphs with unknown node correspondence and as such approximate solutions are considered industry standard. More specifically we define whole-graph similarity
\begin{definition}{Whole Graph Similarity:}
Given two graphs $G_1(v_1,e_1)$ and $G_2(v_2,e_2)$ with possibly different numbers of vertices and edges, find an algorithm which returns a measure of similarity $S | S \in [0,1]$. Furthermore:
\begin{enumerate}
\item $S(G_1, G_1) = 1$
\item $S(G_1, G_2) = S(G_2, G_1)$
\end{enumerate}
\label{ref:similarity_proposal}
\end{definition}
\section*{Aim}
To investigate a quantum algorithmic approach to graph similarity and its applications this project will examine graph similarity by applying the recently proposed 'Quantum Approximate Optimisation Algorithm' (QAOA) \autocite{farhi_quantum_2014}. The QAOA can be formulated to reduce a combinatorial optimisation problem to a parameter search on around two variables. This project primarily aims to investigate if such an approach leads to any improvements in speed, accuracy or robustness over classical methods.\newline
A secondary objective is to extend the generated model of graph similarity into a real-world contextual use such as object tracking or common sub-graph matching for example. Currently it is unknown whether a quantum advantage will yield any benefits in speed, accuracy or robustness over classical counter-parts and hence makes a suitable topic for research.
\section*{Method}
There will be a large amount of theoretical work in the development and validation of a quantum or hybrid quantum/classical algorithms to tackle the problem of graph similarity. Testing will require a series of standardised sources as well as comparison results or implementations of classical algorithms. Due to the compute-heavy nature of simulating quantum systems it is likely the Magnus supercomputer at the Pawsey Super-computing Centre will need to be utilised. \textit{Python}, \textit{C/C++} and possibly \textit{Fortran} will be used to simulate various approaches and to create visualisations of resulting circuit designs.
\subsection*{Status}
This project does not follow from previous work commencing at the start of the 2018 academic year. This research is a collaboration between the \textit{Quantum dynamics and computation} research group and the department of \textit{Computer science and software engineering}.\newline
Currently, a general understanding of quantum computing and potential object detection methods are being investigated in addition to classical object-tracking frameworks with the aim of finding a suitable starting point to apply quantum methods. An investigation into state-of-the-art quantum simulations yields a number of possible methods and frameworks \autocite{smelyanskiy_qhipster:_2016, pednault_breaking_2017, boixo_simulation_2017, chen_64-qubit_2018, chen_classical_2018}. Preliminary-work on simulating the Quantum Approximate Optimisation Algorithm (QAOA) locally is underway.
\section*{Software and Hardware Requirements}
Costs are expected to be negligible as access to required software packages are openly available on-line or through the University of Western Australia. The majority of testing is to be run on personal machines. Use of the Pawsey Supercomputing Centre is available if needed based on an agreement with the University. 
\chapter{Tail Complexity}
\label{ap:Tail}
\begin{table}[htbp]\label{tab:stateSize}
	\begin{center}
		\begin{tabular}{|c|c|c|c|c|c|}
			\hline 
			Vertices ($V$) & $V!$ & Qubits ($q$) & $2^{q}$ & Difference($2^q - V!$) & $2^q/V!$\\ 
			\hline 
			$2$ & $2$  & $1$ & $2$ & $0$ & $0$ \\ 
			\hline 
			$4$ & $6$ & $5$ & $32$ & $26$ & $4.3$ \\ 
			\hline 
			$8$ & $40320$ & $16$ & $65536$ & $25216$ & $0.65$ \\ 
			\hline 
			$10$ & $3628800$ & $22$ & $4194304$ & $565504$ & $0.15$ \\ 
			\hline 
			$12$ & $479001600$ & $29$ & $536870912$ & $57869312$ & $0.12$ \\ 
			\hline
			$15$ & $1.30767E12$ & $41$ & $2.19902E12$ & $8.91349E11$	& $0.68$\\
			\hline
			$22$ & $1.124E21$ & $70$ &$1.18059E21$ & $5.65909E19$ & $0.05$\\
			\hline 
		\end{tabular}
		\caption{Graph-size compared to qubit state-space}
	\end{center}
\end{table}
\section{Series Expansion at $n = \infty$}
\begin{equation}
2^{log_2(n!)}(\frac{\sqrt{\frac{1}{n}}}{\sqrt{2\pi}} + \mathcal{O}((\frac{1}{n})^{3/2}))exp((1 - log(n))n + \mathcal{O}((\frac{1}{n})^2)) - 1
\end{equation}
\chapter{Description of Qolab}
\label{ap:Qolab}
Qolab is a near problem agnostic simulation of the QAOA \cite{farhi_quantum_2014-1} with extended constraint ability \cite{marsh_quantum_2018}. Implementation using Intel's Math Kernel Library \cite{noauthor_intelr_nodate} facilitates maximal desktop and single node performance. The open-source nlopt optimisation suite  \cite{johnson_nlopt_2011} allows for a variety of optimisation algorithms to be tested simply. Qolab supports cluster execution allowing exacting state-space decomposition. Code is available at
\begin{center}
	\url{https://bitbucket.org/qaoa_uwa/graphsimilarity/src/master/}
\end{center}
Currently, Qolab supports the following arguments:
\begin{itemize}
	\item Arbitrary number of qubits
	\item Arbitrary cost function 
	\item Mixing masks on the $\hat{U}_B$ operator with an alternate function evaluation path as per Algorithm \ref{alg:QAOAn} proposed by Marsh and Wang \cite{marsh_quantum_2018}.
	\item Arbitary trotterisation depth ($p$ variable)
	\item Confidence interval based sampling scaling
	\item Arbitrary sampling
	\item Command line support for nlopt optimisation selection
	\item QAOA $(\vec{\gamma}, \vec{\beta})$ argument optimisation tolerances
	\item Variable function output tolerance selection
	\item Full-desktop and cluster implementations 
\end{itemize}
\chapter{Pseudocode}
\label{ap:Pseudo}
\section{Lehmer Code Permutation}
\begin{algorithm}
	\caption{Factoradic Permutation Generator}\label{alg:Perm}
	\begin{algorithmic}[1]
		\Procedure{k\_Perm}{n, k}
			\State facts$[] \gets \emptyset$
			\State items$[] \gets {0,\dots,n}$
			\State out$[] \gets \emptyset$
			\State nnz $\gets 0$
			\State size $\gets n$
			\While{size $> 0$}
				\State $f \gets $factorial(size$ - 1$)
				\State $i \gets k/f$
				\State $x \gets $items$[k]$
				\State $k \gets k $mod$f$
				\State out[nnz] $\gets x$ 
				\For{j = 0 to n - 1}
					\State items[j] $\gets$ items[j+1]
				\EndFor
				\State size $\gets$ size - 1
				\State nnz $\gets$ nnz + 1
			\EndWhile
			\State \textbf{Return} out
		\EndProcedure
	\end{algorithmic}
\end{algorithm}
This procedure returns the k-th permutation of natural numbers $[0,\dots,n)$
\section{NPO QAOA}
\begin{algorithm}
	\caption{NPO QAOA Overview}\label{alg:QAOAn}
	\begin{algorithmic}[1]
		\Procedure{NPO\_QAOA\_Core}{numQubits, P, optimisationMethod, C(x), Mask(x)}
			\State $\hat{U}_C \gets $\textbf{genUC}(C(x))
			\State $\hat{U}_B \gets $\textbf{genUB}(numQubits, Mask(x))
			\State $\vec{\gamma}, \vec{\beta} \gets $\textbf{initialParameters()}
			\While{terminateTest()}
				\State $\ket{\psi} \gets $\textbf{initialState}
				\For{$i = 0$ to $p$}
					\State $\ket{\psi} \gets \hat{U}_B(\beta_i)\ket{\psi}$
					\State $\ket{\psi} \gets \hat{U}_C(\gamma_i)\ket{\psi}$
				\EndFor
				\State $\ket{\psi} \gets \hat{U}_B(\beta_{p+1})\ket{\psi}$
				\State $F_p \gets \textbf{Measure}(\psi)$
				\State $\vec{\gamma}, \vec{\beta} \gets $ \textbf{updateParameters($F_{p}$)}
			\EndWhile
		\EndProcedure
		\State \textbf{Report}()
	\end{algorithmic}
\end{algorithm}
Note the subtle differences to the QAOA in Algorithm \ref{alg:QAOAn}, an additional $\hat{U}_B$ operation is applied at the start and a validation mask is passed to the $\hat{B}$ generation routine.
\chapter{Proofs}
\label{ap:Proofs}
We rely on the Perron-Frobenius theroem which states
\begin{theorem}
	A symmetric elements of only real, non-negative entries has all real non-negative eigenvalues. Furthermore there exists a Perron-Root $r$ for which all other eigenvalues $\lambda \leq r$.
\end{theorem}
\begin{proof}[We show $\lambda_{min, max} = \pm q$]
Let $B$ be our $2^q \times 2^q$ hypercube adjacency matrix constructed according to
\begin{equation}
\hat{B} = \sum_{i=1}^{n}\sigma_i^x
\end{equation}
We note that $B$ is Hermitian and therefore symmetric and by definition contains real values. We shall show that there exists and eigenvalue $\lambda$ of $B$ such that for any vector $\mathbf{v} \in \mathbb{R}^n$ we have 
\begin{equation}
\mathbf{v} \cdot B \mathbf{v} \leq \lambda\norm{\mathbf{v}}^2
\end{equation}
For a real symmetric matrix, there exist eigenvectors $\mathbf{v_1, v_2,\dots v_n}$ corresponding to $\lambda_1,\lambda_2,\dots\lambda_n$ such that
\begin{equation}
E = {\mathbf{v_1, v_2 \dots v_n}}
\end{equation}
forms an orthonormal basis of $\mathbb{R}^n$. Equivalently, every real symmetric matrix is diagonalisable by an orthogonal matrix. If this were not the case, this operator would be impossible to implement in a quantum computer since all operations must be unitary in nature.\newline
Let $\mathbf{v}$ be any vector in $\mathbb{R}^n$\newline
Since $E$ is a basis of $\mathbb{R}^n$ we can write
\begin{equation}
\mathbf{v} = c_1\mathbf{v_1} + c_2\mathbf{v_2} + \dots + c_n\mathbf{v_n}
\end{equation}
Where $c_1\dots c_n \in \mathbb{R}$
We then calculate $B\mathbf{v}$ as
\begin{eqnarray}
B\mathbf{v} = B(c_1\mathbf{v_1}+\dots+c_n\mathbf{v_n}\\
= c_1B\mathbf{v_1}+\dots+c_n\mathbf{v_n}\\
= c_1\lambda_1\mathbf{v_1}+\dots+c_n\lambda_n\mathbf{v_n}
\end{eqnarray}
Knowing $B\mathbf{v_i} = \lambda_i\mathbf{v_i}$ for $i=1,\dots,n$
We can then apply another $\mathbf{v}$
\begin{eqnarray}
\mathbf{v}\cdot B\mathbf{v} = (c_1\mathbf{v_1}+\dots+c_n\mathbf{v_n})\cdot(c_1\lambda_1\mathbf{v_1}+\dots+c_n\lambda_n\mathbf{v_n})\\
= c_1^2\lambda_1+\dots+c_n^2\lambda_n
\end{eqnarray}
Using the facto that $E$ is an orthonormal basis of $\mathcal{R}^3$\newline
Since $\lambda$ is the largest eigenvalue of $B$ we show
\begin{eqnarray}
\mathbf{v}\cdot B\mathbf{v} = c_1^2\lambda_1+\dots+c_n^2\lambda_n\\
\leq c_1^2\lambda+\dots+c_n^2\lambda\\
= \lambda(c_2^2+\dots+c_n^2)\\
= \lambda\norm{\mathbf{v}}^2.
\end{eqnarray}
We finally make the observation that for any $\hat{B}$ according to Equation \ref{eq:B} will have at most $q$ elements in each row. When considering a positive $\mathbf{v}, \lambda = q$ and for a negative $\mathbf{v}, \lambda = -q$
\end{proof}
\chapter{Data}\label{ap:Data}
All source data files, plots, csv aggregates and the code used to generate them are available at 
\begin{center}
	\url{https://bitbucket.org/qaoa_uwa/results/src/master/}
\end{center}
We present the remaining plots for all optimisation methods considered.`
\begin{center}
	\begin{figure}
		\centering
		\begin{tabular}{cc}
			\includegraphics[width=0.5\textwidth]{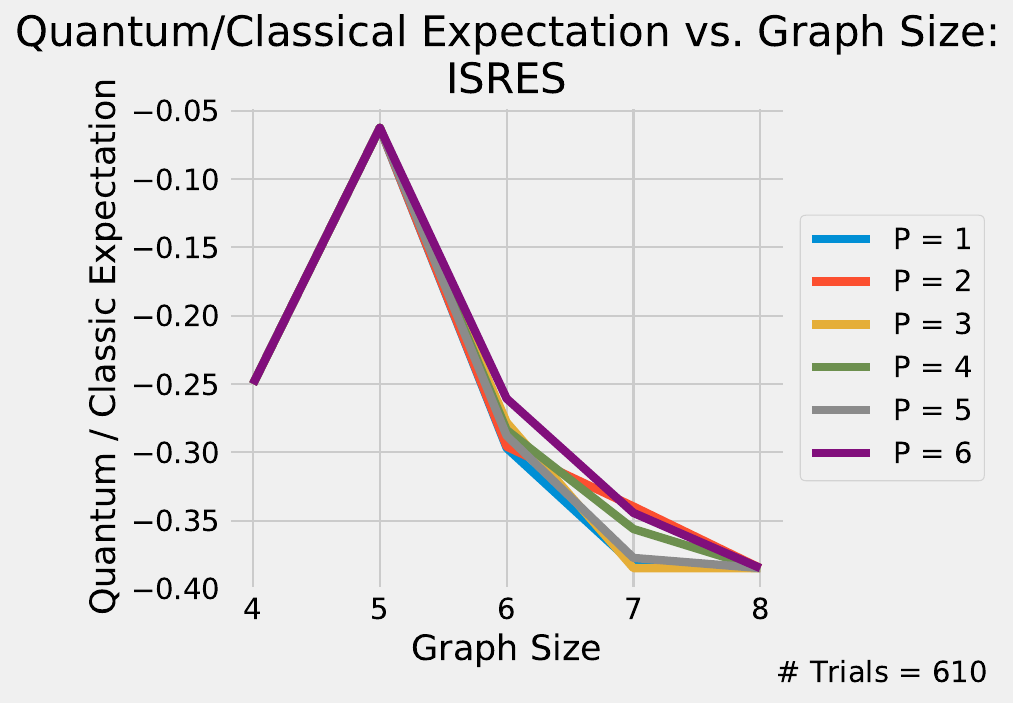} &   \includegraphics[width=0.5\textwidth]{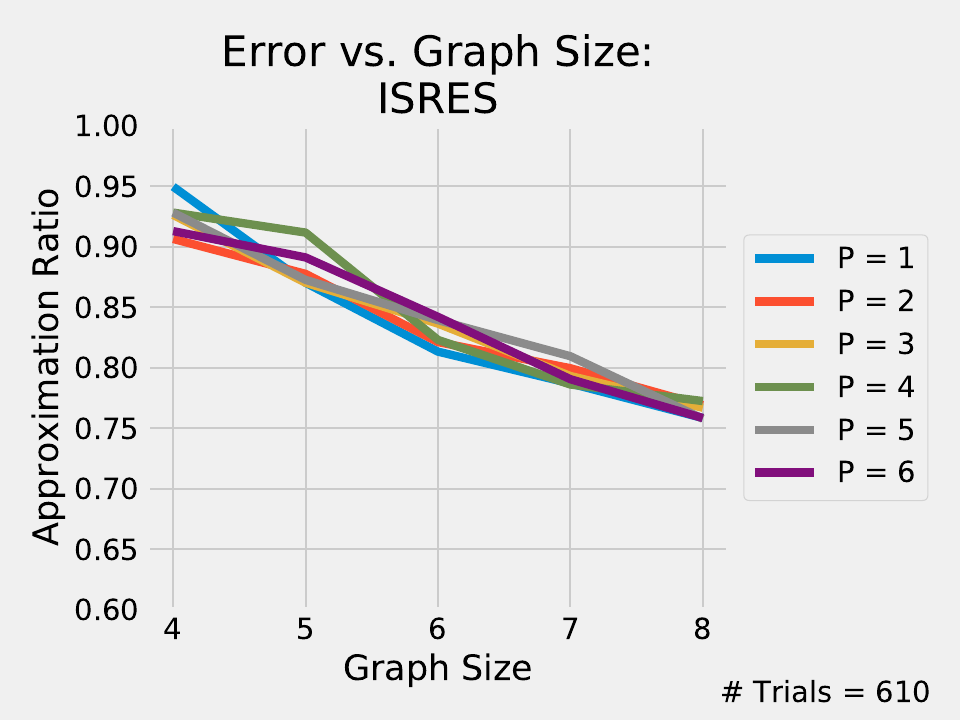} \\
			(a) Expectation value comparison & (b) Solution error \\[6pt]
			\includegraphics[width=0.5\textwidth]{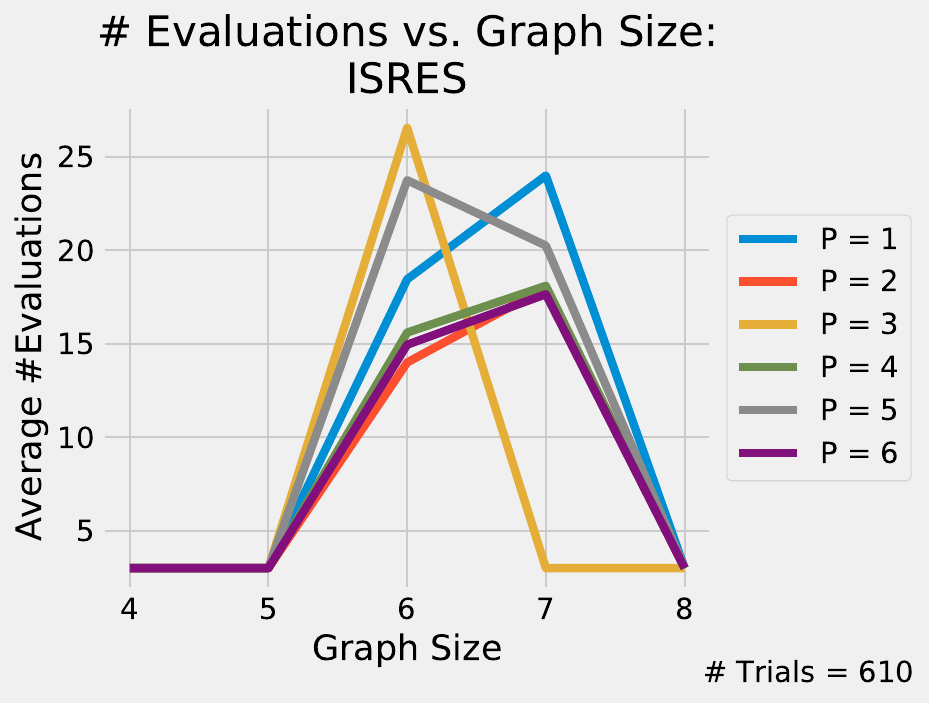} &   \includegraphics[width=0.5\textwidth]{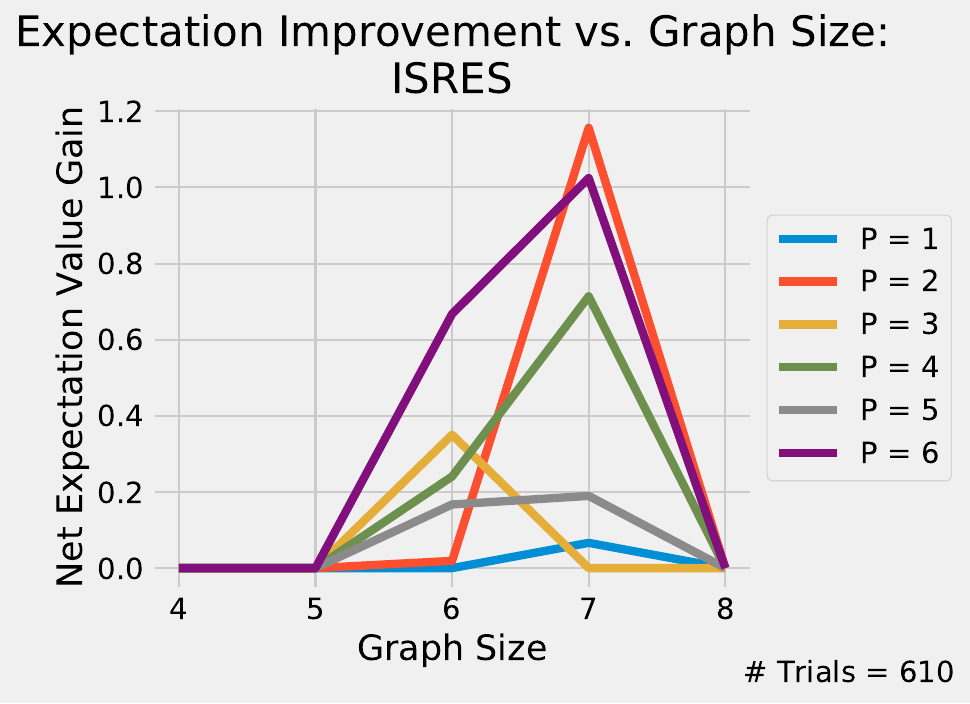} \\
			(c) Function evaluations required & (d) Improvement\\[6pt]
		\end{tabular}
		\label{fig:finalISRES}
		\caption{Final performance metrics for the ISRES algorithm (global) (directed graphs)}
	\end{figure}
\end{center}
\begin{center}
	\begin{figure}
		\centering
		\begin{tabular}{cc}
			\includegraphics[width=0.5\textwidth]{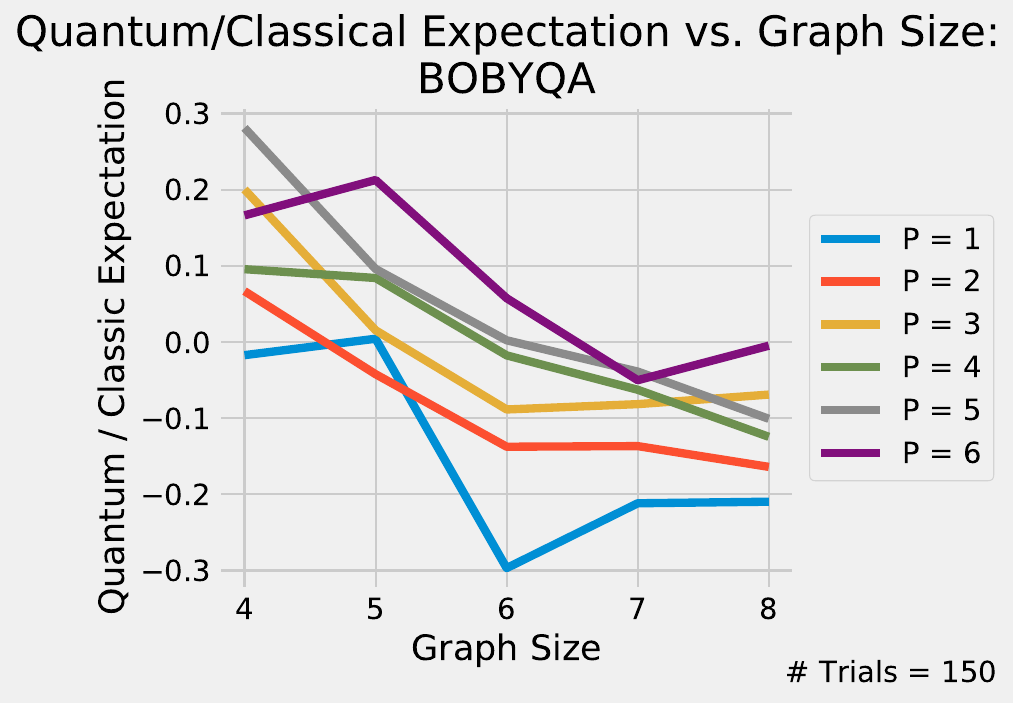} &   \includegraphics[width=0.5\textwidth]{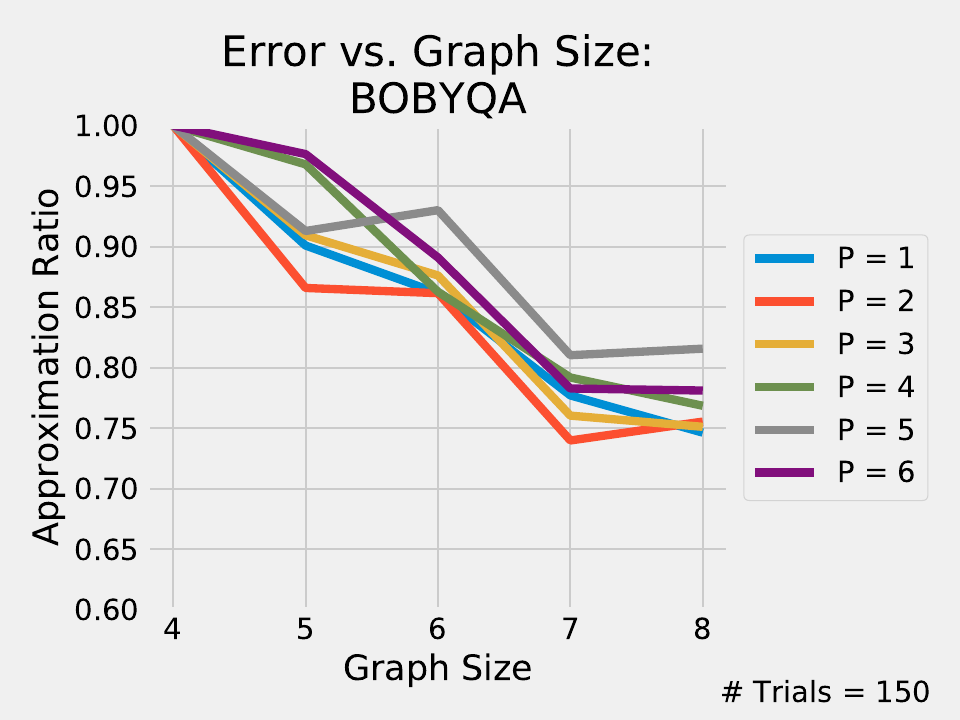} \\
			(a) Expectation value comparison & (b) Solution error \\[6pt]
			\includegraphics[width=0.5\textwidth]{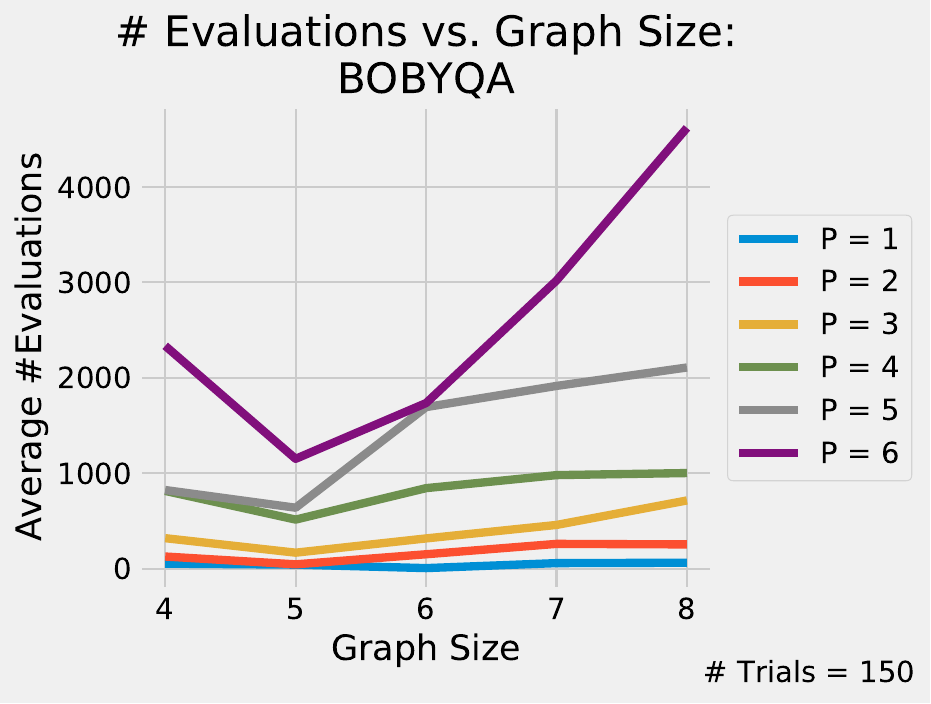} &   \includegraphics[width=0.5\textwidth]{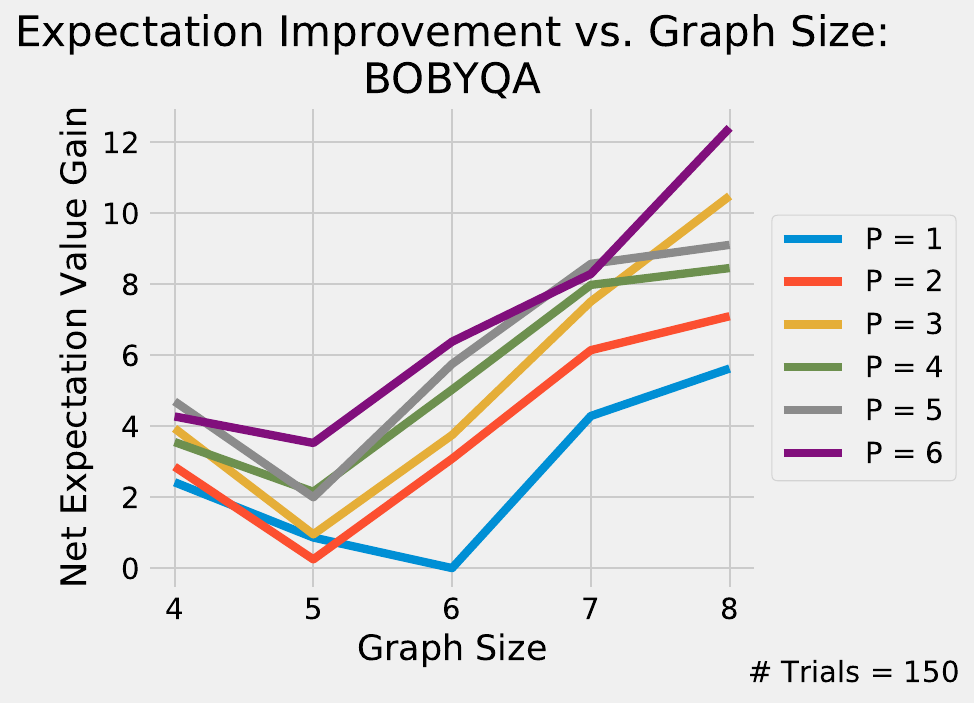} \\
			(c) Function evaluations required & (d) Improvement\\[6pt]
		\end{tabular}
		\label{fig:finalBOBYQA}
		\caption{Final performance metrics for the BOBYQA algorithm (local) (directed graphs)}
	\end{figure}
\end{center}
\begin{center}
	\begin{figure}
		\centering
		\begin{tabular}{cc}
			\includegraphics[width=0.5\textwidth]{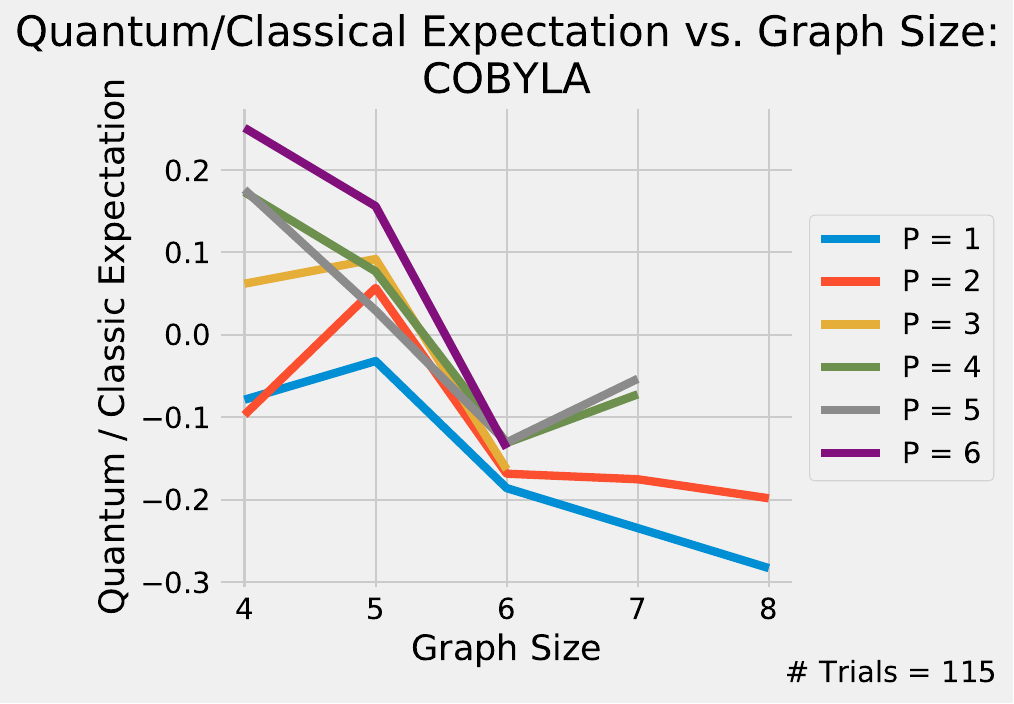} &   \includegraphics[width=0.5\textwidth]{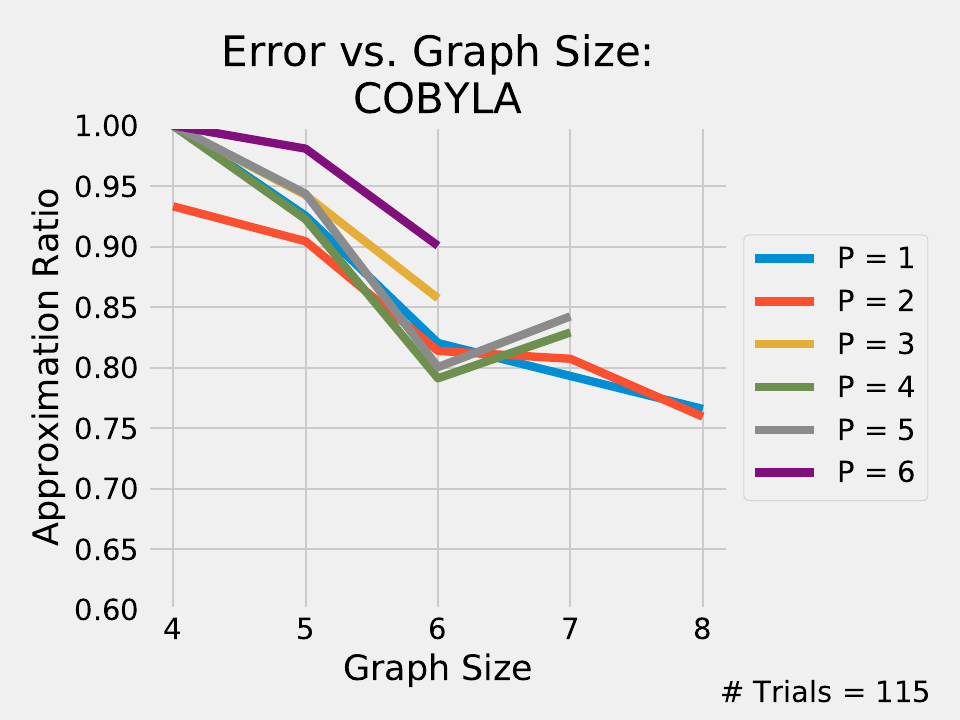} \\
			(a) Expectation value comparison & (b) Solution error \\[6pt]
			\includegraphics[width=0.5\textwidth]{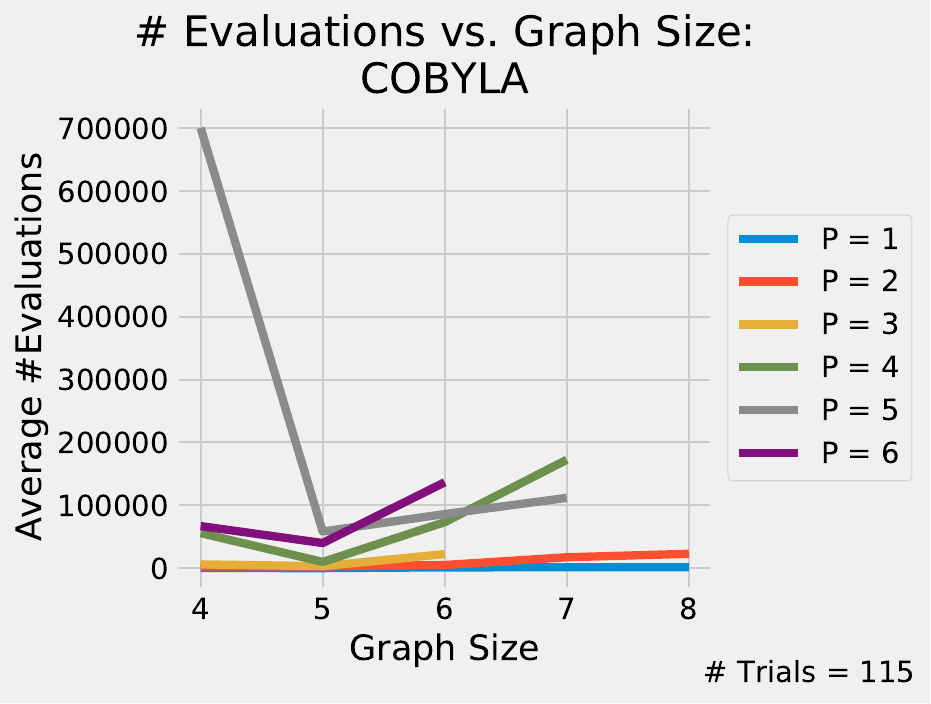} &   \includegraphics[width=0.5\textwidth]{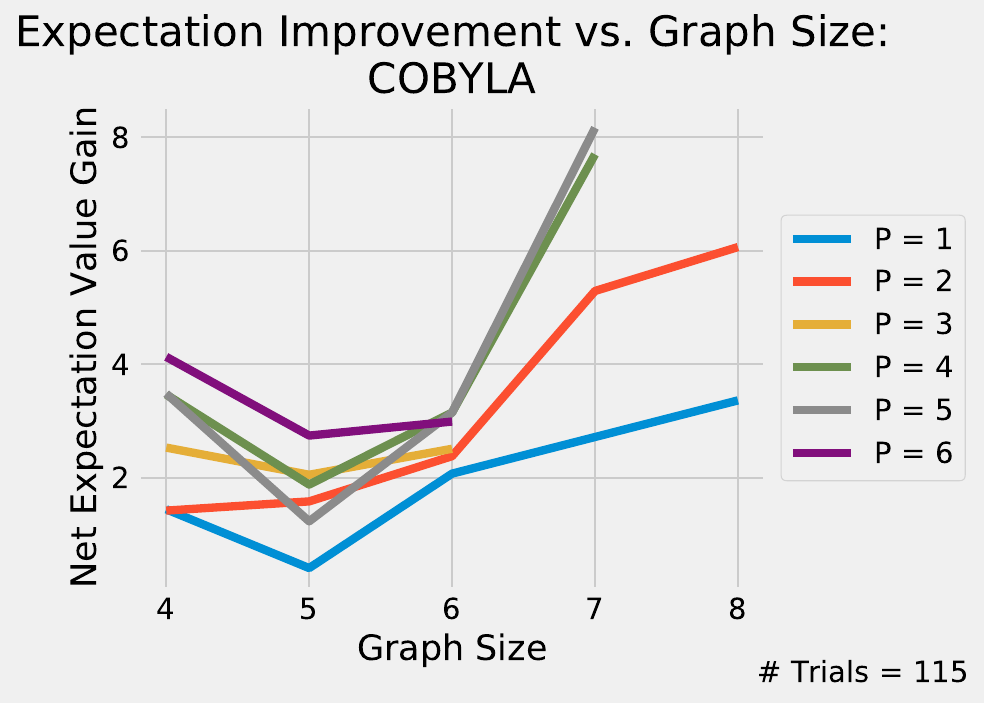} \\
			(c) Function evaluations required & (d) Improvement\\[6pt]
		\end{tabular}
		\label{fig:finalCOBYLA}
		\caption{Final performance metrics for the COBYLA algorithm (directed graphs)}
	\end{figure}
\end{center}
\begin{center}
	\begin{figure}
		\centering
		\begin{tabular}{cc}
			\includegraphics[width=0.5\textwidth]{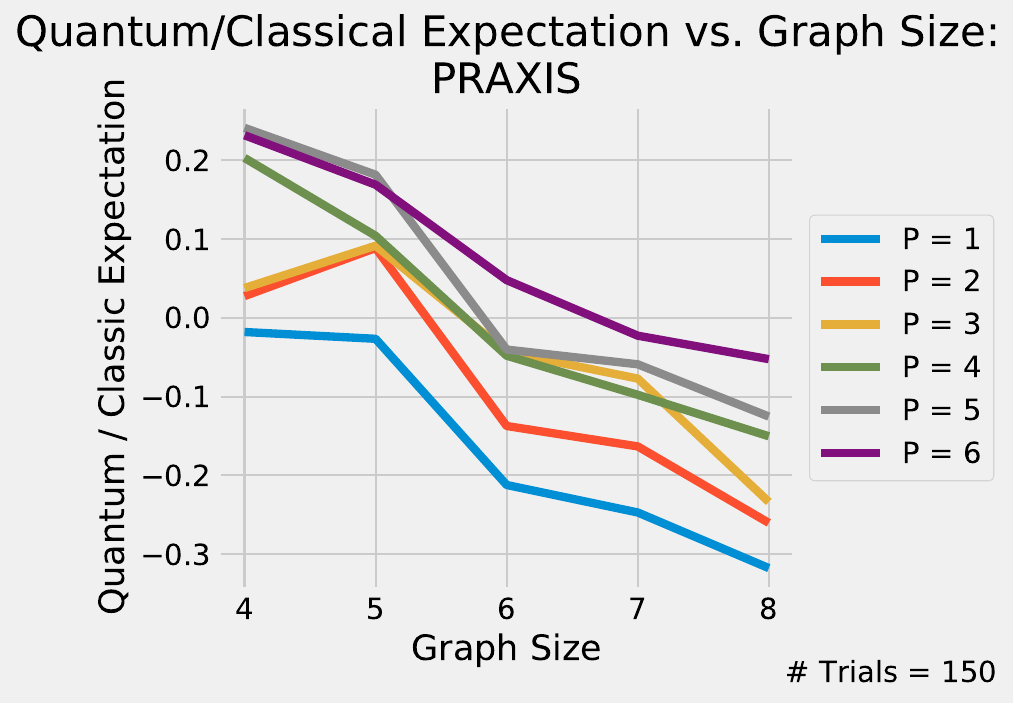} &   \includegraphics[width=0.5\textwidth]{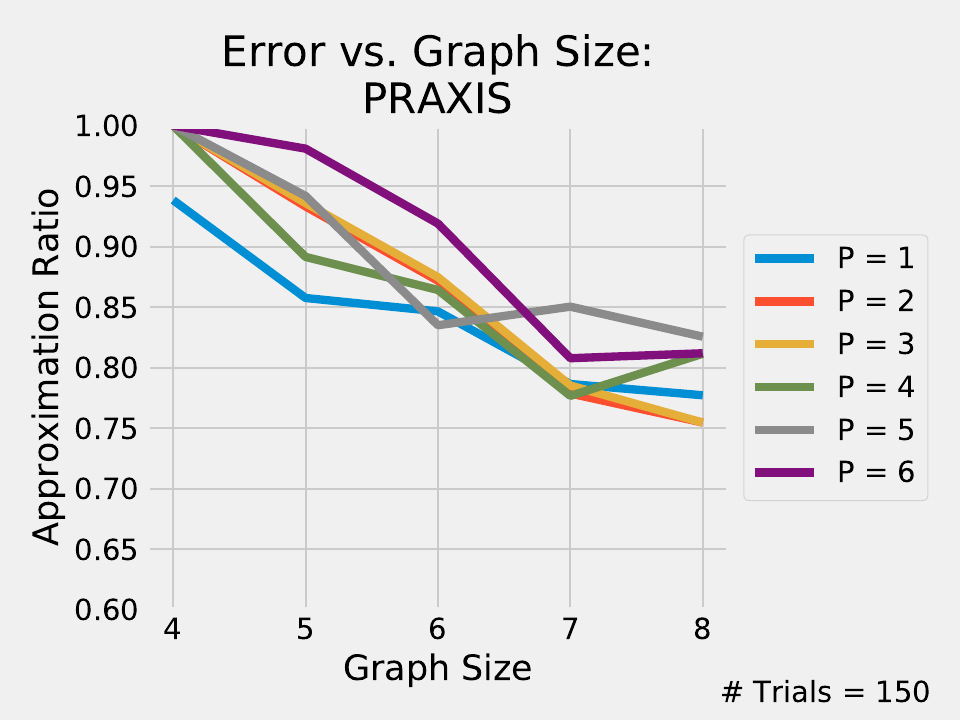} \\
			(a) Expectation value comparison & (b) Solution error \\[6pt]
			\includegraphics[width=0.5\textwidth]{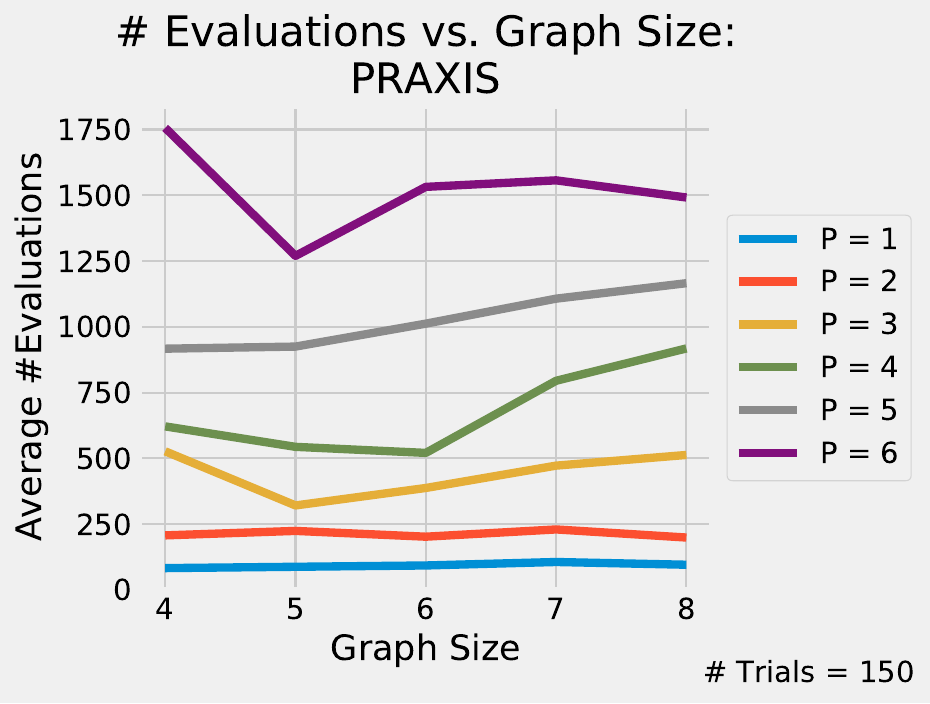} &   \includegraphics[width=0.5\textwidth]{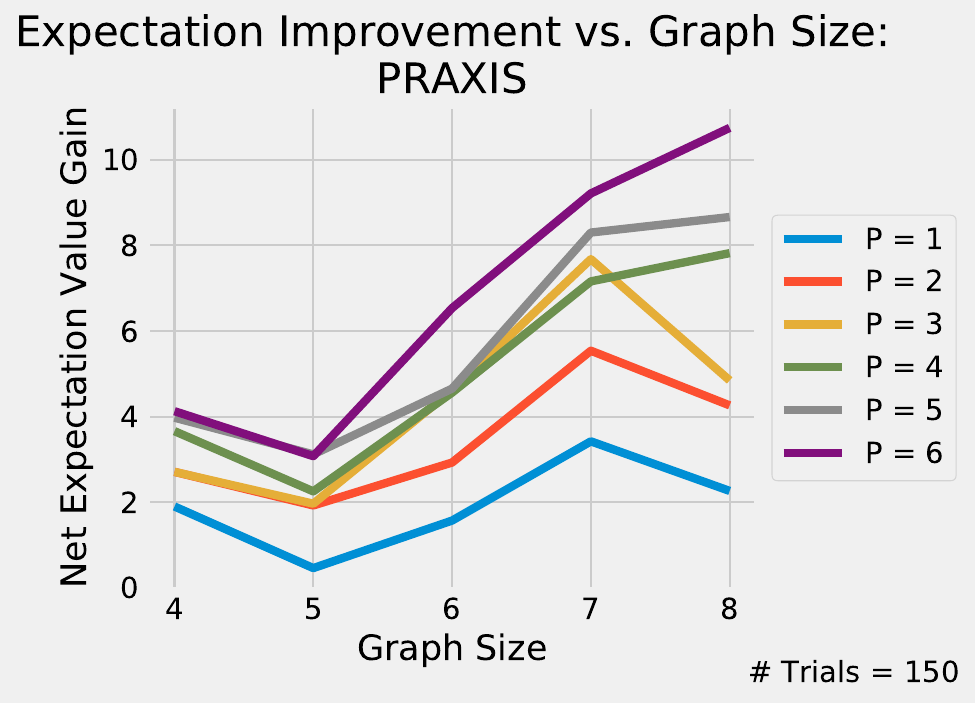} \\
			(c) Function evaluations required & (d) Improvement\\[6pt]
		\end{tabular}
		\label{fig:finalPRAXIS}
		\caption{Final performance metrics for the PRAXIS algorithm (local) (directed graphs)}
	\end{figure}
\end{center}
\end{document}